\def\draft{0}
\documentclass[11pt]{article}

\usepackage{amsfonts}
\usepackage{amsmath}
\usepackage{amssymb}
\usepackage{amsthm}
\usepackage{mathabx}
\usepackage{thmtools, thm-restate}
\usepackage{bbm}
\usepackage{bm}
\usepackage{dsfont}
\usepackage{fullpage}
\usepackage{tikz}
\usepackage[most]{tcolorbox}
\usepackage{xcolor}
\usepackage{sidenotes}

\newcommand*{\myfont}{\fontfamily{bch}\selectfont}
\DeclareTextFontCommand{\textmyfont}{\myfont}

\usepackage[linesnumbered,ruled,vlined]{algorithm2e} 
\SetKwComment{Comment}{/* }{ */}

\SetCommentSty{mycommfont}

\usepackage{color}
\usepackage[pdfstartview=FitH,pdfpagemode=UseNone,colorlinks=true,citecolor=blue,linkcolor=blue,backref=page]{hyperref}
\usepackage[nameinlink]{cleveref}
\usepackage{url}            

\newtheorem{theorem}{Theorem}[subsection]
\newtheorem{corollary}[theorem]{Corollary}
\newtheorem{lemma}[theorem]{Lemma}
\newtheorem{observation}[theorem]{Observation}

\newtheorem{definition}[theorem]{Definition}
\newtheorem{claim}[theorem]{Claim}
\newtheorem{fact}[theorem]{Fact}
\newtheorem{remark}[theorem]{Remark}

\newtheorem{question}{Question}

\newcommand{\prob}[2]{\mathop{\mathrm{Pr}}_{#1}[#2]}

\newcommand{\poly}{\mathop{\mathrm{poly}}}

\newcommand{\F}{\mathbb{F}}
\newcommand{\R}{\mathbb{R}}
\newcommand{\mc}[1]{\mathcal{#1}}

\newcommand{\Boo}{\{0,1 \}}

\newcommand{\bigO}{\mathcal{O}}

\newcommand{\paren}[1]{\left( #1 \right)}
\newcommand{\brac}[1]{\left[ #1 \right]}
\newcommand{\set}[1]{\left\{ #1 \right\}}

\newcommand{\setcond}[2]{\left\{ #1 \;\middle\vert\; #2 \right\}}

\newcommand{\spars}{\textnormal{spars}}
\newcommand{\innerprod}[1]{\left\langle#1\right\rangle}

\DeclareMathOperator*{\E}{\mathbb{E}}

\newcommand{\cF}{\mathcal{F}}
\newcommand{\cP}{\mathcal{P}}
\newcommand{\Z}{\mathbb{Z}}

\definecolor{thmcolor}{RGB}{235, 235, 235}
\definecolor{citecolor}{RGB}{1, 210, 56}

\newtcolorbox{algobox}{colback=lightgray!5!white,colframe=lightgray!75!black}
\newtcolorbox{thmbox}{colback=thmcolor!5!white,colframe=black!75!black}

\usepackage{setspace}
\usepackage{enumerate}
\usepackage[parfill]{parskip}
\usepackage{changepage}
\usepackage{framed}

\allowdisplaybreaks
\hypersetup{
	colorlinks,
	linkcolor={blue},
	citecolor={citecolor},
	urlcolor={blue}
}

\ifnum\draft=1
\newcommand{\anote}[1]{{\color{brown} [Amik: #1]}}
\newcommand{\mnote}[1]{{\color{red} [Madhu: #1]}}
\newcommand{\mpnote}[1]{{\color{pink} [Manaswi: #1]}}
\newcommand{\pnote}[1]{{\color{blue} [Prashanth: #1]}}
\newcommand{\snote}[1]{{\color{green} [Srikanth: #1]}}
\else  
\newcommand{\anote}[1]{}
\newcommand{\mnote}[1]{}
\newcommand{\mpnote}[1]{}
\newcommand{\pnote}[1]{}
\newcommand{\snote}[1]{}
\fi

\def\anon{0}
\date{\today}

\begin{document} 
	\title{Low Degree Local Correction Over the Boolean Cube}

    \if\anon1{}\else{    
    \author{Prashanth Amireddy\thanks{School of Engineering and Applied Sciences, Harvard University, Cambridge, Massachusetts, USA. Supported in part by a Simons Investigator Award and NSF Award CCF 2152413 to Madhu Sudan and a Simons Investigator Award to Salil Vadhan. Email: \texttt{pamireddy@g.harvard.edu}} \and
    Amik Raj Behera\thanks{Department of Computer Science, University of Copenhagen, Denmark. Supported by Srikanth Srinivasan's start-up grant from the University of Copenhagen. The work was begun when the author was a student at the Department of Computer Science, Aarhus University, Denmark and was supported by Srikanth Srinivasan's start-up grant from the Aarhus University. Email: \texttt{ambe@di.ku.dk} } \and
    Manaswi Paraashar\thanks{Department of Mathematical Sciences, University of Copenhagen, Denmark. Supported by the European Union under the Grant Agreement No 101078107, QInteract. Email: \texttt{manaswi.isi@gmail.com} } \and
     Srikanth Srinivasan \thanks{Department of Computer Science, University of Copenhagen, Denmark. Also partially employed by Aarhus University, Denmark. This work was funded by the European Research Council (ERC) under grant agreement no. 101125652 (ALBA).
     Email: \texttt{srsr@di.ku.dk} } \and 
     Madhu Sudan\thanks{School of Engineering and Applied Sciences, Harvard University, Cambridge, Massachusetts, USA. Supported in part by a Simons Investigator Award and NSF Award CCF 2152413. Email: \texttt{madhu@cs.harvard.edu}}}
     }\fi

	\maketitle
        \pagenumbering{arabic}

\begin{abstract}
In this work, we show that the class of multivariate degree-$d$ polynomials mapping $\{0,1\}^{n}$ to any Abelian group $G$ is locally correctable with $\widetilde{\mathcal{O}}_{d}((\log n)^{d})$ queries for up to a fraction of errors approaching half the minimum distance of the underlying code. In particular, this result holds even for polynomials over the reals or the rationals, special cases that were previously not known. Further, we show that they are locally list correctable up to a fraction of errors approaching the minimum distance of the code. These results build on and extend the prior work of \if\anon1{Amireddy, Behera, Paraashar, Srinivasan, and Sudan \cite{ABPSS24-ECCC} (STOC 2024)~}\else{the authors~\cite{ABPSS24-ECCC}~(STOC 2024)~}\fi who considered the case of linear polynomials ($d=1$) and gave analogous results.

Low-degree polynomials over the Boolean cube $\{0,1\}^{n}$ arise naturally in Boolean circuit complexity and learning theory, and our work furthers the study of their coding-theoretic properties. Extending the results of \cite{ABPSS24-ECCC} from linear polynomials to higher-degree polynomials involves several new challenges and handling them gives us further insights into properties of low-degree polynomials over the Boolean cube. For local correction, we construct a set of points in the Boolean cube that lie between two exponentially close parallel hyperplanes and is moreover an interpolating set for degree-$d$ polynomials. To show that the class of degree-$d$ polynomials is list decodable up to the minimum distance, we stitch together results on anti-concentration of low-degree polynomials, the Sunflower lemma, and the Footprint bound for counting common zeroes of polynomials. Analyzing the local list corrector of \cite{ABPSS24-ECCC} for higher degree polynomials involves understanding random restrictions of non-zero degree-$d$ polynomials on a Hamming slice. In particular, we show that a simple random restriction process for reducing the dimension of the Boolean cube is a suitably good sampler for Hamming slices. Thus our exploration unearths several new techniques that are useful in understanding the combinatorial structure of low-degree polynomials over $\{0,1\}^{n}$.
\end{abstract} 
        \newpage 
        
\tableofcontents

        \newpage
        
\section{Introduction}\label{sec:intro}

In this paper, we consider the local correction of low-degree polynomial functions over groups evaluated over $\{0,1\}^n$ and give polylogarithmic query local correctors for every constant degree. This extends and generalizes previous work of \if\anon1{Amireddy, Behera, Paraashar, Srinivasan and Sudan~}\else{the authors}\fi~\cite{ABPSS24-ECCC} who considered and solved the analogous problem for the linear (i.e., $d=1$) case. We define some of the basic terms and review the previous work before describing the challenges in strengthening to higher degrees and the new tools used to overcome them.

\paragraph{Low degree polynomials over groups.} The main objects of interest in this paper are polynomial functions mapping $\{0,1\}^n$ to an Abelian group $G$. Here a function $f$ is a polynomial of degree at most $d$ if it can be expressed as $\sum_{S \subseteq [n]: |S| \leq d} c_S \prod_{i\in S} x_i$, where the product is over the integers and the coefficients $c_S$ come from the Abelian group $G$. We denote the space of polynomial functions of degree at most $d$ by $\mc{P}_{d}(\Boo^{n}, G)$ (which we compress to $\mc{P}_d$ when $G$ and $n$ are known). The standard proof of the Ore-DeMillo-Lipton-Schwartz-Zippel lemma naturally extends to polynomials over groups. It proves that two different degree $d$ polynomials disagree on at least $\delta_d := 2^{-d}$ fraction of the domain (if $d<n$), and thus form natural classes of error-correcting codes. This paper explores the corresponding correction questions focusing on locality.\newline
A special case that is already of interest is when the group $G$ is the group of real numbers (or rationals) - a setting where relatively few codes are shown to exhibit local correction properties.

\paragraph{Local correction of polynomials.} Informally, the local correction problem is that of computing, given oracle access to a function $f:\{0,1\}^n \to G$ and a point $\mathbf{a} \in \{0,1\}^n$, the value $P(\mathbf{a})$ of the nearest degree $d$ polynomial $P$ to the function $f$ at the point $\mathbf{a}$, while making few oracle queries to $f$. More formally, for functions $f,g:\{0,1\}^n\to G$, let $\delta(f,g)$ denote the fraction of points from the domain where they differ. We say $f$ is $\varepsilon$-close to $g$ if $\delta(f,g)\leq \varepsilon$ and $\varepsilon$-far otherwise. For a given $G$, we say that that $\mc{P}_{d}$ is $(\delta,q)$-locally correctable if for every $n$ there is a probabilistic algorithm that, for every function $f:\{0,1\}^n\to G$ that is $\delta=\delta(n)$-close to some polynomial $P \in \mc{P}_{d}(\Boo^{n}, G)$ and for every $\mathbf{a} \in \{0,1\}^n$, outputs $P(\mathbf{a})$ with probability at least $3/4$ while making at most $q=q(n)$ queries to $f$.\newline 
One of the main quests of this work is to give non-trivial upper bounds on the query complexity $q$ for which $\mc{P}_d$ is $(\Omega_d(1),q)$-locally correctable. 

\paragraph{List correction of codes.} 
Note that $(\delta,q)$-locally correctability of $\mc{P}_d$ requires that $\delta$ is less than half the minimum distance of the space, i.e., $\delta <\delta_d/2$. To go beyond one usually resorts to the notion of list-decoding; and in the local setting, to notions like ``local list-decoding'' and ``local list correction''. Roughly list-decoding allows the decoder to output a small list of words with the guarantee that all codewords within a given distance are included in the output. Formally we say $\mc{P}_d$ is (combinatorially) $(\delta,L)$-list correctable if for every $f:\{0,1\}^n\to G$ there are at most $L$ degree-$d$ polynomials $P$ satisfying $\delta(f,P)\leq \delta$.\footnote{The algorithmic version would require the list of nearby polynomials to be algorithmically recoverable from $f$. We don't consider this notion in this work but move ahead to the harder ``local list correction'' problem.} Unlike the unique decoding problem where the maximum $\delta$ such that a code is uniquely correctable up to $\delta$ errors is well understood, the list-decoding radius for higher values of $L$ is not well-understood. A natural question that we study here (for the first time in this generality) is: \textit{What is the largest $\delta$ such that $\mc{P}_d$ is $(\delta, \bigO_d(1))$-list correctable?} We refer to this largest value of $\delta$ as the list-decoding radius of $\mc{P}_d$. 

\paragraph{Local list correction of codes.} Local list correction is the notion of list decoding combined with the notion of local correction. Formalizing this definition is a bit more subtle and was first done in \cite{STV-list-decoding}. The notion allows the decoder to work in two phases --- a preprocessing phase with $q_1=q_1(n)$ queries to the function $f$, that outputs up to $L$ algorithms $\phi_1,\ldots,\phi_L$ and a query phase, where given $\mathbf{a} \in \{0,1\}^n$ each algorithm $\phi_i$ makes $q_2=q_2(n)$ queries to $f$ and outputs $\phi_i(\mathbf{a})$. We say that $\mc{P}_d$ is  $(\delta,q_1,q_2,L)$-local list correctable if for every function $f$ and polynomial $P\in\mc{P}_d$ that are $\delta$-close, there is a decoder as above such that one of its outputs includes $P$ with high probability (say $3/4$). 
See \Cref{defn:local-list-algo} for a formal definition. The final goal of this paper is to locally list-correct $\mc{P}_d$ using non-trivially small number of queries (in both the preprocessing and query phases) where the fraction of errors approaches the list-decoding radius.

\subsection{Motivation and previous work}

Local decoding of polynomials over finite fields has played a central role in computational complexity and in particular in breakthrough results like IP=PSPACE and the PCP theorem. While most of these results consider functions over the entire multivariate domain (i.e., $\F^n$), low-degree polynomials over $\{0,1\}^n$  do arise quite naturally in complexity theory, notably in circuit complexity capturing classes like $\mathrm{AC}^{0}$ \cite{Razborov, Smolensky} and $\mathrm{ACC}$ \cite{ACC-Torus-Learning}, and in learning theory. Many of these results exploit basic distance properties of multivariate polynomials as given by the Ore-DeMillo-Lipton-Schwartz-Zippel lemma (see \Cref{thm:basic}). This lemma roughly says that the space of degree-$d$ polynomial functions mapping $S^n$ to a field $\F$ where $S \subseteq \F$ is finite form an error-correcting code of relative distance $d/|S|$ when $d < |S|$, and $|S|^{-d/(|S|-1)}$ when $d \geq |S|$. 

The special case of $S = \F$ is extensively studied and heavily used, e,g., in PCPs and cryptography. In this setting, the lemma can also be made algorithmic, with the first such instance handling the special case of $\F_2$ dating back to the works of Reed and Muller~\cite{Reed, Muller}. More recently, local list correction algorithms were discovered in the works of Goldreich and Levin~\cite{GoldreichL} for linear polynomials, Sudan, Trevisan and Vadhan~\cite{STV-list-decoding} for higher-degree polynomials over large fields, Gopalan, Klivans and Zuckerman~\cite{gkz-list-decoding} for higher-degree polynomials over $\F_2$, and Bhowmick and Lovett~\cite{BhowmickL} for polynomials over any small finite field. 

The case of general $S$ however has not received much attention and is mostly unexplored. This was first highlighted in a relatively recent work of Kim and Kopparty \cite{kim-kopparty-prod-decoding} who gave polynomial time (but not local) unique-decoding algorithms correcting errors up to half the minimum distance. Their work exposed the fact that many other algorithmic and even some coding-theoretic questions were not well understood when $S \ne \F$, and our work aims to fill some gaps in knowledge here.

Another motivation for our work is the design of locally correctable codes over the reals. A series of works~\cite{BDWY,DSW1, DSW2} has exposed that there are no known $(\Omega(1),\bigO(1))$-locally correctable codes over the reals of arbitrarily large dimensions. The underlying challenge here leads to novel questions in incidence geometry. Roughly the goal here is to design a finite set of points $T \subseteq \R^n$ such many pairs of points in $T$ are contained in constant-dimensional ``degenerate'' subspaces, where a $q$-dimensional subspace is said to be degenerate if it contains $q+1$ points from $T$. Till recently no sets that possessed this property with $q = o(n)$ were known, and the recent results of \cite{ABPSS24-ECCC} may be viewed as showing that the set $T = \{0,1\}^n \subseteq \R^n$ has $\widetilde{O}(\log n)$ dimensional subspaces covering most pairs of points of $T$. Local correctability of degree $d$ polynomials would translate to showing that moment vectors\footnote{$d$-moment vector of $\mathbf{v} \in \Boo^{n}$ is a vector in $\Boo^{\binom{n}{d}}$ which is an evaluation vector of $\mathbf{v}$ on all multilinear monomials of degree $\leq d$.} of $\{0,1\}^n$ (viewed as vectors in $\R^{N}$ for $N = \bigO(n^d)$) also exhibit a similar property, thus adding to the body of sets in $\R^N$ that show non-trivial degeneracies.

Turning to previous work on local correction of polynomials over grids, the local correction question when $S = \{0,1\}$ was first explored by Bafna, Srinivasan and Sudan \cite{bafna2017local}, who mainly presented a lower bound of $\widetilde{\Omega}(\log n)$ on the number of queries to recover even when $d=1$ and from some $o(1)$ fraction of errors and $\F = \R$. On the positive side, for fields of characteristic $p$, they gave an $\bigO_{d,p}(1)$ query algorithm to recover from $\Omega_{d,p}(1)$ fraction of errors. This left the case for large and zero characteristic fields open. 

The recent work of \if\anon1{~Amireddy, Behera, Paraashar, Srinivasan, and Sudan~}\else{the authors}\fi~\cite{ABPSS24-ECCC} investigated the case of general fields, and more generally, polynomials over Abelian groups (i.e., $\{0,1\}^n \to G$), for the special case of $d=1$. For this setting, they consider all three questions posed in the previous section, namely the (unique) local correction limit, the list-decoding radius, and the local list correction problem. In this setting where distinct degree $1$ polynomials disagree with each other on at least half the domain, they show that up to $1/4$ fraction of errors can be uniquely locally corrected with $\widetilde{\bigO}(\log n)$ queries. They further show that the list-decoding radius approaches $1/2$ and that there are $\widetilde{\bigO}(\log n)$ query algorithms to locally list correct $\mc{P}_1$ for any fraction of errors bounded away from $1/2$. Their work exposes a number of technical challenges in going beyond the $d=1$ case and we address those in this paper.

\subsection{Technical challenges in extending beyond the linear case}

We start by reviewing the main ideas in \cite{ABPSS24-ECCC} and outlining the challenges in the higher-degree extension. Their unique local corrector correcting up to half the minimum distance works in three steps: Given an oracle for a function $f(\mathbf{x})$ at a distance less than $1/4$ from a linear polynomial $P(\mathbf{x})$,
\begin{enumerate}[$\blacktriangleright$]
    \item It first provides oracle access to a function $f_{1}(\mathbf{x})$ at any tiny but constant distance from $P(\mathbf{x})$, using $\bigO(1)$ queries to $f(\mathbf{x})$.
    \item Then the algorithm provides oracle access to a function $f_{2}(\mathbf{x})$ at distance $1/\poly(\log n)$ from $P(\mathbf{x})$ while making $\poly(\log \log n)$ queries to $f_{1}(\mathbf{x})$.
    \item Finally, the algorithm provides oracle access to the linear polynomial $P(\mathbf{x})$ making $\bigO(\log n)$ queries to $f_{2}$. Composing the three steps gives the desired unique local corrector.
\end{enumerate}
The first two steps in their result are general enough to work for all degrees. The third step in the unique local corrector of \cite{ABPSS24-ECCC} is the most significant one and does not extend immediately to the higher degree setting. \cite{ABPSS24-ECCC} reduce this step to show that any point in $\{0,1\}^n$ can be expressed as a linear combination of $\widetilde{\bigO}(\log n)$ roughly balanced vectors\footnote{Hamming weight is very close to $n/2$}. To extend their approach to higher degrees, we would need an analogous result for the degree $\leq d$-moment vectors of vectors in $\{0,1\}^n$ but we are unable to find such an extension directly. Instead, as we elaborate further below, we manage to find a new path for this step, which results in a somewhat different proof even for the linear ($d=1$) case.

Next, we turn to the combinatorial analysis of the list-decoding radius. For simplicity assume that the Abelian group $G$ is a finite field $\F_p$. The analysis in the linear case \cite{ABPSS24-ECCC} splits into two cases: the low\footnote{The precise constants here are not important, as the analysis works as long as the $p \leq \bigO(1)$.} characteristic case ($p\leq 3$) and the high characteristic case ($p > 5$). The former case is handled via a suitable version of the Johnson bound. In the latter setting, a key insight used in \cite{ABPSS24-ECCC} for the linear case is that non-sparse linear polynomials tend to be non-zero with very high probability, i.e.~\emph{anti-concentration} of non-sparse linear polynomials. In the higher degree setting, both analyses become much more involved. In the low characteristic case, the Johnson bound no longer yields the right answer. For the high characteristic case, the primary obstacle is understanding the anti-concentration statement for non-sparse low-degree polynomials. 

Generalizing the result above to all Abelian groups involves multiple steps in \cite{ABPSS24-ECCC} - they extend the latter approach above to groups where every element has a sufficiently high order (specifically order at least $5$). Then they consider groups where every element has order a power of $2$ or $3$ separately and analyze these special cases; and finally use some ``special intersection properties'' of the agreement sets\footnote{Agreement set of a polynomial $P$ and a function $f$ is defined as the subset of $\Boo^{n}$ on which $P$ and $f$ agree.} of different polynomials with any given function to apply a counting tool from the work of Dinur, Grigorescu, Kopparty, and Sudan \cite{DGKS} to combine these different steps. While many of the steps extend to the higher degree setting (sometimes needing heavier machinery) the final step involving ``special intersection properties'' simply does not work in our setting. (Roughly the difference emanates from the fact that the probability that two linearly independent degree-$1$ functions vanish at a random point is at most $1/4$, which is the square of the probability for a single degree-$1$ function. This fails for degree $2$.) Overcoming this barrier leads to further new challenges in the higher-degree case. 

The local list-correction algorithm for degree-$1$ polynomials of \cite{ABPSS24-ECCC} is inspired by the local list-correction algorithm of Reed-Muller codes from \cite{STV-list-decoding}. The high-level idea is to start with a function $f(\mathbf{x})$ that is $(1/2-\varepsilon)$-close to a set $\mathsf{List}$ of linear polynomials and produce a small list of oracles such that each polynomial in $\mathsf{List}$ is within a small constant distance to an oracle from the list, at which point the unique local corrector becomes applicable. This `error-reduction' step involves choosing a random subcube $\mathsf{C}$ of $\{0,1\}^n$ (as defined in \Cref{sec:prelims} below) of sufficiently large but constant dimension $k$ and doing a brute-force list decoding on $\mathsf{C}$ to find a list $\mathsf{List}'$: it is not hard to argue that the restricted version $P'$ of each polynomial $P\in \mathsf{List}$ appears in the list $\mathsf{List}'$ with high probability (this needs the combinatorial bound mentioned above). To complete the error-reduction, they need to decode $P$ at a random point $\mathbf{b}\in \{0,1\}^n.$ This is done by repeating the above brute-force algorithm with the subcube $\mathsf{C}^\mathbf{b}$ `spanned' by $\mathsf{C}$ and $\mathbf{b}$: informally, this is the smallest subcube spanned by $\mathsf{C}$ and $\mathbf{b}$ and has dimension $2k$ (see \Cref{sec:prelimslistcorrection} for details.) They now obtain a new list of polynomials $\mathsf{List}''$ and need to isolate the polynomial $P''$ corresponding to $P$ in this list to get $P(\mathbf{b}).$ Here, \cite{ABPSS24-ECCC} uses the fact that we know the restriction of $P$ to the subcube $\mathsf{C}$ inside $\mathsf{C}^{\mathbf{b}}.$ The bad event in this case is that there are two polynomials in $\mathcal{L}''$ that disagree on $\mathbf{b}$ but agree on $\mathsf{C}.$ Bounding the probability of this event is the key step in the analysis of \cite{ABPSS24-ECCC} and is done by giving a complete understanding of which kinds of distinct polynomials over $\mathsf{C}^{\mathbf{b}}$ can collapse to the same polynomial over $\mathsf{C}$. This kind of understanding seems difficult to obtain for higher degrees, and we need to develop a new analysis for bounding the probability of the bad event in this setting.
 
We now turn to our results before elaborating on the techniques used to overcome the challenges.

\subsection{Our main results}
Briefly, our results provide poly-logarithmic query algorithms for unique and list-decoding to the maximal fraction of errors that are allowed. Specifically, the unique decoding algorithm works up to half the distance. We also establish that the list-decoding radius approaches the distance (as the list size tends to infinity) and give matching local algorithms. We give our specific theorems below.

\begin{thmbox}
\begin{restatable}[Local correction algorithms for $\mc{P}_{d}$ up to the unique decoding radius]{theorem}{uniquedegd}\label{thm:uniquedegd}
    For every Abelian group $G$ and for every constant $\varepsilon > 0$, the space $\mc{P}_{d}$ has a $(\delta,q)$-local correction algorithm where $\delta = \frac{1}{2^{d+1}}-\varepsilon$ and $q = \widetilde{\bigO}_{\varepsilon}(\log n)^{d}.$
\end{restatable}
\end{thmbox}

We show that if all its elements of the group have a constant order, then the query complexity of the local correction algorithm can be brought down from $\widetilde{\bigO}_d((\log n)^d)$ to a constant (i.e., independent of $n$). Specifically, we say that an Abelian group $G$ is a {\em torsion group} if all its elements have finite order, and the {\em exponent} of a torsion group is the least common multiple of the orders of all the elements. While~\cite{bafna2017local} shows this only for groups underlying fields of constant characteristic and for {\em some} constant error, we extend their proof to all groups of constant exponent and error up to the unique-decoding radius. 

\begin{restatable}{theorem}{consttorsion}\label{thm:const-torsion}
    If $G$ is an Abelian torsion group of exponent $M$, then for every $\varepsilon > 0$, $\cP_d$ has a $(\delta, q)$-local correction algorithm where $\delta = \frac{1}{2^{d+1}}-\varepsilon$ and $q=\bigO_{M,\varepsilon}(1)$.
\end{restatable} 

As noted earlier, some dependence on $n$ is needed even when the degree is $1$ and $G$ is a field of large characteristic (or characteristic $0$), as an $\Omega(\log n/\log \log n)$ lower bound was shown in this setting by earlier work of Bafna, Srinivasan, and Sudan \cite{bafna2017local} (and shown to be tight up to $\poly(\log\log n)$ factors in~\cite{ABPSS24-ECCC}). Our upper bound above is thus optimal to within polynomial factors (for constant $d$). However, we do not know if the query complexity can be improved to, say, $\tilde{\bigO}_d(\log n)$ for degree $d.$

We also extend the above algorithm from \Cref{thm:uniquedegd} to the list decoding regime. For this, we first establish a bound on the list-decoding radius. As far as we know, the following result  was not known before for $G$ being any group other than the field $\F_2.$ 

\begin{thmbox}
\begin{restatable}[Combinatorial list decoding bound for $\mc{P}_{d}$]{theorem}{comblistdegd}
    \label{thm:comblistdegd}
    For every Abelian group $G$ and for every constant $\varepsilon > 0$, the space $\mc{P}_d$ over any Abelian group $G$ is $(1/2^d - \varepsilon, \exp(\bigO_d(1/\varepsilon)^{\bigO(d)})$-list correctable.
\end{restatable}
\end{thmbox}

The above is tight in the sense that the number of codewords at distance $1/2^{d}$ can depend on both $n$ and the size of the group $G$ (and is infinite when $G$ is infinite). We do not know if the dependence on $\varepsilon$ is tight. Note that for the setting where $d=1$, \cite{ABPSS24-ECCC} gives a polynomial bound in terms of $1/\varepsilon$. Our bound as stated above is exponential and while we can see a path to improving this to a quasi-polynomial, we don't see a polynomial upper bound using the proofs of this paper even when $d=1$. 

Finally, we state our local list correction result. 

\begin{thmbox}
\begin{restatable}[Local list correction for $\mathcal{P}_{d}$]{theorem}{listdecoding}
\label{thm:listdecoding}
For every Abelian group $G$ and for every $\varepsilon>0$, the space $\mathcal{P}_{d}$ is $(1/2^{d}-\varepsilon, \bigO_\varepsilon(1), \widetilde{\bigO}_{\varepsilon}(\log n)^{d}, \exp(\bigO_d(1/\varepsilon)^{\bigO(d)}))$-locally list correctable.\\

Specifically, there is a randomized algorithm $\mathcal{A}$ that, when given oracle access to a polynomial $f$ and a parameter $\varepsilon > 0$, outputs with probability at least $3/4$ a list of randomized algorithms $\phi_1,\ldots, \phi_L$ ($L\leq \exp(\bigO_d(1/\varepsilon)^{\bigO(d)})$)  such that the following holds. For each $P \in \mathcal{P}_{d}$ that is $(1/2^{d} - \varepsilon)$-close to $f$, there is at least one algorithm $\phi_i$ that, when given oracle access to $f$, computes $P$ correctly on every input with probability at least $3/4.$\\

The algorithm $\mathcal{A}$ makes $\bigO_{\varepsilon}(1)$ queries to $f$, while each $\phi_i$ makes $\widetilde{\bigO}_{\varepsilon}((\log n)^{d})$ queries to $f.$ 
\end{restatable}
\end{thmbox}

\begin{remark}
    \label{rem:const-torsion-list} If $G$ is an Abelian torsion group of exponent $M$, $\cP_d$ is $(1/2^d-\varepsilon, \bigO_\varepsilon(1)$, $\bigO_{M,\varepsilon}(1)$, $
    \exp(\bigO_d(1/\varepsilon)^{\bigO(d)}))$-locally list-correctable. This follows in a similar manner as~\Cref{thm:listdecoding}, except we replace the generic local corrector in the unique-decoding regime with that given by~\Cref{thm:const-torsion}.
\end{remark}


\subsection{Technical tools}


In the process of proving our main results, we prove several lemmas that we believe are independently interesting. 

In the proof of~\Cref{thm:uniquedegd}, the main step is to construct, for any $\mathbf{a}\in \{0,1\}^n$, a distribution $\mathcal{D}_\mathbf{a}$ over $(\Boo^{n})^{q}$ such that the marginal distribution of each point is close to the uniform distribution over $\Boo^{n}$, and for any degree-$d$ polynomial $P$, $P(\mathbf{a})$ can be computed via the evaluations of $P$ on a sample from $\mathcal{D}_\mathbf{a}$. Constructing such a distribution $\mathcal{D}$ reduces to the following problem of constructing a geometric set with some nice algebraic properties. We discuss in \Cref{sec:deg-1-decoding} how such a set leads to the distribution $\mathcal{D}_\mathbf{a}$.

\begin{question}
\textcolor{red}{Find two parallel hyperplanes in $k$ dimensions that are $\varepsilon$-close in Euclidean distance such that every non-zero degree-$d$ multilinear polynomial is non-zero on the points of the Boolean cube $\{0,1\}^k$ lying between the two hyperplanes.}
\end{question}

The `closeness' parameter $\varepsilon$ plays a crucial role in the efficiency of the local correction algorithm. It is easy to see (and folklore) that we can get $\varepsilon = 1/k^{\Omega(1)}$. However, we show that we can obtain a construction where $\varepsilon$ is \emph{exponentially small} in $k.$ This is done in \Cref{lemma:local-correction-main}.

The proof of the combinatorial list decoding bound is broken down into two cases depending on the order of elements in the group. The first case is when all elements have order larger than a prime $p_{0}(d)$ (a constant dependent on $d$) and the second case is when the group is a product of $p$-groups for $p < p_0$. For the first case, the key step is an understanding of the anti-concentration properties of low-degree non-sparse polynomials. More precisely, we have the following question.

\begin{question}
    \textcolor{red}{Let $P(x_1,\ldots,x_n)$ be a polynomial of degree $d$ with at least $s$ non-zero monomials. Then how large can $\Pr_{\mathbf{a}\sim \{0,1\}^n}{[P(\mathbf{a}) = 0]}$ be?}
\end{question}
If we take $P = x_1x_2\cdots x_{d-1}\cdot L(x_d,\ldots, x_n)$ where $L$ is a linear polynomial with $s$ monomials, we see that $P$ vanishes with probability approximately $1-2^{-(d-1)}-\bigO(1)/\sqrt{s}$ over (say) the reals. In \Cref{lem:anticonc}, we build on known anti-concentration results~\cite{erdos, MNV} show that this is essentially the worst possible in groups with no elements of small order. 

In the second case, an important step in our proof is a tail inequality for events defined by the vanishing of low-degree polynomials over a field. Let $Q_1(\mathbf{x}),\ldots, Q_t(\mathbf{x})$ be $t$ degree-$d$ polynomials on the same variable set. We know that each is non-zero on a random input with probability at least $2^{-d}$ and hence that the expected number of polynomials vanishing at a uniformly random point $\mathbf{a}\sim \{0,1\}^n$ is at most $(1-2^{-d})\cdot t.$ This leads to the following question.

\begin{question}
\textcolor{red}{Suppose we have a collection of $t$ degree-$d$ polynomials. Can we bound the probability that more than $(1-2^{-d} + \varepsilon)\cdot t$ many of these polynomials vanish at a random point $\mathbf{a}$ in $\{0,1\}^n$?}
\end{question}

Clearly, we cannot get a strong tail bound unless we impose some `independence' constraints on the polynomials (for example, we cannot hope for a strong bound if all the polynomials are in the linear span of a small number of polynomials). We show that we can get a Chernoff-style tail bound under the constraint that the `leading monomials' (under a suitable ordering) of these polynomials are pairwise disjoint. This is done in \Cref{lem:tailbound} using the `Footprint bound'~\cite{GH} (essentially a tool from commutative algebra) and an idea due to Panconesi and Srinivasan~\cite{PS-chernoff}.

We build on this result to prove an optimal bound on the list-decoding radius for degree-$d$ polynomials over small finite fields $\F_p$ (and more generally over groups that are products of $p$-groups, where $p < p_0$). In the setting when we are working with polynomials mapping $\F_p^n$ to $\F_p$, this was done in the works of~\cite{gkz-list-decoding, BhowmickL} via an involved mixture of algebraic and analytic techniques. Unfortunately, these do not seem to be applicable here: one significant reason that appears again and again in our work is that we cannot restrict a given function to an arbitrary subspace in our ambient space since the domain $\{0,1\}^n$ does not have this algebraic structure. Instead we use other combinatorial techniques such as the Sunflower lemma in conjunction with the above tail bound to obtain the stated result.

For local list correction, we follow the algorithm of \cite{ABPSS24-ECCC} modulo changes in parameters to handle higher degree polynomials. The main innovation involves analyzing the behavior of degree-$d$ polynomials on a Hamming slice (points with a fixed Hamming weight) after a random process of reducing the dimension. In particular, assume that we have a fixed degree-$d$ polynomial $R(x_1,\ldots, x_{2k})$ in $2k$ dimensions such that $R$ is non-zero on the Hamming slice of weight $k$. We now choose a random subcube $\mathsf{C}$ by pairing the $2k$ variables at random into $k$ pairs and identifying the variables in each pair. We would like to upper bound the probability that $R$ is zero on the cube $\mathsf{C}$ by a function that goes to $0$ with $k.$\footnote{One might hope to prove such a statement under the weaker assumption that $R$ is simply a non-zero multilinear polynomial of degree $d$. Unfortunately, the simple example $R = x_1+\cdots+ x_{2k}$ over the group $G = \F_2$ shows that such a statement is not possible even in the degree-$1$ case.} We do this by addressing the following two questions.


The first question is on how the density of any fixed set on a slice changes under the aforementioned random process. In this work, we are particularly interested in the middle slice, i.e. points of Hamming weight $k$ in $\{0,1\}^{2k}$.

\begin{question}
\textcolor{red}{For any fixed subset $S$ of the middle slice, how does the density of the set $S\cap \mathsf{C}$ (as a subset of the middle slice of $\mathsf{C}$) compare with the density of $S$?}
\end{question}
In \Cref{lem:sampling} we show that the density is almost preserved under the random process. In other words, this random process is a good sampler for subsets of the middle slice. To prove \Cref{lem:sampling}, we show that certain kinds of \emph{Johnson graphs} are good spectral expanders.

The second question is on a quantitative estimate of the number of non-zero points of a degree-$d$ polynomial on a Hamming slice. This  is a natural question, but has not been addressed before as far as we know.

\begin{question}
\textcolor{red}{For a degree-$d$ polynomial $R$ which is non-zero on a Hamming slice, on how many points of the Hamming slice is it non-zero?}
\end{question}

We give a simple lower bound on the number of non-zero points for a degree-$d$ polynomial on the Hamming slice by modifying the proof of the Ore-DeMillo-Lipton-Schwartz-Zippel lemma. We show it in \Cref{lem:DLSZ-slice}.

\section{Preliminaries}
\label{sec:prelims}
Most of our notation and definitions are identical to~\cite{ABPSS24-ECCC}.
\subsection{Notation}
Let $(G, +)$ denote an Abelian group $G$ with addition as the binary operation. For any $g \in G$, let $-g$ denote the inverse of $g \in G$. For any $g \in G$ and integer $a \geq 0$, $a \cdot g$ (or simply $ag$) is the shorthand notation of $\underbrace{g + \ldots + g}_{a \; \mathrm{times}}$ and $-ag$ denotes $a\cdot (-g).$ We say that a group is a {\em $p$-group} if the order of each element is a power of $p$. We say that a group is  a {\em torsion group} if all its elements have finite order. The {\em exponent} of a torsion group is the least common multiple of the orders of all its elements.

For a natural number $n$, we consider functions $f :\Boo^{n} \to G$. We denote the set of functions that can be expressed as a multilinear polynomial of degree $d$, with the coefficients being in $G$ by $\mathcal{P}_{d}(n,G)$. We will simply write $\mathcal{P}_{d}$ when $n$ and $G$ are clear from the context. For a polynomial $P$, we refer to the number of monomials with a non-zero coefficient as the {\em sparsity} of $P$ and denote it by $\spars(P)$. Similarly we use $\deg(P)$ to denote the degree of $P$.

For every alphabet set $\Sigma$ and $\mathbf{x}, \mathbf{y} \in \Sigma^{n}$, let $\delta(\mathbf{x},\mathbf{y})$ denote the relative Hamming distance between $\mathbf{x}$ and $\mathbf{y}$, i.e. $\delta(\mathbf{x}, \mathbf{y}) = |\setcond{i \in [n]}{x_{i} \neq y_{i}}|/n$. For $0 \leq m \leq n$, let $\Boo^{n}_{m}$ denote the set of points in $\Boo^{n}$ of Hamming weight exactly $m$.

For any $\mathbf{x} \in \Boo^{n}$, $|\mathbf{x}|$ denotes the Hamming weight of $\mathbf{x}$. $\tilde{\bigO}(\cdot)$ notation hides factors that are poly-logarithmic in its argument.
For a polynomial $P(\mathbf{x})$, let $\mathrm{vars}(P)$ denote the variables on which $P$ depends, i.e. the variables that appear in a monomial with non-zero coefficient in $P$.

For any natural number $n$, $U_{n}$ denotes the uniform distribution on $\Boo^{n}$. 

\subsection{Basic definitions and tools}

\paragraph{Probabilistic notions.} For any distribution $X$ on $\Boo^{n}$, let $\mathrm{supp}(X)$ denote the subset of $\Boo^{n}$ on which $X$ takes non-zero probability. For two distributions $X$ and $Y$ on $\Boo^{n}$, the statistical distance between $X$ and $Y$, denoted by $\mathrm{SD}(X,Y)$ is defined as 
\begin{align*}
    \mathrm{SD}(X,Y) = \max_{T\subseteq \{0,1\}^n} |\Pr[X \in T] - \Pr[Y \in T]|
\end{align*}
We say $X$ and $Y$ are $\varepsilon$-close if the statistical distance between $X$ and $Y$ is at most $\varepsilon$.

\paragraph{Coding theory notions.} Fix an Abelian group $G$. We use $\mathcal{P}_d$ to denote the space of multilinear polynomials from $\{0,1\}^n$ to $G$ of degree at most $d.$ More precisely, any element $P\in \mathcal {P}_d$ can be described as
\[
P(x_1,\ldots, x_n) = \sum_{I\subseteq [n] \, : \, |I|\leq d} \alpha_I \prod_{i\in I}x_i
\]
where $\alpha_I\in G$ for each $I.$ On an input $\mathbf{a}\in \{0,1\}^n$, each monomial evaluates to a group element in $G$ and the polynomial evaluates to the sum of these group elements. 

The following is a standard fact about multilinear polynomials which also holds in the setting when the range is an arbitrary Abelian group $G.$ The proof is standard and omitted.
\begin{theorem}[Basic facts about multilinear polynomials]\label{thm:basic}
Let $G$ be any Abelian group. 
\begin{enumerate}
\item Any function $f:\{0,1\}^n\rightarrow G$ has a unique representation as a multilinear polynomial over $G$. In particular two distinct multilinear polynomials cannot agree on all points in $\{0,1\}^n.$ 

\item (DeMillo-Lipton-Schwartz-Zippel (DLSZ) lemma~\cite{DL78,Zippel79,Schwartz80}) More generally, any two distinct multilinear polynomials $P,Q\in \mathcal{P}_d$ differ with probability at least $2^{-d}$ at a uniformly random input from $\{0,1\}^n.$ Equivalently, $\delta(\mathcal{P}_d)\geq 2^{-d}.$
\end{enumerate}
\end{theorem}

\noindent
We now turn to the kinds of algorithms we will consider. Below, let $\mathcal{F}$ be any space of functions mapping $\{0,1\}^n$ to $G$.\\

\noindent
\begin{definition}[Local Correction Algorithm]
We say that $\mc{F}$ has a $(\delta,q)$-local correction algorithm if there is a probabilistic algorithm that, when given oracle access to a function $f$ that is $\delta$-close to some $P\in \mc{F}$, and given as input some $\mathbf{a} \in \{0,1\}^n$, returns $P(\mathbf{a})$ with probability at least $3/4.$ Moreover, the algorithm makes at most $q$ queries to its oracle.
\end{definition}

\begin{definition}[Local List-Correction Algorithm]\label{defn:local-list-algo}
We say that $\mc{F}$ has a $(\delta,q_1,q_2, L)$-local list correction algorithm if there is a randomized algorithm $\mc{A}$ that, when given oracle access to a function $f$, produces a list of randomized algorithms $\phi_1,\ldots, \phi_L$, where each $\phi_{i}$ has oracle access to $f$  and have the following property: with probability at least $3/4$, for each codeword $P$ that is $\delta$-close to $f$, there exists some $i\in [L]$ such that the algorithm $\phi_i$ computes $P$ with error at most $1/4$, i.e. on any input $\mathbf{a}$, the algorithm $\phi_i$ outputs $P(\mathbf{a})$ with probability at least $3/4$. \newline
Moreover, the algorithm $\mc{A}$ makes at most $q_1$ queries to $f$, while the algorithms $\phi_1,\ldots, \phi_L$ each make at most $q_2$ queries to $f$.
\end{definition}

\noindent
\begin{remark}\label{rem:algos}
Our algorithms can all be implemented as standard Boolean circuits with the added ability to manipulate elements of the underlying group $G$. Specifically, we assume that we can store group elements, perform group operations (addition, inverse) and compare two group elements to check if they are equal.
\end{remark}

\begin{definition}[Combinatorial List Decodability]
We say that $\mc{F}$ is $(\delta,L)$-list decodable if for any function $f$, the number of elements of $\mc{F}$ that are $\delta$-close to $f$ is at most $L.$
\end{definition}

\paragraph{Subcubes of $\{0,1\}^n$.} It will be instrumental in our algorithms to be able to restrict the given function to a small-dimensional subcube and analyze this restriction. We construct such subcubes by first negating a subset of the variables and then identifying them into a smaller set of variables. More precisely, we have the following definition from~\cite{ABPSS24-ECCC}. 

\begin{definition}[Embedding a smaller cube into $\{0,1\}^n$]\label{defn:random-embedding}
Fix any $k \in \mathbb{N}$ and $k \leq n$.  Fix a point $\mathbf{a} \in \Boo^{n}$ and a function $h: [n] \to [k]$. For every $\mathbf{y} \in \Boo^{k}$, $x(\mathbf{y})$ is defined with respect to $\mathbf{a}$ and $h$ as follows:
\begin{align*}
    x(\mathbf{y})_{i} = y_{h(i)} \oplus a_{i} = \begin{cases}
        a_{i}, & \text{if } \, y_{h(i)} = 0 \\
        1 \oplus a_{i}, & \text{if } \, y_{h(i)} = 1
     \end{cases}
\end{align*}
$C_{\mathbf{a},h}$ is the subset in $\Boo^{n}$ consisting of $x(\mathbf{y})$ for every $\mathbf{y} \in \Boo^{k}$, i.e. $C_{\mathbf{a},h} := \setcond{x(\mathbf{y})}{\mathbf{y} \in \Boo^{k}}$. In particular, note that this subcube contains the point $\mathbf{a}$, since $x(0^k) = \mathbf{a}.$

Given any polynomial $P(x_1,\ldots, x_n)$ and any subcube $C_{\mathbf{a},h}$  as above, $P$ restricts naturally to a degree-$d$ polynomial $Q(y_1,\ldots, y_k)$ on $C_{\mathbf{a},h}$ obtained by replacing each $x_i$ by $y_{h(i)}\oplus a_i$. We use $P|_{C_{\mathbf{a},h}}$ to denote the polynomial $Q$.
\end{definition}

\paragraph{Random subcubes.} Now assume that we choose a subcube $C_{\mathbf{a},h}$ by sampling  $\mathbf{a} \sim \{0,1\}^n$ and sampling a random hash function $h:[n]\rightarrow [k]$. For every $\mathbf{y} \in \Boo^{k}$, $x(\mathbf{y})$ is the image of $\mathbf{y}$ in $\Boo^{n}$ under $\mathbf{a}$ and $h$ and $C_{\mathbf{a}, h}$ is the subcube consisting of all $2^{k}$ such images. From \Cref{defn:random-embedding}, we have the following simple observation: For every $\mathbf{y} \in \Boo^{k}$, distribution of $x(\mathbf{y})$ is the uniform distribution over $\Boo^{n}$. This is because $\mathbf{a}$ is uniformly distributed over $\Boo^{n}$.

We use the following sampling lemma for subcubes from \cite{ABPSS24-ECCC} that will be useful at multiple points in the paper.

\begin{lemma}[Sampling lemma for random subcubes]
\label{lemma:sampling-subcube}
(\cite[Lemma 2.4]{ABPSS24-ECCC}) Sample $\mathbf{a}$ and $h$ uniformly at random, and let $\mathsf{C} = C_{\mathbf{a}, h}$ be the subcube of dimension $k$ as described in \Cref{defn:random-embedding}. Fix any $T\subseteq \{0,1\}^n$ and let $\mu:= |T|/2^n.$ Then, for any $\varepsilon, \eta > 0$
\begin{align*}
    \Pr_{\mathbf{a},h}\left[\left|\frac{|T\cap \mathsf{C}|}{2^k} - \mu\right| \geq \varepsilon \right] < \eta
\end{align*}
as long as $k\geq \frac{A}{\varepsilon^4\eta^2}\cdot \log\left(\frac{1}{\varepsilon\eta}\right)$ for a large enough absolute constant $A > 0.$
\end{lemma}

\section{Local correction in the unique decoding regime}\label{sec:deg-1-decoding}
In this section, we prove \Cref{thm:uniquedegd}, i.e. we give a local correction algorithm for degree-$d$ polynomials when the error is less than the unique decoding radius (half the minimum distance). The proof of \Cref{thm:uniquedegd} will proceed in two phases:
\begin{itemize}
    \item We give error reduction algorithms that reduce the error from half the minimum distance to sub-constant.
    \item We give a local correction algorithm for degree-$d$ polynomials when the error is sub-constant, say less than $\bigO_{d}(1/ (\log n)^{d})$.
\end{itemize}
The first phase follows from the error reduction algorithm of \cite{ABPSS24-ECCC} and we describe it in \Cref{subsec:error-reduction-close-radius}. In \Cref{subsec:sub-constant-error}, we describe the second phase, which is the new result in this work. We start with a proof overview of the second phase.

\paragraph{Proof overview.} We describe the proof idea behind our local corrector for degree-$d$ polynomials in the sub-constant error regime. Assume that we have oracle access to $f:\{0,1\}^n\rightarrow G$ that is $\delta$-close to an (unknown) polynomial $P\in \mathcal{P}_d(\{0,1\}^n,G).$ For simplicity, let us assume that we want to output the value of $P$ at $\mathbf{a} = 0^{n}$. The proof for an arbitrary $\mathbf{a}$ is more or less the same, except for a minor change. 

The idea (as in other local correction algorithms) is to query $f$ at a set of uniformly distributed (but not independent) points $\mathbf{u}^{(1)},\ldots, \mathbf{u}^{(q)}$. Since these points are uniformly distributed, we are likely to obtain $P(\mathbf{u}^{(1)}),\ldots,P(\mathbf{u}^{(q)})$ in this way. If we could use this information to determine $P(0^n)$ for any polynomial $P$, then we would be done. Indeed, this strategy works when $G=\mathbb{F}_2$~\cite{BLR,AKKLR}. If we restrict to a random $\F_2$-linear subspace $V$ of dimension $d+1$, then the non-zero points of $V$ are (marginally) uniformly distributed and determine the value of $P$ (which is still degree $d$ restricted to $V$) at $0^n.$ Unfortunately, this idea does not make sense for polynomials mapping the Boolean cube to groups other than (vector spaces over) $\F_2$. 

An analogous strategy we can employ over any group is to restrict to a random \emph{subcube}. More precisely, we choose a random function $h:[n]\rightarrow [k]$ (for $d< k\ll n$) according to some distribution and consider the random subcube $C = C_{0^n,h}$ as defined in \Cref{sec:prelims} above (informally, we use the function $h$ to identify co-ordinates in $k$ blocks) and query the function $f$ at points in $C$. To use the strategy above, we would like to find points $\mathbf{u}^{(1)},\ldots, \mathbf{u}^{(q)}\in C$ such that 
\begin{enumerate}
    \item[(a)] the points $\mathbf{u}^{(1)},\ldots, \mathbf{u}^{(q)}\in C$ are uniformly distributed, and
    \item[(b)] the values of $P$ at these points determine $P(0^n)$ for an unknown degree-$d$ polynomial $P.$  
\end{enumerate}
These two properties are in tension with each other. To see this, assume that the function $h$ is chosen uniformly at random. In this case, it can be checked that unless $k$ is large (approximately $\sqrt{n}$), the only points in $\{0,1\}^k$ that correspond\footnote{the correspondence maps $\mathbf{y}\in \{0,1\}^k$ to $x(\mathbf{y})\in C$ as defined in \Cref{sec:prelims}} to uniformly random points in $\{0,1\}^n$ are the \emph{perfectly balanced} points (i.e. the points with an equal number of $0$s and $1$s). All other points of $\{0,1\}^k$ correspond to points with expected Hamming weight outside the range $[n/2 - n/k, n/2 + n/k]$ and hence do not `look uniform'. Unfortunately, querying $P$ at the set of perfectly balanced points does not determine $P(0^n)$. Informally, this is because the set of perfectly balanced points lie on a hyperplane not containing $0^n$, and hence even a degree-$1$ polynomial can `distinguish' between these points and $0^n.$\footnote{This reasoning fails over finite fields of fixed positive characteristic and this failure was used in~\cite{bafna2017local} to devise a local correction algorithm over fixed characteristic via this principle. However, this is true over fields of large characteristic and other groups.}

We fix this by choosing $h$ in a non-uniform manner and relaxing criterion (a) to finding \emph{nearly}\footnote{i.e. statistically close to} uniformly distributed points in $C$. More specifically, assume that we have a probability distribution $\mu = (\mu_1,\ldots, \mu_k)$ over $[k]$, and we sample each $h(i)$ independently according to $\mu.$ Again the points in the cube $C$ that are uniformly distributed correspond to points on the hyperplane $\sum_{j\in [k]}\mu_j y_j = 1/2$ in $\{0,1\}^k$. However, we consider the points corresponding to $\mathbf{y}$ satisfying $|\sum_{j\in [k]}\mu_j y_j - 1/2|\leq \varepsilon$, i.e. points in between two close-by hyperplanes, i.e. these are the `nearly-balanced points' under a weighted version of Hamming weight on $\{0,1\}^k$. Standard probabilistic arguments imply that if $\varepsilon \ll 1/\sqrt{n},$ then these correspond to nearly uniformly distributed points in $\{0,1\}^n$ thus satisfying the modified version of criterion (a). Intuitively, getting $\varepsilon$ to be so small and meaningful at the same time seems easier when some of the weights $\mu_1,\ldots, \mu_k$ are also similarly small (though not all of them can be since they sum to $1$). Standard results about Boolean threshold functions~\cite{MTT, Muroga} show that a hyperplane in $k$ dimensions of this form does not require weights smaller than (approximately) $1/k^k$. This forces $k = \tilde{\Omega}(\log n)$ for this strategy. (More generally,~\cite{bafna2017local} showed a lower bound of $\tilde{\Omega}(\log n)$ queries even for decoding linear polynomials over the Boolean cube and fields of large characteristic.) Indeed, we take $k = \Theta_d(\log n)$, so that this strategy becomes feasible. For this value of $k$, we  will show that we can take $\varepsilon = 1/2^{\Omega_d(k)}.$

The problem of designing $\mu$ can now be stated (in a more general form) in geometric language: find two parallel hyperplanes $H_1, H_2$ in $k$ variables that are at distance $1/2^{\Omega_d(k)}$ (this corresponds to making $\varepsilon$ small) such that evaluating a degree-$d$ polynomial $P$ at the points of $\{0,1\}^k$ between $H_1$ and $H_2$ allow us to deduce the value of $P$ at all other points of the Boolean cube. A set with the latter property is sometimes called a (degree-$d$) \emph{interpolating set} in the literature. A standard interpolating set is the set of points of Hamming weight in the range $\{a,\ldots, a+d\}$ for any non-negative integer $a\leq k-d.$ Unfortunately, the two hyperplanes in this case are at a distance of $\Omega(1/\sqrt{k})$ from each other, which is not good enough in our setting. Indeed, the main technical innovation of this section (\Cref{lemma:local-correction-main} below) is finding a pair of hyperplanes with this specific property that is exponentially close. We believe that this result is independently interesting.

To do this, we use a pair of carefully chosen Boolean threshold functions that require weights of exponentially different magnitudes to describe. We do this in a way that allows us to prove the interpolating set property via a modified version of the DeMillo-Lipton-Schwartz-Zippel lemma (\Cref{lemma:local-correction-main} below). To reduce the query complexity of the algorithm, we also need to choose a \emph{small} subset of the nearly-balanced points as defined above that form a small interpolating set. Over a field, an interpolating set of size $\bigO(k^d)$ follows immediately from a linear algebraic argument. We can get a set $\mc{S}$ of similar size $q = \bigO_d(k^d)$ that works for any Abelian group. We then sample the function $h$ as described above to obtain the corresponding (nearly uniformly distributed) points $\mathbf{u}^{(1)},\ldots, \mathbf{u}^{(q)}\in \{0,1\}^n.$

\paragraph{Comparison with~\cite{ABPSS24-ECCC}.} At a high level, the final construction seems to use similar ideas to an analogous step in the low-error local correction algorithm of~\cite{ABPSS24-ECCC}, but the technical details are quite different. If we consider the random $n\times q$ matrix $A$ that contains $\mathbf{u}^{(1)},\ldots, \mathbf{u}^{(q)}$ as its columns, then in the result of~\cite{ABPSS24-ECCC} it is shown how to use a single hyperplane\footnote{in fact the Boolean points on the hyperplane} with large coefficients to define the \emph{rows} of the matrix $A$. Here, each point in between the hyperplanes allows us to sample a distinct \emph{column} of the matrix $A.$

\subsection{Regime of sub-constant error}\label{subsec:sub-constant-error}
In this subsection, we give a local correction algorithm for degree $d$ polynomials in the setting of sub-constant error. Formally, we prove the following statement in this subsection.

\begin{thmbox}
\begin{restatable}[Local correction for sub-constant error]{theorem}{uniquesmallerr}\label{thm:uniquedeg-d-smallerror}
Fix a degree parameter $d \in \mathbb{Z}_{> 0}$ and an Abelian group $G$. The space $\mathcal{P}_{d}(\Boo^{n}, \, G)$ has a $(\delta, q)$-local correction algorithm where $q = \bigO_{d}((\log n)^{d})$ and $\delta = 1/100 q$.
\end{restatable}
\end{thmbox}

To prove \Cref{thm:uniquedeg-d-smallerror}, we use \Cref{thm:weight-interpolating} (see below). It roughly says that there exists a set of points such that an arbitrary evaluation of any degree-$d$ polynomial can be computed using evaluations on this set and this set consists of points whose relative ``weighted Hamming weights'' are very close to $1/2$.

\begin{restatable}[Weight balanced interpolating set]{theorem}{weightinterpolating}\label{thm:weight-interpolating}
Fix a degree parameter $d\geq 0$ and a dimension parameter $k \in \mathbb{Z}_{> 0}$ that is divisible by $10(d+1)$. There exists a set $\mathcal{S} \subseteq \Boo^{k}$ such that for every Abelian group $G$, $\mathcal{S}$ satisfies the following properties:
\begin{enumerate}
    \item\protect{[Interpolating set].} For each point $\mathbf{b} \in \Boo^{k}$, there exists integral coefficients $c_{1},\ldots,c_{|S|}$ such that for every degree-$d$ polynomial $Q(y_{1},\ldots,y_{k}) \in \mathcal{P}_{d}(\Boo^{k}, \; G)$, we have,
    \begin{align*}
        Q(\mathbf{b}) \; = \; \sum_{\mathbf{u} \in \mathcal{S}} \, c_{\mathbf{u}} Q(\mathbf{u})
    \end{align*}
    \item\protect{[Weighted balanced].} There exists positive integers $w_{1},\ldots,w_{k}$ such that
    \begin{align*}
        \mathcal{S} \; \subseteq \; \setcond{\mathbf{y} \in \Boo^{k}}{ \left| \sum_{j = 1}^{k} w_{j} y_{j} - \frac{W}{2} \right| \, \leq \, \frac{W}{2^{\Omega(k/(d+1))}} }, 
    \end{align*}
    where $W := \sum_{j=1}^{k} w_{j}$. 
\end{enumerate}
Furthermore, $|\mathcal{S}|$ is at most $\bigO_d(k^{d})$.
\end{restatable}

\noindent
Before we prove \Cref{thm:weight-interpolating}, let us first see how \Cref{thm:weight-interpolating} is useful in designing a local correction algorithm and to prove \Cref{thm:uniquedeg-d-smallerror}. Below, we assume \Cref{thm:weight-interpolating} and prove \Cref{thm:uniquedeg-d-smallerror}.

\begin{proof}[Proof of \Cref{thm:uniquedeg-d-smallerror}]
Fix an input point $\mathbf{a} \in \Boo^{n}$ to the local correction algorithm. Let $f(x_{1},\ldots,x_{n}) : \Boo^{n} \to G$ be the input function with $\delta(f, \mathcal{P}_{d}) \leq \delta$. Let $P(\mathbf{x})$ be the unique degree $d$ polynomial such that $\delta(f, P) = \delta$. Our goal is to compute $P(\mathbf{a})$ with probability at least $3/4$ using oracle queries to $f(\mathbf{x})$.

For a fixed function $h: [n] \to [k]$, let $C_{\mathbf{a}, h}$ denote the subcube as defined in \Cref{defn:random-embedding}. For a probability distribution $\mu = (\mu_{1},\ldots,\mu_{k})$ on $[k]$, sampling a random function $h:[n] \to [k]$ according to $\mu$ means the following: For each $i \in [n]$ sample independently $h(i) \sim \mu$, i.e. $h(i)$ is equal to $j$ with probability $\mu_{j}$. We will define it shortly using $h \sim \mu$.

Let $\mathcal{S}$ be the set and $w_{1},\ldots,w_{k}$ be the positive integers as described in \Cref{thm:weight-interpolating}. Let $\mu := \paren{\dfrac{w_{1}}{W}, \ldots, \dfrac{w_{k}}{W}}$. We now describe the local correction algorithm.

\begin{algobox}
\begin{algorithm}[H]
\DontPrintSemicolon
\KwIn{$f(x_{1},\ldots,x_{n})$, $\mathbf{a} \in \Boo^{n}$, $\delta$}
\vspace{2mm}

$k \leftarrow A\cdot (d+1)\cdot (\log n) $ \tcp*{$A$ an absolute constant chosen below}
\vspace{2mm}
Sample a random function $h: [n] \to [k]$ according to the distribution $\mu$ \label{algoline:local-correction-sample-hash} \tcp*{The only source of randomness}
\vspace{2mm}
$g(y_{1},\ldots,y_{k}) \leftarrow f(x_{1},\ldots,x_{n})|_{C_{\mathbf{a}, h}}$ \;
\vspace{2mm}
Let $\mathbf{b} = 0^{k}$ and $c_{1},\ldots,c_{|\mathcal{S}|}$ be the integral coefficients for $0^{k}$ from \Cref{thm:weight-interpolating}. \;
\vspace{2mm}
\Return{$\sum_{\mathbf{u} \in \mathcal{S}} \; c_{\mathbf{u}} \, g(\mathbf{u}) $}

\caption{Local correction algorithm for sub-constant error}
\label{algo:local-correction-sub-constant-error}
\end{algorithm}
\end{algobox}

\paragraph{Queries:}The number of queries is equal to $|\mathcal{S}| \leq \bigO_d(k^d) = \bigO_d(\log n)^d$ by \Cref{thm:weight-interpolating} and the value of $k$ in \Cref{algo:local-correction-sub-constant-error}.

 \paragraph{Correctness:}We now argue that \Cref{algo:local-correction-sub-constant-error} returns $P(\mathbf{a})$ with probability $\geq 3/4$. Let $E \subset \Boo^{n}$ denote the set of points where $f$ and $P$ disagree, i.e.
 \begin{align*}
     E \; = \; \setcond{\mathbf{x} \in \Boo^{n}}{f(\mathbf{x}) \neq P(\mathbf{x})}
 \end{align*}
We have $|E|/2^{n} \leq \delta$ because $\delta(f, P) \leq \delta$. Recall that for each $\mathbf{y} \in \Boo^{k}$, for all $i \in [n]$, $x(\mathbf{y})_{i} = y_{h(i)} \oplus a_{i}$, where $h$ is function sampled in \Cref{algoline:local-correction-sample-hash} of \Cref{algo:local-correction-sub-constant-error}.

 We first argue that if $f$ and $P$ agree at $x(\mathbf{y})$ for every $\mathbf{y} \in \mathcal{S}$, then \Cref{algo:local-correction-sub-constant-error} returns $P(\mathbf{a})$.  If for every $\mathbf{y} \in \mathcal{S}$, $x(\mathbf{y})$ is not in $E$, then $g = P|_{\mathcal{S}}$. Since $x(0^{k}) = \mathbf{a}$, the first property of \Cref{thm:weight-interpolating} implies $g(0^{k}) = P(\mathbf{a})$.

 Next we show that with probability at least $3/4$, for every $\mathbf{y} \in \mathcal{S}$, $x(\mathbf{y}) \notin E$. Equivalently, we show the following:
 \begin{equation}\label{eqn:uniquedeg-d-error-probability}
     \Pr_{h}[\exists \; \mathbf{y} \in \mathcal{S} \, \text{ s.t. } \, x(\mathbf{y}) \in E] \; < \; \dfrac{1}{4}
 \end{equation}
 First, we understand the distribution of $x(\mathbf{y})$ for a fixed $\mathbf{y} \in \mathcal{S}$ under a random function $h: [n] \to [k]$ sampled according to $\mu$. Fix any $\mathbf{y} \in \mathcal{S}$ and a coordinate $i \in [n]$. Since $\mathbf{a}$ is fixed, we have,
 \begin{align*}
     \Pr_{h \sim \mu}[x(\mathbf{y})_{i} = 1] \, = \, \mathbb{E}_{h \sim \mu}[x(\mathbf{y})_{i}] \, = \, \sum_{j = 1}^{k} \dfrac{w_{j}}{W} \cdot y_{j}
 \end{align*}
 From the second property in \Cref{thm:weight-interpolating}, we have,
 \begin{align*}
    \left| \Pr_{h \sim \mu}[x(\mathbf{y})_{i} = 1] - \dfrac{1}{2} \right| \leq \dfrac{1}{2^{\Omega(k/(d+1))}}
 \end{align*}
 For each $\mathbf{y} \in \Boo^{k}$, the coordinates $\setcond{x(\mathbf{y})_{i}}{i \in [n]}$ are mutually independent (this is because $h[i]$ is sampled independently for each $i \in [n]$) and $1/2^{\Omega(k/(d+1))}$-close to the uniform distribution over $\Boo$. From \Cref{fact:close-to-uniform}, we know that for every $\mathbf{y} \in \mathcal{S}$, $x(\mathbf{y})$ is $\sqrt{n}/2^{\Omega(k/(d+1))}$-close to the uniform distribution (in statistical distance).
 \begin{fact}[Closeness to the uniform distribution]\label{fact:close-to-uniform}
    (See \cite[Theorem 5.5, Claim 5.6]{MansourLec}). Let $\eta > 0$. Let $\mathcal{D'}$ be a distribution on $\Boo^{n}$ such that for any $\mathbf{y} \sim \mathcal{D'}$, the co-ordinates of $\mathbf{y}$ are independent and for all $i \in [n]$,
\begin{align*}
    1/2 - \eta \leq \Pr[y_{i} = 1] \leq 1/2 + \eta.
\end{align*}
Then $\mathcal{D'}$ is $\bigO(\eta \sqrt{n})$-close to the uniform distribution over $\Boo^{n}$. 
 \end{fact}
 By the definition of statistical distance, we have that for every $\mathbf{y} \in \mathcal{S}$,
 \begin{align*}
     \Pr_{h}[x(\mathbf{y}) \in E] \; \leq \; \sqrt{n}/2^{\Omega(k/(d+1))} + \delta
 \end{align*}
 Taking a union bound over all $\mathbf{y} \in \mathcal{S}$, we have,
 \begin{align*}
     \Pr_{h}[\exists \, \mathbf{y} \in \mathcal{S} \, \text{ s.t. } \, x(\mathbf{y}) \in E] \; \leq \; |\mathcal{S}| \paren{\sqrt{n}/2^{\Omega(k/(d+1))} + \delta}
 \end{align*} 
 Recall that the number of queries is $q=|\mathcal{S}|$ and by assumption $\delta < 1/100 q$. Thus, we get
 \begin{align*}
     \Pr_{h}[\exists \, \mathbf{y} \in \mathcal{S} \, \text{ s.t. } \, x(\mathbf{y}) \in E] \leq \frac{\Tilde{\bigO}(\sqrt{n})}{2^{\Omega(k/(d+1))}} + \dfrac{1}{100} \leq \frac{1}{4}
 \end{align*}
 as long as the constant $A$ in \Cref{algo:local-correction-sub-constant-error} is chosen to be large enough. This shows \Cref{eqn:uniquedeg-d-error-probability} as claimed and thus we have described a $(\delta, q)$ local correction algorithm for $\mathcal{P}_{d}$.
 \end{proof}

\subsection{Weight balanced interpolating set}
In this subsection, we prove \Cref{thm:weight-interpolating}. We start by proving \Cref{lemma:local-correction-main}, which is our main technical lemma of this subsection. The difference between \Cref{lemma:local-correction-main} and \Cref{thm:weight-interpolating} is in the first condition. In \Cref{lemma:local-correction-main} we require that every non-zero degree-$d$ polynomial is non-zero on $\mathcal{S}$. Later, we will see that this is sufficient to allow us to compute $Q$ at any point.\footnote{For polynomials over fields, this follows simply from linear algebra. For Abelian groups, the proof is similar, but we need a result from the theory of Diophantine linear equations.}

\begin{thmbox}
\begin{lemma}[Main lemma for local correction]\label{lemma:local-correction-main}
Fix a degree parameter $d\geq 0$ and a dimension parameter $k \in \mathbb{Z}_{> 0}$ such that $k$ is divisible by $10\cdot (d+1)$. There exists a set $\mathcal{S} \subseteq \Boo^{k}$ such that for every Abelian group $G$, the set $\mathcal{S}$ satisfies the following properties:\\
\begin{itemize}
    \item For every non-zero degree-$d$ polynomial $Q(y_{1}, \ldots, y_{k}) \in \mathcal{P}_{d}(\Boo^{k}, \, G)$, there exists a point $\mathbf{z} \in \mathcal{S}$ such that $Q(\mathbf{z}) \neq 0$.
    \vspace{2mm}
    \item There exists positive integers $w_{1},\ldots,w_{k}$ such that
    \begin{align*}
        \mathcal{S} \; \subseteq \; \setcond{\mathbf{y} \in \Boo^{k}}{ \left| \sum_{j = 1}^{k} w_{j} y_{j} - \frac{W}{2} \right| \, \leq \, \frac{W}{2^{\Omega(k/(d+1))}} }, 
    \end{align*}
    where $W := \sum_{j=1}^{k} w_{j}$. 
\end{itemize}
Furthermore, $|\mathcal{S}|$ is at most $\bigO_d(k^{d})$.
\end{lemma}
\end{thmbox}

\begin{proof}[Proof of \Cref{lemma:local-correction-main}]
Let $r = 10\cdot (d+1)$ and let $\Sigma := \Boo^{r}$. Let $m = k/r$.  In this proof, it will be convenient to make the following identification:
\begin{align*}
    [mr] \cong [m] \times [r] \quad \text{ and } \quad \Boo^{k} = \Sigma^{m}
\end{align*}
We interpret $\mathbf{y} \in \Boo^{mr}$ as $\mathbf{y} = (\mathbf{y}[1], \ldots, \mathbf{y}[m])$ where for every $i \in [m]$, $\mathbf{y}[i] = (y[i,1], \ldots, y[i,r]) \in \Sigma = \Boo^{r}$.

We now describe the weights $w_1,\ldots, w_k$. For every for $(i,j) \in [m] \times [r] \cong [mr]$, let $w_{(i,j)} := 2^{i-1}$. Define the set $\mathcal{B}_{m}$ as follows.
\begin{align*}
    \mathcal{B}_{m} \; := \; \setcond{\mathbf{y} \in \Sigma^{m}}{\left| \sum_{i = 1}^{m} \sum_{j = 1}^{r} w_{(i,j)} y[i,j] - \dfrac{W_{m}}{2} \right| \, \leq \, t } \;
    = \; \setcond{\mathbf{y} \in \Sigma^{m}}{\left| \sum_{i = 1}^{m} 2^{i-1} \sum_{j = 1}^{r} y[i,j] - \dfrac{W_{m}}{2} \right| \, \leq \, t }
\end{align*}
where $W_{m} := \sum_{i = 1}^{m} 2^{i-1} \sum_{j = 1}^{r} 1 = r (1 + 2 \cdots + 2^{m-1})$ and $t = \lceil \frac{d}{2}\rceil$. We will choose the set $\mathcal{S}$ to be a subset of $\mathcal{B}_m.$ Note that $t = W_m/2^{m}$, which shows that the weights are as required by the second item of the statement of the lemma.

More generally, we can define a subset $\mathcal{B}_\ell\subseteq \{0,1\}^{\ell\cdot r}$ for each $\ell\in [m]$ in a similar way. We let $W_\ell$ denote the sum of the weights in this case. We will need the following claim about extending points in $\mathcal{B}_\ell$ to points in $\mathcal{B}_{\ell+1}$ in many different ways.

\begin{claim}
    \label{clm:extendBell}
    Fix $\ell \in [m-1]$. Let $\mathbf{b} = (\mathbf{b}[1],\ldots, \mathbf{b}[\ell])$ be any point in $\mathcal{B}_\ell$. There exists an interval (i.e. set of consecutive integers) $I_{\mathbf{b}}\subseteq \{0,\ldots,r\}$ of size at least $d+1$ such that for every point $\mathbf{z}\in \{0,1\}^r$ such that $|\mathbf{z}|\in I_{\mathbf{b}},$ the point $\mathbf{b}' = (\mathbf{z},\mathbf{b}[1],\ldots, \mathbf{b}[\ell])\in \mathcal{B}_{\ell+1}.$
\end{claim}

\begin{proof}[Proof of \Cref{clm:extendBell}]
    Note that $W_{\ell+1} = 2W_\ell + r.$ Since $\mathbf{b}\in \mathcal{B}_\ell,$ we have
    \begin{equation}
        \label{eq:extendBell1}
        \sum_{i=1}^\ell |\mathbf{b}[i]|\cdot 2^{i-1} = \frac{W_\ell}{2}+\tau
    \end{equation}
    for some integer $\tau$ such that $|\tau|\leq t$ (here $|\mathbf{b}[i]|$ is the Hamming weight of $\mathbf{b}[i]$). For $\mathbf{b}' = (\mathbf{z},\mathbf{b}[1],\ldots, \mathbf{b}[\ell])$ to lie in $\mathcal{B}_{\ell+1}$, we need
    \[
    |\mathbf{z}| + \sum_{i=1}^\ell |\mathbf{b}[i]|\cdot 2^{i}\in \left[\frac{W_{\ell+1}}{2}-t, \frac{W_{\ell+1}}{2}+t\right].
    \]
    By \Cref{eq:extendBell1} and the relationship between $W_\ell$ and $W_{\ell+1}$, the latter is equivalent to $|\mathbf{z}|\in I_\mathbf{b} := [\frac{r}{2}-2\tau-t, \frac{r}{2} -2\tau+t]$. Since $r/2 \geq 5(d+1) \geq 2|\tau| + t,$ the interval $I_{\mathbf{b}}\subseteq \{0,\ldots, r\}$ and further, $|I_{\mathbf{b}}| = 2t+1 \geq d+1.$
\end{proof}

We also need a basic fact about degree-$d$ polynomials over the Boolean cube. This is a combination of a few folklore facts, but we prove it here for completeness.

\begin{claim}\label{clm:hamming-ball-d-non-zero}
Fix any $n \geq 1$ and degree parameter $d \leq n.$ For every interval $I \subseteq \{0,\ldots,n\}$ (i.e. set of consecutive integers) of size $(d+1),$ there exists a set $\mathcal{H}_{I,d}\subseteq \{0,1\}^n$ of size at most $(2(n+1))^{d}$ such that
\begin{itemize}
    \item $\mathcal{H}_{I,d}$ consists only of points $\mathbf{z}$ such that $|\mathbf{z}|\in I$, and
    \item For any non-zero $P\in \mathcal{P}_d(\{0,1\}^n, G)$, there is a point $\mathbf{z}\in \mathcal{H}_{I,d}$ such that $P(\mathbf{z})\neq 0$.
\end{itemize}
\end{claim}

\begin{proof}[Proof of \Cref{clm:hamming-ball-d-non-zero}]
Assume that $I = \{a, a+1,\ldots, a+d\}$ for some $a \in [0, n-d]$. For every subset $A \subseteq [n]$ of size $\leq d$, we define the set $\mathcal{H}_{I, d}^{A}$ as follows:
\begin{align*}
    \mathcal{H}_{I, d}^{A} \; := \; \setcond{\mathbf{x} 1^{a} 0^{n - |A| - a}}{\mathbf{x} \in \Boo^{|A|}},
\end{align*}
where $\mathbf{x} 1^{a} 0^{n - |A| - a}$ is a shorthand notation for the point where the variables indexed by $A$ are set to $\mathbf{x}$, the first $a$ variables of the remaining variables are set to $1$, and the last $(n-|A|-a)$ variables of the remaining variables are set to $0$. In other words, $\mathcal{H}_{I, d}^{A}$ consists of points where the variables indexed by $[n] \setminus A$ are set to $1^{a} 0^{n - |A| - a}$ (this has Hamming weight $a$) and the variables indexed by $A$ can be any point of Hamming weight $\leq |A| \leq d$. We define the set $\mathcal{H}_{I, d}$ as follows:
\begin{align*}
    \mathcal{H}_{I, d} \, = \, \bigcup_{\substack{A \subseteq [n] \\ |A| \leq d}} \; \mathcal{H}_{I, d}^{A}
\end{align*}
For every subset $A$, the set $\mathcal{H}_{I, d}^{A}$ consists only of points $\mathbf{z}$ such that $|\mathbf{z}| \in I$, and so thus every point in $\mathcal{H}_{I, d}$. Note that for each subset $A \subseteq \binom{[n]}{\leq d}$, the set $|\mathcal{H}_{I,d}^{A}| \leq 2^d$. The size of $\mathcal{H}_{I, d}$ is at most $2^d \cdot (\sum_{i = 0}^{d} \binom{n}{i}) \leq (2(n+1))^{d}$.

Next we show that for any non-zero $P \in \mathcal{P}_{d}(\Boo^{n}, G)$, there is a point $\mathbf{z} \in \mathcal{H}_{I, d}$ such that $P(\mathbf{z}) \neq 0$. Fix any non-zero polynomial $P \in \mathcal{P}_{d}(\Boo^{n}, G)$. $P$ can be uniquely expressed as the following multilinear polynomial:
\begin{align*}
    P(x_{1},\ldots,x_{n}) \, = \, \sum_{\substack{A \subseteq [n] \\ |A| \leq d}} c_{A} \mathbf{x}^{A},
\end{align*}
where $\mathbf{x}^{A}$ is the product of variables indexed by the set $A$. Since $P$ is a non-zero polynomial, at least one of the coefficients $c_{A}$'s is non-zero. Let $A_{0}$ be the maximal (with respect to the inclusion partial order) subset with a non-zero coefficient. Set the variables outside $A_{0}$ to $1^{a} 0^{n - |A_0| - a}$. The resulting polynomial is a non-zero degree $d$ polynomial in variables indexed by $A_{0}$, i.e. in $\leq d$ variables. We know (\Cref{thm:basic}) that every non-zero degree-$d$ polynomial is non-zero on the Boolean cube. Thus the resulting polynomial is non-zero on a point in $\Boo^{|A_{0}|}$. This implies that $P$ is non-zero on a point in $\mathcal{H}_{I,d}^{A_{0}} \subseteq \mathcal{H}_{I,d}$.
\end{proof}

We now show how to construct the set $\mathcal{S}$ as required in the statement of the lemma. In fact, we will show a stronger property: for each $\ell\in [m]$ and $j\in \{0,\ldots,d\}$, we show that there is a set $\mathcal{S}_{\ell,j}\subseteq \mathcal{B}_\ell$ that satisfies the first item of the lemma w.r.t. the space of polynomials $\mathcal{P}_j(\{0,1\}^{\ell\cdot r}, G)$. Furthermore, for each $\ell,j$, we will have $|\mathcal{S}_{\ell,j}|\leq (2(r+1))^{j}\cdot \ell^j.$

We prove the above by induction on $\ell+j$. The base case corresponds to $\ell=1$ and $j=0$, where we can take $\mathcal{S}_{\ell,0}$ to be any fixed point $\mathbf{z}\in \mathcal{B}_\ell.$

Now consider the inductive case for some $j > 0.$ Given a non-zero polynomial $P(\mathbf{y}[1],\ldots,\mathbf{y}[\ell])\in \mathcal{P}_j(\{0,1\}^{\ell\cdot r}, G)$, we can decompose it as a polynomial in the variables in $\mathbf{y}[1]$ with coefficients coming from the space of polynomials in the remaining variable sets 
$\mathbf{y}[2],\ldots,\mathbf{y}[\ell]$. This gives the following equality.
\begin{equation}
    \label{eq:decompose-P}
    P(\mathbf{y}[1],\ldots,\mathbf{y}[\ell]) = \sum_{A \subseteq [r]: |A|\leq j} Q_A(\mathbf{y}[2],\ldots,\mathbf{y}[\ell])\cdot \mathbf{y}[1]^A
\end{equation}
where $\mathbf{y}[1]^A$ denotes the product of the variables in $\mathbf{y}[1]$ indexed by $A$ and $Q_A$ denotes the sum of all monomials in $\mathbf{y}[2],\ldots,\mathbf{y}[\ell]$ multiplying this monomial. Note that $Q_A$ has degree at most $j-|A|.$

Fix a set $A_0$ such that $Q_{A_0}$ is non-zero and $|A_0|$ is as large as possible. Assume that $|A_0| = j'\in \{0,\ldots, j\}.$ We know by induction that there is a point $\mathbf{b} \in \mathcal{S}_{\ell-1,j-j'}$ such that $Q_{A_0}(\mathbf{b})\neq 0.$ Let $P_{\mathbf{b}}(\mathbf{y}[1])$ denote the restriction of the polynomial $P$ when the variable sets $\mathbf{y}[2],\ldots, \mathbf{y}[\ell]$ are set according to $\mathbf{b}$. The polynomial $P_{\mathbf{b}}$ is a non-zero polynomial of degree $j'.$

We want to extend $\mathbf{b}$ to an assignment also setting the variables $\mathbf{y}[1]$ that keeps the polynomial $P$ non-zero. By \Cref{clm:extendBell}, there is an interval $I_\mathbf{b}$ of size at least $(d+1)$ such that for any $\mathbf{z}$ such that $|\mathbf{z}|\in I_{\mathbf{b}}$, the point $(\mathbf{z}, \mathbf{b}[1],\ldots,\mathbf{b}[\ell-1])\in \mathcal{B}_\ell.$ Fix any subinterval $I_{\mathbf{b},j'}\subseteq I_\mathbf{b}$ of size $j'+1$. By \Cref{clm:hamming-ball-d-non-zero}, there is a set $\mathcal{H}_{I_{\mathbf{b},j'},j'}\subseteq \{0,1\}^r$ of size at most $(2(r+1))^{j'}$ such that each point $\mathbf{z}$ has Hamming weight in $I$ and further $P_{\mathbf{b}}(\mathbf{z})\neq 0.$

We have thus shown that $P$ must be non-zero at one of the points in the following set.
\[
\mathcal{S}_{\ell,j,j'} = \{(\mathbf{z},\mathbf{b}[1],\ldots, \mathbf{b}[\ell-1])\ |\ \mathbf{b} = (\mathbf{b}[1],\ldots, \mathbf{b}[\ell-1])\in \mathcal{S}_{\ell-1,j-j'},\ \mathbf{z}\in \mathcal{H}_{I_{\mathbf{b},j'},j'}\}.
\]
However, the above assumes that we know the parameter $j'$ of $P.$ To define the set $\mathcal{S}_{\ell,j}$, we take a union of all the sets $\mathcal{S}_{\ell,j,j'}$ for $j'\in \{0,\ldots,j\}.$ This satisfies the required inductive property.

It remains to bound $|\mathcal{S}_{\ell,j}|.$ We have
\begin{align*}
    |\mathcal{S}_{\ell,j}| &\leq \sum_{j'=0}^j |\mathcal{S}_{\ell,j,j'}|\leq \sum_{j'=0}^j |\mathcal{S}_{\ell-1,j-j'}|\cdot |\mathcal{H}_{I,j'}|\leq \sum_{j'=0}^j (2(r+1))^{(j-j')}\cdot (\ell-1)^{j-j'}\cdot (2(r+1))^{j'}\\
    &= (2(r+1)^{j}\cdot ((\ell-1)^j + (\ell-1)^{j-1} + \cdots + (\ell-1)^0) \leq (2(r+1))^{j}\cdot \ell^j
\end{align*}
proving the required bound on $|\mathcal{S}_{\ell,j}|.$ This proves the inductive claim.

To conclude the proof of \Cref{lemma:local-correction-main}, if we take $\mathcal{S} = \mathcal{S}_{m,d},$ then we have a set with the required properties and size.
\end{proof}

We now show how \Cref{lemma:local-correction-main} implies \Cref{thm:weight-interpolating}. We recall \Cref{thm:weight-interpolating} here.

\weightinterpolating*

\begin{proof}[Proof of \Cref{thm:weight-interpolating}]
Let $\mathcal{S} \subseteq \Boo^{k}$ be the subset as given by \Cref{lemma:local-correction-main}. Fix any point $\mathbf{b} \in \Boo^{k}$. Let $\mathsf{B}_{d}$ denote the set of multilinear monomials of degree $\leq d$ in $\set{x_{1},\ldots,x_{k}}$. $\mathsf{B}_{d}$ forms a spanning set of $\mathcal{P}_{d}(\Boo^{k}, \, G)$ for every $G$, i.e. every polynomial $Q \in \mathcal{P}_{d}(\Boo^{k}, \, G)$ can be expressed as a unique linear combination of monomials from $\mathsf{B}_{d}$ (with coefficients from $G$). Fix some total orders on $\mathcal{S}$ and $\mathsf{B}_{d}$.

Construct the matrix $M$ of dimensions $|\mathsf{B}_{d}| \times |\mathcal{S}|$ as follows: The rows are indexed by monomials in $\mathsf{B}_{d}$ and the columns are indexed by points in $\mathcal{S}$. For $1 \leq i \leq |\mathsf{B}_{d}|$ and $1 \leq j \leq |\mathcal{S}|$, $M[i, j]$ is equal to $m(\mathbf{u})$, where $m$ is the $i^{th}$ monomial in $\mathsf{B}_{d}$ and $\mathbf{u}$ is the $j^{th}$ point in $\mathcal{S}$. In other words, the $j^{th}$ column of $M$ denotes the vector whose entries are the evaluation of all the monomials in the spanning set $\mathsf{B}_{d}$ of the $j^{th}$ point in $\mathcal{S}$.

We will first prove \Cref{claim:matrix-integral-solution} and later show that \Cref{claim:matrix-integral-solution} is enough to prove \Cref{thm:weight-interpolating}.

\begin{claim}\label{claim:matrix-integral-solution}
Let $M$ be the matrix of dimensions $|\mathsf{B}_{d}| \times |\mathcal{S}|$ as described above. Define $\bm{\beta} \in \mathbb{Z}^{|\mathsf{B}_{d}|}$ as follows: For $1 \leq i \leq |\mathsf{B}_{d}|$, $\beta_{i}$ is equal to $m(\mathbf{b})$, where $m$ is the $i^{th}$ monomial of $\mathsf{B}_{d}$. There exists an integral vector $\mathbf{c} = (c_{1}, \ldots, c_{|\mathcal{S}|})$ such that the following equation is satisfied:
\begin{equation}\label{eqn:matrix-integral-solution}
    M \mathbf{c} \, = \, \bm{\beta}
\end{equation}
Equivalently, there exists an integral vector $\mathbf{c} = (c_{1}, \ldots, c_{|\mathcal{S}|})$ such that for every monomial $m \in \mathsf{B}_{d}$,
\begin{align*}
    \sum_{\mathbf{u} \in \mathcal{S}} c_{\mathbf{u}} m(\mathbf{u}) \, = \, m(\mathbf{b}).
\end{align*}
\end{claim}
\begin{proof}
To prove the existence of an integral vector $\mathbf{c}$, we need the following lemma.
\begin{lemma}\label{lemma:schriver-integral}
\cite[Corollary 4.1.a]{Schrijver} Let $A$ be a rational matrix and let $\mathbf{a}$ be a rational vector. Then the system $A \mathbf{x} = \mathbf{a}$ has an integral solution $\mathbf{x}$ if and only if for every row vector $\mathbf{y}$ for which $\mathbf{y} A$ is integral, $\mathbf{y} \mathbf{a}$ is an integer.
\end{lemma}
\Cref{lemma:schriver-integral} says that to show the existence of an integral solution $\mathbf{c}$ to \Cref{eqn:matrix-integral-solution}, it is equivalent to show that for every rational row vector $\mathbf{y} \in \mathbb{Q}^{|\mathsf{B}_{d}|}$ for which $\mathbf{y} M$ is integral, $\mathbf{y} \bm{\beta}$ is an integer.

Consider any $\mathbf{y} \in \mathbb{Q}^{|\mathsf{B}_{d}|}$ for which $\mathbf{y} M$ is integral. Let $H $ denote the quotient group $ \mathbb{Q}/\mathbb{Z}$, which can be identified with rational numbers in $[0,1)$ where addition is carried out modulo $1$. Define $\mathbf{y}'$ to be the image of $\mathbf{y}$ in $H$ under the natural projection from $\mathbb{Q}$ to $H$, i.e. for each coordinate $1 \leq i \leq |\mathsf{B}_{d}|$, $y'_{i} := y_{i} - \lfloor y_{i} \rfloor$. 

In what follows, we are still treating the entries of $\bm{\beta}$ and $M$ as integers, and thus it makes sense to multiply these entries with the entries of $\mathbf{y}'$ to get elements of $H$. The hypothesis that $\mathbf{y} M$ is integral is equivalent to saying that over the group $H$, $\mathbf{y}' M = 0^{|\mathcal{S}|}\in H^{|\mathcal{S}|}$ is the all-zeroes vector. Similarly, showing that $\mathbf{y}\bm{\beta}$ is an integer is equivalent to showing that $\mathbf{y}' \bm{\beta}$ is $0$. 

Assume for the sake of contradiction that $\mathbf{y}' \bm{\beta}$ is non-zero. Let $Q$ be the polynomial in $\mathcal{P}_{d}(\Boo^{k}, \, H)$ whose coefficient vector is $\mathbf{y}'$, i.e.
\begin{align*}
    Q(\mathbf{x}) := \sum_{m \in \mathsf{B}_{d}} y'_{m} m
\end{align*}
The hypothesis $\mathbf{y}' M = 0$ means that the polynomial $Q$ vanishes on $\mathcal{S}$. On the other hand, $\mathbf{y}' \bm{\beta} \in (0,1)$ means $Q(\mathbf{b})$ is non-zero, i.e. $Q$ is a non-zero polynomial. So we have a non-zero polynomial $Q$ that vanishes on the set $\mathcal{S}$, which contradicts \Cref{thm:weight-interpolating}. Hence $\mathbf{y}' \bm{\beta} = 0,$ implying that $\mathbf{y} \bm{\beta}$ is an integer. 

As $\mathbf{y}$ is an arbitrary row vector for which $\mathbf{y} M$ is integral, using \Cref{lemma:schriver-integral} we get the existence of an integral solution to \Cref{eqn:matrix-integral-solution}. This finishes the proof of \Cref{claim:matrix-integral-solution}.
\end{proof}
Finally, we argue that \Cref{claim:matrix-integral-solution} is sufficient to finish the proof of \Cref{thm:weight-interpolating}. This is essentially because $\mathcal{P}_{d}(\Boo^{k}, \, G)$ is spanned by $\mathsf{B}_{d}$. Consider any polynomial $Q \in \mathcal{P}_{d}(\Boo^{k}, \, G)$. There exists coefficients $\alpha_{1}, \ldots, \alpha_{|\mathsf{B}_{d}|}$ such that
\begin{align*}
    Q(\mathbf{x}) \, = \, \sum_{m \in \mathsf{B}_{d}} \alpha_{m} m
\end{align*}
Let $c_{1},\ldots,c_{|\mathcal{S}|}$ be the coefficients from the above claim. Then we have,
\begin{align*}
    \sum_{\mathbf{u} \in \mathcal{S}} c_{\mathbf{u}} Q(\mathbf{u}) \, = \, \sum_{\mathbf{u} \in \mathcal{S}} c_{\mathbf{u}} \; \sum_{m \in \mathsf{B}_{d}} \alpha_{m} m(\mathbf{u}) \, = \, \sum_{m \in \mathsf{B}_{d}} \alpha_{m} \; \sum_{\mathbf{u} \in \mathcal{S}} c_{\mathbf{u}} m(\mathbf{u})
\end{align*}
Since $M \mathbf{c} = \bm{\beta}$, for every $m \in \mathsf{B}_{d}$, we have,
\begin{align*}
    \sum_{\mathbf{u} \in \mathcal{S}} c_{\mathbf{u}} m(\mathbf{u}) \, = \, m(\mathbf{b})
\end{align*}
Thus we get,
\begin{align*}
    \sum_{\mathbf{u} \in \mathcal{S}} c_{\mathbf{u}} Q(\mathbf{u}) \, = \, \sum_{m \in \mathsf{B}_{d}} \alpha_{m} \; \sum_{\mathbf{u} \in \mathcal{S}} c_{\mathbf{u}} m(\mathbf{u}) \, = \, \sum_{m \in \mathsf{B}_{d}} \alpha_{m} m(\mathbf{b}) \, = \, Q(\mathbf{b})
\end{align*}
This finishes the proof of \Cref{thm:weight-interpolating}.
\end{proof}

\subsection{Error close to half the minimum distance (Proof of~\Cref{thm:uniquedegd})}\label{subsec:error-reduction-close-radius}
In this subsection, we explain the second step towards proving \Cref{thm:uniquedegd}. In the previous subsection, we described a local correction algorithm for $\mathcal{P}_{d}$ when the error is $1/\bigO_d((\log n)^{d})$. We now want to locally correct degree-$d$ polynomials when the error is close to the unique decoding radius, which is $1/2^{d+1}$.


\paragraph{}Suppose we have oracle access to a function $f$ that is $(1/2^{d+1} - \varepsilon)$-close to $\mathcal{P}_{d}(\Boo^{k}, \, G)$ for some constant $\varepsilon > 0$ and let $P$ be the unique polynomial in $\mathcal{P}_{d}(\Boo^{k}, \, G)$ such that $\delta(f, P) \leq (1/2^{d+1} - \varepsilon)$. The idea is to design a randomized algorithm $\mathcal{A}$ that has oracle access to the function $f$ and returns a probabilistic oracle $\mathcal{A}^{f}$ such that $\delta(\mathcal{A}^{f}, P) < 1/\bigO((\log n)^{d})$, with high probability. The algorithm $\mathcal{A}$ will be referred to as \textit{error reduction} algorithm. More specifically, the error reduction algorithm $\mathcal{A}$ will have two subroutines as follows:
\begin{enumerate}
    \item There is a randomized algorithm $\mathcal{A}_{1}$ that reduces the error from $(1/2^{d+1} - \varepsilon)$ down to $1/1000$.
    \item There is a randomized algorithm $\mathcal{A}_{2}$ that reduces the error from $1/1000$ down to $1/\bigO((\log n)^{d})$.
\end{enumerate}

\cite{ABPSS24-ECCC} gave the error reduction algorithms $\mathcal{A}_{1}$ and $\mathcal{A}_{2}$, which we state below.

\begin{lemma}[Error reduction for error close to half the minimum distance]\label{lem:error-reduction-main}
\cite[Lemma 3.13]{ABPSS24-ECCC} Fix any Abelian group $G$ and a positive integer $d$. For any $\eta_{1}, \delta$, where $\eta < \delta$ and $\delta < 1/2^{d+1} - \varepsilon$ for $\varepsilon > 0$, there exists a randomized algorithm $\mathcal{A}_{1}$ with the following properties:\newline
Let $f: \Boo^{n} \to G$ be a function and let $P: \Boo^{n} \to G$ be a degree $d$ polynomial such that $\delta(f,P) \leq \delta$, and let $\mathcal{A}_{1}^{f}$ denotes that $\mathcal{A}$ has oracle access to $f$, then
\begin{align*}
    \Pr[\delta(\mathcal{A}_{1}^{f}, P) > \eta_{1}] < 1/20,
\end{align*}
where the above probability is over the internal randomness of $\mathcal{A}_{1}$, and for every $\mathbf{x} \in \Boo^{n}$, $\mathcal{A}_{1}^{f}$ makes $2^{k}$ queries to $f$, where $k = \poly(\frac{1}{\varepsilon},\frac{1}{\eta_{1}})$.
\end{lemma}

\begin{lemma}[Error reduction for constant error]\label{lem:error-reduction-small-constant-main}
\cite[Lemma 3.8]{ABPSS24-ECCC} Fix any Abelian group $G$ and a positive integer $d$. The following holds for $\eta_1 < 1/2^{\bigO(d)}$ and $K = 2^{\bigO(d)}$ where the $\bigO(\cdot)$ hides a large enough absolute constant.

For any $\eta_2< \eta_1$, there exists a randomized algorithm $\mathcal{A}_{2}$ with the following properties: Let $f: \Boo^{n} \to G$ be a function and let $P: \Boo^{n} \to G$ be a degree-$d$ polynomial such that $\delta(f,P) \leq \delta$, and let $\mathcal{A}_{2}^{f}$ denotes that $\mathcal{A}_{2}$ has oracle access to $f$, then
\begin{align*}
    \Pr[\delta(\mathcal{A}_{2}^{f}, P) > \eta_{2}] < 1/20,
\end{align*}
where the above probability is over the internal randomness of $\mathcal{A}_{2}^{f}$. Further, for every $\mathbf{x} \in \Boo^{n}$, $\mathcal{A}_{2}^{f}$ makes $K^{T}$ queries to $f$ and $T = \bigO\paren{  \log\paren{ \dfrac{\log(1/\eta_{2})}{\log(1/\delta)} }  }  $.
\end{lemma}

Using \Cref{lem:error-reduction-main} and \Cref{lem:error-reduction-small-constant-main} along with \Cref{thm:uniquedeg-d-smallerror}, we get \Cref{thm:uniquedegd}. We restate \Cref{thm:uniquedegd} and finish the proof.

\uniquedegd*

\begin{proof}[Proof of \Cref{thm:uniquedegd}]
Let $f$ be a function with oracle access that is $\delta$-close to a degree-$d$ polynomial $P \in \mathcal{P}_{d}(\Boo^{k}, \, G)$.
\begin{enumerate}
    \item \Cref{lem:error-reduction-main} (with $\eta_{1} = 1/2^{\bigO(d)}$ chosen to satisfy the hypothesis of \Cref{lem:error-reduction-small-constant-main}) yields an oracle $\mathcal{A}_{1}^{f}$ that makes $\bigO_{\varepsilon}(1)$ queries to $f$ and is $\eta_{1}$-close to $P$ with probability at least $19/20$. We fix the randomness of $\mathcal{A}_1^f$ so that this holds.
    \item With probability at least $19/20$, we have oracle access to a function $\mathcal{A}_{1}^{f}$ that is $1/2^{\bigO(d)}$-close to $P$. Let $g := \mathcal{A}_{1}^{f}$. \Cref{lem:error-reduction-small-constant-main} (with $\eta_{2} = 1/\bigO((\log n)^{d})$) yields a probabilistic oracle $\mathcal{A}_{2}^{g}$ that makes $\mathrm{poly}(\log \log n, d)$ queries to $g$ and is $\eta_{2}$-close to $P$ with probability at least $19/20$. We again fix the randomness of $\mathcal{A}_2^g$ so that this holds.
\end{enumerate}
In other words, we have oracle access to an oracle $\mathcal{A}_{2}^{g}$ that makes $\mathrm{poly}(\log \log n, d) \cdot \bigO_{\varepsilon}(1)$ queries to $f$ and is $\eta_{2}$-close to $P$ with probability at least $9/10$. We now apply the local correction algorithm from \Cref{thm:uniquedeg-d-smallerror} with oracle access to $\mathcal{A}_{2}^{g}$ to get a local correction algorithm for $\mathcal{P}_{d}$ with $\delta = 1/2^{d+1} - \varepsilon$ for any constant $\varepsilon > 0$ and $q = \Tilde{\bigO}_\varepsilon((\log n)^{d})$.
\end{proof}

\begin{remark}
    \label{rem:list-correction}
    It should be noted that \Cref{thm:uniquedegd} also follows from the theorem on local list-correction~\Cref{thm:listdecoding} along with known results on \emph{locally testing low-degree polyomials}~\cite{bafna2017local}. However, we state this theorem separately for multiple reasons. Firstly, a weaker form of this theorem is required for the results on local list-correction. Secondly, the above result is natural and this gives a simpler proof of this than the one outlined above. And finally, this proof yields a better dependence on the degree parameter $d$.
\end{remark}

Having finished the proof of~\Cref{thm:uniquedegd}, in the following subsection, we now prove~\Cref{thm:const-torsion} by giving a local correction algorithm with improved query complexity for groups of small exponent (see~\Cref{sec:prelims} for a definition of exponent of a group).

\subsection{Local correction for groups of constant exponent}
\label{sec:const-torsion}

In this subsection, we show that we can bring down the query complexity of local correction from $\widetilde{\bigO}_d((\log n)^d)$ to a constant (i.e., independent of $n$) when $G$ is an Abelian torsion group of constant exponent. More specifically, we prove~\Cref{thm:const-torsion} from the introduction.

\consttorsion*




We note that to prove~\Cref{thm:const-torsion} for every $\delta = \frac{1}{2^{d+1}}-\varepsilon$, it suffices to show the following lemma for {\em some} constant $\delta=\Omega_{M,d}(1)$, since the error reduction steps from~\Cref{subsec:error-reduction-close-radius} can be applied without change.

\begin{lemma}
    \label{lem:const-torsion-sub} For every Abelian torsion group $G$ of exponent $M$, the family $\cP_d$ has a $(\delta, q)$-local correction algorithm for some $\delta = \Omega_{M,d}(1)$ and $q=\bigO_{M,d}(1)$. 
\end{lemma}

The proof of the above lemma proceeds in an identical manner to the analysis of~\cite{bafna2017local} where the authors show this when $G$ is the underlying group of a {\em field} of constant characteristic. They do this by using Lucas' theorem which gives a criterion for a binomial coefficient to be divisible by a given prime. To handle the more general case of groups of constant exponent, we instead make use of Kummer's theorem which may be thought of as an analog of Lucas' theorem for prime powers. We state Kummer's theorem below, where the notation $S_p(n)$ denotes the sum of the digits of $n$ when written in base $p$.

\begin{theorem}[Kummer's theorem~\cite{kummer1852}] 
\label{thm:kummer}
    Let $p \in \mathbb{N}$ be a prime. Then for any integers $a \ge b \ge 0$, the largest power of $p$ that divides ${a \choose b}$ is equal to $\frac{S_p(b)+S_p(a-b)-S_p(a)}{p-1}$.
\end{theorem}

We are now ready to prove~\Cref{lem:const-torsion-sub}.

\begin{proof}[Proof of~\Cref{lem:const-torsion-sub}]
    Let $M=\prod_{j=1}^{\ell} p_j^{r_j}$ be the prime factorization of $M$ (so $\ell \le \log M$). For each $j\in [\ell]$, let $s_j \in \mathbb{N}$ be the smallest integer such that $p_j^{r_js_j} > d$. Then, we choose $k=\prod_{j\in [\ell]} p_j^{3r_js_j}$. Note that $p_j^{r_j(s_j-1)} \le d$ and hence $k \le \prod_{j\in [\ell]}(dp_j^{r_j})^3 \le d^{3\ell}M^3=\bigO_{M,d}(1) $. We set $\delta = \frac{1}{4{2k \choose k}}=\Omega_{M,d}(1)$. 
    
    We claim that the algorithm below (\Cref{algo:constant-torsion}) is the desired local corrector. For a given point ${\bf a} \in \{0,1\}^n$, it queries $f$ and outputs $P({\bf a})$ with probability at least $3/4$, where $P \in \mathcal{P}_d$ is the unique degree-$d$ polynomial such that $\delta(f,P) \le \delta$. It is similar to~\Cref{algo:local-correction-sub-constant-error} in that it samples a random subcube $C_{{\bf a}, h}$ passing through ${\bf a}$ and queries it, but the crucial difference now is that we use a different interpolating set in the last step (as opposed to the ``weight balanced interpolating set'' of~\Cref{thm:weight-interpolating}). In particular, we will prove following claim. 

    \begin{claim}
    \label{clm:interpolate-slice}
        There exist integers $c_{{\bf b}}$ for ${\bf b} \in {[2k] \choose k}$ such that for every degree-$d$ polynomial $Q({\bf y}) \in \mathcal{P}_d(\{0,1\}^{2k}, G)$, we have that 
        \begin{align}Q(0^{2k}) = \sum_{{\bf b}\in {[2k]\choose k}} c_{\bf b} \cdot Q({\bf b}).\label{eqn:linear}\end{align}
    \end{claim}
    
\begin{algobox}
\begin{algorithm}[H]
\DontPrintSemicolon
\KwIn{$f(x_{1},\ldots,x_{n})$, $\mathbf{a} \in \Boo^{n}$, $\delta = \frac{1}{4{2k \choose k}}$}
\vspace{2mm}
Sample a uniformly random function $h: [n] \to [2k]$\label{algoline:local-correction-sample-hash} \\
\vspace{2mm}
$g(y_{1},\ldots,y_{2k}) \leftarrow f(x_{1},\ldots,x_{n})|_{C_{\mathbf{a}, h}}$ \;
\vspace{2mm}
Let $(c_{{\bf b}})_{{\bf b} \in {[2k] \choose k}}$ be the integral coefficients  given by \Cref{clm:interpolate-slice}. \;
\vspace{2mm}
\Return{$\sum_{\mathbf{b} \in {[2k] \choose k}} \; c_{\mathbf{b}} \cdot g(\mathbf{b}) $}
\caption{Local correction algorithm for sub-constant error}
\label{algo:constant-torsion}
\end{algorithm}
\end{algobox}

Assuming the correctness of~\Cref{clm:interpolate-slice}, we shall now finish the proof of~\Cref{lem:const-torsion-sub}. 

\paragraph{Queries:} The local corrector makes $ {2k \choose k} = \bigO_{M,d}(1)$ queries since to get the value of $g({\bf b})$ for some ${\bf b} \in {[2k] \choose k}$, we only need to know $f(x({\bf b}))$ under the mapping $h$.

\paragraph{Correctness:} For every ${\bf b} \in {[2k] \choose k}$, we note that the corresponding query point $x({\bf b}) \in \{0,1\}^n$ is uniformly distributed since the map $h$ used in~\Cref{algo:constant-torsion} is uniformly random and ${\bf b}$ has equal number of zeroes and ones. Hence, with probability at least $1-\delta \cdot {2k \choose k} \ge 3/4$, $g({\bf b}) = Q({\bf b})$ for all ${\bf b} \in {[2k] \choose k}$, where $Q \in \mathcal{P}_d(\{0,1\}^{2k}, G)$ is the restriction of $P$ on the subcube $C_{{\bf a}, h}$. Hence, by~\Cref{clm:interpolate-slice}, the outputted value equals $Q(0^{2k}) = P({\bf a})$ with probability at least 3/4. 
\end{proof}

It remains to prove~\Cref{clm:interpolate-slice}. For every ${\bf b} \in {[2k] \choose k}$, we set $c_{\bf b} = 0$ if ${\bf b}$ contains a 1 in any of the last $k-d$ coordinates and we set $c_{\bf b} = A$ otherwise, where $A \in \Z$ will be decided later. Recall that $M=\prod_{j\in[\ell]} p_j^{r_j}$ and $k=\prod_{j\in[\ell]} p_j^{3r_js_j}$, and we have that $p_j^{r_j s_j} > d \ge p_j^{r_j(s_j-1)}$ for all $j\in [\ell]$. By linearity, it suffices to show~\eqref{eqn:linear} for $Q({\bf y})$ of the form $g\cdot \prod_{j\in I} y_j$ for all $I \in {[2k] \choose \le d}$ and $g\in G$. According to our assignment of $c_{\bf b}$, it is clear that~\eqref{eqn:linear} holds true (LHS = RHS = 0) if $I$ contains any of the last $k-d$ coordinates. Otherwise, we have that $I \subseteq {[k+d] \choose \le d}$. If $I=\emptyset$, we have $Q(0^{2k})=g$ and $\sum_{{\bf b}\in {[2k] \choose k}} c_{\bf b}\cdot Q({\bf b}) = {k+d \choose k}A\cdot g$. On the other hand, if $|I|=i\ge 1$, we have $Q(0^{2k}) = 0$ and $\sum_{{\bf b}\in {[2k] \choose k}} c_{\bf b}\cdot Q({\bf b}) = {k+d-i \choose k-i}A\cdot g$ since every non-zero term must have $b_j=1$ for all $j\in I$. Hence, it suffices to find an integer $A$ satisfying the following two conditions:
\begin{align*}
     g & = {k+d \choose k} A \cdot g, \text{~for all~} g\in G, \text{~and~}\\
     0 & = {k+d-i \choose k-i} A \cdot g, \text{~for all~} g\in G \text{~and~} i\in [d].
\end{align*}

Since the order of every element $g$ divides the exponent $M$ of the group, for the above two conditions to hold, it suffices if for all $j\in [\ell]$ and $i\in [d]$, $p_j$ does not divide ${k+d \choose k}$ and that $p_j^{r_j}$ divides ${k+d-i \choose k-i}$. Then we can take $A$ to be any integer such that $A {k+d \choose k} + A' M =1$ for some integer $A'$ (such $A$ and $A'$ are guaranteed to exist as $M$ and ${k+d \choose k}$ are coprime). The rest of the proof is dedicated to verifying these divisibility constraints hold.

\begin{itemize}
    \item \textbf{$p_j$ does not divide ${k+d \choose k}$:} We will represent all the numbers $k,d,i$ etc.~in base $p_j$. We note that the last $r_j s_j$ digits of $k$ are zeroes since $p_j^{r_j}$ divides $k$. Furthermore, since $d < p_j^{r_j s_j}$, all the digits of $d$ except the last $r_j s_j$ many are zeroes. Hence, the sum of digits of $k+d$ is equal to the sum of the digits of $k$ and $d$ combined. That is, $S_{p_j}(k)+S_{p_j}(d)-S_{p_j}(k+d)=0$. Applying Kummer's theorem (\Cref{thm:kummer}) now finishes the proof. 
    
    \item \textbf{$p_j^{r_j}$ divides ${k+d-i \choose k-i}$:} By Kummer's theorem (\Cref{thm:kummer}), it suffices to show that $$\frac{S_{p_j}(d)+S_{p_j}(k-i)-S_{p_j}(k+d-i)}{p_j-1} \ge r_j.$$ We note that $S_{p_j}(k+d-i) = S_{p_j}(k)+S_{p_j}(d-i)$ by the same argument as the above paragraph. In addition, we have the trivial bounds $S_{p_j}(d) \ge 1$ and $S_{p_j}(d-i) \le (p_j-1)r_j s_j$. Finally, we give a lower bound for $S_{p_j}(k-i)$. Since $k$ has at least $3r_j s_j$ trailing zeroes, we get that $S_{p_j}(k-1) \ge S_{p_j}(k)+3r_js_j(p_j-1)-1$. But observe that $S_{p_j}(k-i) = S_{p_j}((k-1)-(i-1)) = S_{p_j}(k-1) - S_{p_j}(i-1)$ since the number of trailing $(p_j-1)$'s of $k-1$ exceeds the total number of (non-zero) digits of $(i-1)$. Therefore, we get
    \begin{align*}
        {S_{p_j}(d)+S_{p_j}(k-i)-S_{p_j}(k+d-i)} & \ge 1+S_{p_j}(k-1)-S_{p_j}(i-1)-S_{p_j}(k)-S_{p_j}(d-i)\\
        & \ge 1+(3r_js_j({p_j}-1)-1)-(p_j-1)r_js_j-(p_j-1)r_js_j\\
        & \ge r_js_j(p_j-1)\\
        & \ge r_j(p_j-1).
    \end{align*}
\end{itemize}

This finishes the proof of~\Cref{clm:interpolate-slice} and hence~\Cref{lem:const-torsion-sub} and~\Cref{thm:const-torsion}.

\section{Combinatorial list-decoding bound}

In this section, we are going to prove the following theorem. 

\comblistdegd*

In other words, we will show that for any function $f: \Boo^{n} \to G$, the number of degree-$d$ polynomials that are $(1/2^d-\varepsilon)$-close to $f$ is $\exp(\bigO_d(1/\varepsilon)^{\bigO(d)})$.

We use the following result of~\cite{ABPSS24-ECCC} which gives a naive double-exponential upper bound on the list size. While~\cite{ABPSS24-ECCC} prove it for linear polynomials, the same proof extends to higher degree without much change. 

\begin{claim}[\cite{ABPSS24-ECCC}, Claim 4.1]
\label{clm:finite-list} For any function $f:\{0,1\}^n \to G$, the number of degree-$d$ polynomials that are $(1/2^d-\varepsilon)$-close to $f$ is at most $2^{2^n}$. 
\end{claim}
We will subsequently improve the above bound to something independent of $n$, but to do that, we will need this naive bound.
Furthermore, using a result from previous work~\cite[Claim 4.2]{ABPSS24-ECCC},\footnote{Though this result is only stated for degree $1$ in~\cite{ABPSS24-ECCC}, it works without any change for all degrees.} we know that proving~\Cref{thm:comblistdegd} for finite Abelian $G$ implies the same theorem for all Abelian $G$. Hence, we will assume that $G$ is a finite Abelian group. By using the structure theorem of finite Abelian groups, we can decompose $G$ as 
$$G \cong G_1 \times G_2,$$
where $G_1$ is the product of finitely many $p$-groups where each $p$ is a prime number that is at least $p_0$ (for some appropriate choice of $p_0=p_0(d)$ to be fixed later) and $G_2$ is the product of finitely many $p$-groups where each $p$ is a prime less than $p_0$. We provide upper bounds for list-decoding over $G_1$ and $G_2$ separately and combine the two bounds to get a final bound on the list size over $G$. We state the upper bounds formally below, where we use the notation $\mathsf{List}^f_\varepsilon$ to denote the set of degree-$d$ polynomials that are $(1/2^d-\varepsilon)$-close to $f$.

\begin{theorem}[Combinatorial bound for a product of $p$-groups where each $p\geq p_0$]\label{thm:largechar}
Let $d\ge 1$ and $G$ be a product of finitely many $p$-groups, where each $p \geq p_0 = 2^{2^{\alpha d^3}}$ for a sufficiently large constant $\alpha$. Then for every function $f: \Boo^{n} \to G$, we have  $|\mathsf{List}_{\varepsilon}^{f}| \leq (1/\varepsilon)^{2^{2^{\bigO(d^3)}}}$.

In particular, the list size is polynomial in $1/\varepsilon$ for a constant $d$.\\
\end{theorem}

\noindent
We now state the combinatorial bound for the second case.

\begin{restatable}[Combinatorial bound for a product of $p$-groups where each $p<p_0$]{theorem}{combsmallchar}\label{thm:smallchar}
Let $d\ge 1$ and $G$ be a product of finitely many $p$-groups, where each $p \le 2^{2^{\bigO(d^2)}}$. Then for every function $f: \Boo^{n} \to G$ we have, $|\mathsf{List}_{\varepsilon}^{f}| \leq \exp(\bigO_d(1/\varepsilon)^{\bigO(d)})$.
\end{restatable}

Assuming \Cref{thm:largechar} and \Cref{thm:smallchar}, we immediately get \Cref{thm:comblistdegd} because for any function $f$ with co-domain $G=G_1 \times G_2$ can be written as $f = (f_1,f_2)$ where $f_1$ has co-domain $G_1$ and $f_2$ has co-domain $G_2$. Furthermore, if $P = (P_{1},P_{2}) \in \mathsf{List}_{\varepsilon}^{f}$, then for each $i \in [2]$, $P_{i}$ must be in $\mathsf{List}_{\varepsilon}^{f_{i}}$. 

We move on to proving the above two theorems in the next two subsections respectively.

\subsection{Combinatorial bound for a product of $p$-groups ($p\geq p_0$)}\label{subsec:large-groups}

Our proof of~\Cref{thm:largechar} builds upon some of the ideas of the combinatorial bound of~\cite{ABPSS24-ECCC} (Theorem 4.4) which handles degree $d=1$. However, there are various places where higher degree polynomials are not as well-behaved and need more complicated analysis. Indeed the anti-concentration bound (see \Cref{lem:anticonc} below) we need is more involved and can be of independent interest.   
Before the full proof, we now give  a rough outline of the proof of~\Cref{thm:largechar}. It can be divided into the following two parts.

\begin{enumerate}
    \item \textbf{Anti-concentration of non-sparse polynomials}: Suppose there are two degree-$d$ polynomials $P_1$ and $P_2$, both $(1/2^d-\varepsilon)$-close to a function $f$. Then they must agree with each other on a sufficiently large fraction of the domain. Indeed, $\delta(P_1,P_2) \le \delta(f,P_1)+\delta(f,P_2) < 1/2^d + 1/2^d=1/2^{d-1}$. Hence, the density of the zeroes of $P=P_1-P_2$ in $\{0,1\}^n$ is greater than $1-1/2^{d-1}$. Suppose that it is at least $1-1/2^{d-1}+c$ for some constant $c$.\footnote{We will show the existence of such a $c$ in the formal proof.} Our main idea here is that this cannot happen for polynomials $P$ with many monomials. In particular, we show that if $P$ has sufficiently large sparsity (defined as the number of non-zero monomials), then this fraction is less than $1-1/2^{d-1}+c$. This allows us to reduce the combinatorial bound to the case of just counting polynomials of ``small'' sparsity (by a small blow-up in the list size). We expand more on this in the second step. Going back to showing the anti-concentration bound itself, we will prove the following.


        \begin{thmbox}
        \begin{lemma}[Anti-concentration bound for non-sparse polynomials]
        \label{lem:anticonc}
        For all positive integers $d,s$ and for every Abelian group $G$ in which all the non-zero elements have order greater than $(s+1)!$, the following holds: For every degree-$d$ polynomial $Q({\bf x})$ over $G$ of sparsity at least $s$, we have 
        $$\Pr_{{\bf x}\sim \{0,1\}^n}[Q({\bf x})\ne 0] \ge 1/2^{d-1} - {2^{\bigO(d^3)}}/\sqrt{s}.$$
        \end{lemma}
        \end{thmbox}
    
    When $d=1$,~\cite{ABPSS24-ECCC} show an anti-concentration lemma along the lines of Littlewood and Offord~\cite{littlewood-offord} which bounds the density of the zeroes of non-sparse linear polynomials by an arbitrarily small constant (as long as the sparsity is large enough). However, even for $d=2$, this problem is somewhat subtle. For example $P({\bf x})=x_1\cdot P'(x_2,\dots,x_n)$ (where $P'$ is of degree $1$ and large sparsity), is always $0$ when $x_1=0$. In other words, we cannot hope to bound the density of the zeroes of non-sparse degree-2 polynomials by an arbitrarily small constant. Nevertheless, we can still argue that it cannot be much larger than $1/2$ (i.e., $1-1/2^{d-1}$ when $d=2$). For this particular example of $P=x_1 \cdot P'$, we note that when $x_1=1$, we can defer to the $d=1$ case to bound the fraction of roots by a small constant (say $c$) and when $x_1=0$, $P({\bf x})$ is always zero; thus the fraction of roots of $P$ over $\{0,1\}^n$ is less than $1/2+c$, which is what we wanted to prove. We formalize this for general polynomials (of large sparsity) and general degree $d$. In particular, we rely on an anti-concentration bound of Meka, Nguyen and Vu~\cite{MNV} when $P$ has many disjoint (non-zero) monomials of degree $d$ -- in other words, a ``$d$-matching''. Otherwise, there has to be a small vertex cover among the monomials and we use this to reduce to the case of smaller degree (and perform an induction on the degree). There is the further complication of the fact that~\cite{MNV} state their results only for polynomials over the reals whereas our goal is to also prove it over groups without elements of small order. However, we show that we can use a linear-algebraic argument to deduce the same bound for our setting by making use of the anti-concentration statement over the reals.
    \item \textbf{Counting sparse polynomials}: We want to show that the number of degree-$d$ polynomials $P_1,P_2,\dots,P_t$ of {\em constant sparsity} that are $(1/2^d-\varepsilon)$-close to a function $f$ is $\poly(1/\varepsilon)$. The case of $d=1$ was handled by~\cite{ABPSS24-ECCC} by reducing (at least implicitly) to the case of $P_1,\dots,P_t$ depending on {\em disjoint} sets of variables and uses the ``independence'' of such polynomials to get a bound on $t$. This part of the reduction is more involved for higher degrees. The reduction in~\cite{ABPSS24-ECCC} occurs by setting certain subsets of variables to constants. For linear polynomials,~\cite{ABPSS24-ECCC} has the advantage that setting one variable cannot make a polynomial zero (assuming it depends on at least two variables). However, even for $d=2$, we cannot afford to set variables to arbitrary constants. For example, $P({\bf x})=x_1\cdot (x_2 + 3x_3 - x_4)$ vanishes if we set $x_1=0$. We get around this by analyzing the structure of the polynomials and setting variables in two stages: in each stage we prove that the list size does not change too much. We defer the remaining details about these two stages (and the full argument) to~\Cref{sec:prune}.
\end{enumerate}



We now proceed with the proof of~\Cref{thm:largechar} with all the details. We will start by assuming the anti-concentration lemma (\Cref{lem:anticonc}) and deducing \Cref{thm:largechar} in \Cref{sec:prune}. We then show how to prove \Cref{lem:anticonc} in \Cref{sec:anticoncentration}.

\subsubsection{Pruning the list}
\label{sec:prune}
Roughly speaking, we first show that it suffices to bound the number of ``sparse polynomials'' in the list to get an upper bound on the total list size. From now on, we will use $\spars(P)$ to denote the number of monomials with non-zero coefficient in the polynomial $P.$

\paragraph{Reducing to counting sparse polynomials.}
Let $P_1,P_2,\dots,P_t$ be all the distinct degree-$d$ polynomials that are $(1/2^d-\varepsilon)$-close to $f$. We consider the following graph $G$ with vertex set $[t]$ and an edge between $i$ and $j$ if and only if $\spars(P_i-P_j) \le s_0$ where $s_0={2^{\bigO(d^3)}}$ is large enough so that the probability on the RHS of~\Cref{lem:anticonc} is at least $1/2^{d-1}-0.5/2^{2d}$.  Here we are using the fact for a sufficiently large constant $\alpha$, the order of all the non-zero elements of $G$ are greater than $p_0=2^{2^{\alpha d^3}} \ge (s_0+1)!$. 

We now show that $G$ cannot contain an independent set of size $\ell = 4^d$. Here, we are assuming $t > \ell$ as otherwise we are done. That is, without loss of generality, assume that the polynomials $P_1,P_2,\dots,P_\ell$ are such that $\spars(P_i-P_j) \ge s_0$ for all $i\ne j \in [\ell]$. We will use the following fact.

\begin{lemma}[e.g.~\cite{jukna2011extremal}, Lemma 2.1]
    Suppose $A_1,\dots,A_\ell \subseteq U$ are subsets each of size $r$ such that the pairwise intersections $A_i \cap A_j$ are of size at most $r'$ for all $i\ne j\in [\ell]$. Then the size of the union of the sets $\cup_{i=1}^\ell A_i$ is at least $r^2\ell/(r+(\ell-1)r')$.
\end{lemma}
We take $U$ to be $\{0,1\}^n$ and for $i\in [\ell]$, $A_i$ to be any subset of $\{{\bf x}~|~f({\bf x})=P_i({\bf x})\}$ of size $(1-1/2^d)2^n$ and apply the above lemma. We note that for any $i\ne j\in [\ell]$, $|A_i \cap A_j| \le r'$ for $r' = 2^n ((1-1/2^{d-1})+0.5/2^{2d})$ by applying~\Cref{lem:anticonc} for $Q=P_i-P_j$ since we have assumed that $\spars(P_i-P_j) \ge s_0$. Since $|\cup_{i=1}^\ell A_i| \le 2^n$, we obtain
$$2^n \ge \frac{((1-1/2^d)2^n)^2 \ell}{(1-1/2^d)2^n + (\ell -1)(1-1/2^{d-1}+0.5/2^{2d})2^n}.$$

Simplifying the above, we get that $\ell < 4^d$. Thus, there is no independent set of size $4^d$ in $G$, which in turn implies by Tur\'{a}n's theorem that there is at least one vertex $\nu\in [t]$ of $G$ with degree at least 
\begin{align}\label{eqn:e1}t'\ge t/4^d-1.\end{align} Let $\nu_1,\nu_2,\dots,\nu_{t'}$ be distinct neighbors of $\nu$. Then consider the polynomials $Q_i = P_{\nu_i}-P_{\nu}$ for $i\in [t']$. We note that $\spars(Q_i) \le s_0$ and $\delta(Q_i,f') \le 1/2^d -\varepsilon$ for $f'=f-P_{\nu}$. 

We now bucket the polynomials $Q_1,Q_2,\dots,Q_{t'}$ based on which subset of variables they depend on\footnote{We say that a polynomial $Q({\bf x})$ {\em depends} on $x_i$ if $x_i$ appears in at least one monomial that has non-zero coefficient in $Q$.}. Since the sparsity of each $Q_i$ is at most $s_0$, it must depend on at most $s_0 d$ variables. 

In the next paragraph, we bound the size of each bucket.

\paragraph{Counting sparse polynomials depending on the same set of variables.} We will show that the number of polynomials $Q$ that depend (only) on the variables $x_i$ for $i\in I$ for some fixed $I\in {[n] \choose \le s_0 d}$ such that $Q$ is $(1/2^d-\varepsilon)$-close to $f'$ is at most $(2/\varepsilon)2^{2^{s_0 d}}$. Suppose that $Q_1,Q_2,\dots,Q_{t''}$ are such polynomials over the variables indexed by $I$ and we want to show that $t''\le (2/\varepsilon)2^{2^{s_0 d}}$. Note that if $f'$ also depends only on the variables indexed by $I$, then we get the bound $2^{2^{s_0d}}$ by applying~\Cref{clm:finite-list}. However, in general, $f'$ can depend on variables outside $I$ and this results in an additional $2/\varepsilon$ factor. 

To make this precise, we define a function $f'':\{0,1\}^I\to G$ over just the variables in $I$ as $f''({\bf y}) = f'({\bf z})$ where ${\bf z}|_I = {\bf y}$ and ${\bf z}|_{I^c}$ is a uniformly random Boolean assignment. Let $\mathcal{X}_i$ be the indicator random variable for the event that $\delta(f'',Q_i) \le 1/2^{d}-\varepsilon/2$ for $i\in [t'']$.  Since $\delta(f',Q_i) \le 1/2^d - \varepsilon$, we conclude that with probability at least $\varepsilon/2$, it holds that $\delta(f'',Q_i) \le 1/2^d -\varepsilon/2$ i.e., $\Pr[\mathcal{X}_i=1] \ge \varepsilon/2$. By linearity of expectation, there must be a setting of ${\bf z}|_{I^c}$ such that for the corresponding $f''$, at least $(\varepsilon/2)t''$ indices $i\in [t'']$ exist such that $\delta(Q_i,f'') \le 1/2^d-\varepsilon/2$. But by~\Cref{clm:finite-list} this is at most $2^{2^{s_0 d}}$. Hence, we get $t'' \le (2/\varepsilon)2^{2^{s_0d}}$. 

Therefore, there must be at least $t'/((2/\varepsilon)2^{2^{s_0d}})$ non-empty buckets. Recall that we label each bucket by a subset of variables that the polynomials in that bucket depend on i.e., a subset of $[n]$ of size at most $k=s_0 d$. Thus there must be at least \begin{align}\label{eqn:e2}t'''\ge t'/((2/\varepsilon)2^{2^{s_0d}})\end{align} non-empty buckets that are labeled with a subset of $[n]$ of size at most $k$. 
\paragraph{Reducing to the case where the variable sets form a sunflower.} We now invoke the sunflower lemma with the sets $S_i$'s below being the labels of those non-empty buckets.

\begin{lemma}[\cite{erdos1960intersection}, Theorem 3]
    \label{lem:sunflower}
    Suppose $S_1,S_2,\dots,S_{t'''}$ are distinct subsets of $[n]$ of size at most $k$ with $t'''\ge k!(r-1)^k$ for some integer $r\ge 3$. Then there exists at least $r$ sets among $S_1,S_2,\dots,S_{t'''}$ that form a sunflower. That is, there exists distinct indices $i_1,i_2,\dots,i_r\in[t''']$ and $C \subseteq [n]$ (called the {\em core} of the sunflower) such that $C = S_{i_{j_1}} \cap S_{i_{j_2}}$ for all $j_1\ne j_2 \in [r]$, and $S_{i_j} \setminus C$ are non-empty for all $j\in [r]$.
\end{lemma}

Using the bound $k\le s_0 d$, we can take 
\begin{align}\label{eqn:e3}r=(t'''/(s_0d)!)^{1/s_0d}\end{align} in the above lemma. Now, since the corresponding buckets are non-empty, we can choose one (arbitrary) polynomial from each bucket that forms the sunflower -- suppose, without loss of generality,~that $Q_1,Q_2,\dots,Q_r$ are polynomials depending on variables indexed by subsets $S_1,S_2,\dots,S_r \subseteq [n]$ respectively such that the $S_i$'s form a sunflower, say with core $C\subseteq [n]$. Furthermore, recall that we have $\delta(f',Q_i) \le 1/2^d - \varepsilon$ for all $i\in [r]$.

\paragraph{Reducing to the case where the variable sets are pairwise disjoint.} The main idea is to set the variables in the core of the sunflower at random. However, we will need to do this in two steps. 

\begin{itemize}
\item For the sake of analysis, we shall relabel the variables indexed by $C$ by ${\bf z}=\{z_1,z_2,\dots,z_{n_0}\}$ arbitrarily (note that an empty core $C$ corresponds to $n_0=0$). We also relabel the variables indexed by $S_i \setminus C$ by ${\bf y}^{(i)}=\{{y}^{(i)}_1,y^{(i)}_2,\dots,y^{(i)}_{n_i}\}$ for $i\in [r]$ and some integers $n_i \ge 1$. Then, being a degree-$d$ polynomial, we can express each polynomial $Q_i({\bf x})=Q_i({\bf z},{\bf y}^{(i)})$ as a polynomial in the ${\bf y}^{(i)}$ variables with coefficients being some polynomials over ${\bf z}$ variables. That is,
\begin{align}
    Q_i({\bf z},{\bf y}^{(i)}) = \sum_{I \in {[n_i] \choose \le d}} (\mathbf{y}^{(i)})^I Q_{i,I}({\bf z}),
\end{align} where $(\mathbf{y}^{(i)})^I$ denotes the monomial corresponding to taking the product of variables indexed by $I$ and $Q_{i,I}$ is a polynomial of degree at most $d-|I|$. Since $Q_i$ depends on the variables ${\bf y}^{(i)}$, there must exist a \emph{non-empty} subset $I$ for which the polynomial $Q_{i,I}({\bf z})$ is non-zero. Thus, we can define the ${\bf y}$-degree of each $Q_i$ as the maximum size of an $I\in {[n_i] \choose \le d}$ such that $Q_{i,I}({\bf z})\ne 0$. By pigeonhole principle, there must be at  least $r/d$ indices $i\in [r]$ such that the corresponding ${\bf y}$-degrees of $Q_i$'s are identical -- we will denote this by $d'\in [d]$. Without loss of generality, we will assume that $Q_1,Q_2,\dots,Q_{r'}$ have ${\bf y}$-degree equal to $d'$ for some \begin{align}\label{eqn:e4}r'\ge r/d.\end{align} Now for each $i\in [r']$, let $m(i) \subseteq {C}$ denote an arbitrary non-zero monomial of $Q_{i,I}({\bf z})$ for an arbitrary $I\in {[n_i] \choose d'}$ such that $Q_{i,I}\ne 0$. Note that $m(i)$ can only take at most $2^{|C|}\le 2^{s_0d}$ values. Thus, by pigeonhole principle, there exist at least \begin{align}\label{eqn:e5}r''\ge r'/2^{s_0d}\end{align} indices for all which $m(i)=C'$ for some $C' \subseteq C$. Again, without loss of generality, we assume that $Q_1,Q_2,\dots,Q_{r''}$ are such polynomials. 

We now show that there exists an assignment to variables in $C\setminus C'$ such that \begin{align}\label{eqn:e6}r''' \ge r''\varepsilon/2\end{align} many of the respective restricted polynomials $Q_1',Q_2',\dots,Q'_{r'''}$ (here we are again assuming that the first $r'''$ polynomials satisfy this property) depend on at least one variable outside $C$ i.e., on some variable(s) in their respective ${\bf y}^{(i)}$. Furthermore, the distance from the corresponding restriction of $f''$ is small, i.e., $\delta(f''',Q'_i) \le 1/2^d-\varepsilon/2$ for all $i\in [r''']$. For a uniformly random assignment to the variables in $C\setminus C'$, we have that the expected distance $\delta(f''',Q'_i)$ is at most $1/2^d-\varepsilon$. Hence, with probability at least $\varepsilon/2$ over the choice of the random assignment, we have that $\delta(f''',Q'_i) \le 1/2^d - \varepsilon/2$ for each $i\in [r'']$. Now, by linearity of expectation, we conclude that there exists at least one assignment such that $\delta(f''',Q'_i) \le 1/2^d - \varepsilon/2$ holds for at least $r''\varepsilon/2$ many polynomials. 

\item We note that none of the restricted polynomials $Q_i'$ become zero or become identical to each other since for each $i$, there exists an $I\in {[n_i] \choose d'}$ such that $Q_{i,I}({\bf z})$ remains non-zero. This is because, by construction, it contains the monomial ${\bf z}^{C'}$ with a non-zero coefficient as a monomial of maximum degree, and setting the variables outside $C'$ cannot change the coefficient of this monomial. Let ${\bf z'} \subseteq {\bf z}$ be the variables indexed by $C'$ and let $Q'_{i,I}({\bf z'})$ denote the restriction of $Q_{i,I}({\bf z})$ for each $i\in[r''']$ and $I \in {[n_i]\choose \le d}$ corresponding to the above assignment to the variables in $C\setminus C'$. We have
$$Q'_i({\bf z'},{\bf y}^{(i)}) = \sum_{I \in {[n_i]\choose \le d}}({\bf y}^{(i)})^I Q'_{i,I}({\bf z'}).$$
Therefore, for each $i$, there exists an $I_i \in {[n_i]\choose d'}$ for which $Q'_{i,I_i}({\bf z'})$ has a non-zero coefficient for the product of all ${\bf z'}$ variables. Thus we conclude that for a random setting of ${\bf z'}={\bf a}$, with probability at least $1/2^{\deg(Q'_{i,I_i})}$ it holds that $Q'_{i,I_i}({\bf a}) \ne 0$. But notice that $\deg(Q'_{i,I_i}) \le d-d'$ since $\deg(Q_i')\le d$ and $|I_i|=d'$. Hence, by linearity of expectation, there exists some assignment ${\bf a}\in \{0,1\}^{C'}$ such that at least $r'''/2^{d-d'}$ of the corresponding restricted polynomials (which we will denote by $Q''_i({\bf y}^{(i)})=Q'_i({\bf z}'={\bf a},{\bf y}^{(i)})$) are non-constant. In particular, the coefficient of $({\bf y}^{(i)})^I$ is non-zero in $Q_i''$. 

We further claim that $\delta(f'''',Q''_i) \le 1/2^{d'}-\varepsilon/2$ where $f''''$ is the restriction of $f'''$ obtained by setting ${\bf z'}={\bf a}$. This follows by recalling that $\delta(f''',Q'_i)\le 1/2^d-\varepsilon/2$ and we are setting $d-d'$ variables to get $f''''$ and $Q''_i$. Indeed, $\delta(f'''',Q''_i) \le 2^{d-d'}(1/2^d - \varepsilon/2)\le 1/2^{d'}-\varepsilon/2$ for any choice of ${\bf a}$. In turn, this implies that for the setting ${\bf z}'={\bf a}$, we have $m\ge r'''/2^{d-d'}$ polynomials $Q''_1,Q_2'',\dots,Q''_m$ of degree at most $d'$ (recall that the ${\bf y}$-degree of the polynomials we are considering is $d'$) depending on pairwise disjoint sets of variables ${\bf y}^{(1)},{\bf y}^{(2)},\dots,{\bf y}^{(m)}$ respectively such that $\delta(f'''',Q''_i) \le 1/2^{d'}-\varepsilon/2$. 
\end{itemize}

\paragraph{Counting polynomials depending on pairwise disjoint variables.}
For the rest of the analysis, we treat $f''''$ and the polynomials $Q''_i$ as functions over all the $n$ variables even though we fixed the values of certain variables in the preceding steps. With this setup, note that $\Pr_{{\bf x}\in \{0,1\}^n, i\in [m]}\bigg[f''''({\bf x}) = Q''_i({\bf x})\bigg] \ge 1-1/2^{d'}+\varepsilon/2$, which implies that with probability at least $\varepsilon/4$ over the choice of ${\bf x}$, we have that $\Pr_{i\in [m]}[f''''({\bf x})=Q''_i({\bf x})] \ge 1-1/2^{d'}+\varepsilon/4$.  Note that for any such ${\bf x}\in \{0,1\}^n$, if $Q_i''(\mathbf{x}) = f'''(\mathbf{x})$, then $Q_i''(\mathbf{x}) = Q_j''(\mathbf{x})$ for at least $(1-1/2^{d'}+\varepsilon/4)m-1$ many $j\in [m]\setminus \{i\}.$ In particular, this gives  
\begin{align}
\label{eqn:prob-lower-bd}
    \Pr_{{\bf x}\in \{0,1\}^n}\bigg[\exists i\in [m]: \bigg|\{j~|~Q''_j({\bf x})=Q''_i({\bf x})\}\bigg|\ge (1-1/2^{d'}+\varepsilon/4)m-1\bigg] \ge \varepsilon/4.
\end{align}

We will now prove an upper bound on the same quantity above: fixing an arbitrary $i\in [m]$ and an arbitrary setting to the variables appearing in $Q_i,$ we see that since the polynomials depend on disjoint subsets of variables, the indicator random variables for the events $Q_j''({\bf x})=Q''_i({\bf x})$ for a uniformly random ${\bf x}\in \{0,1\}^n$ are Bernoulli random variables $\text{bern}(p_j)$ for some $p_j\le 1-1/2^{d'}$ (here we are applying the Schwartz-Zippel lemma~\Cref{thm:basic}) and across $j\in [m]\setminus \{i\}$. Thus, by a Chernoff bound, we have
\begin{align}\label{eqn:prob-upper-bd}\Pr_{{\bf x}\in \{0,1\}^n}\bigg[\bigg|\{j~|~Q''_j({\bf x})=Q''_i({\bf x})\}\bigg|\ge (1-1/2^{d'}+\varepsilon/4)(m-1)\bigg] \le \exp(-\Omega(\varepsilon^2 m)).\end{align}
Applying a union bound over $i\in [m]$ for~\eqref{eqn:prob-upper-bd} and combining with~\eqref{eqn:prob-lower-bd}, we get 
$\varepsilon/4 \le m\exp(-\Omega(\varepsilon^2 m)),$ which gives that \begin{align}\label{eqn:e7}m\le \bigO(1/\varepsilon^3).\end{align} Chaining together the sequence of inequalities in $t,t',t''',r,r',r'',r'''$ and $m$, and using $s_0 = 2^{\bigO(d^3)}$ and $\varepsilon \le 1/2^d$, we get the desired bound on the list size: $$t \le (1/\varepsilon)^{2^{2^{\bigO(d^3)}}}.$$ This finishes the proof of~\Cref{thm:largechar} assuming~\Cref{lem:anticonc}.

\subsubsection{Anti-concentration lemma}\label{sec:anticoncentration}
    We now prove the anti-concentration lemma~(\Cref{lem:anticonc}). 
    That is, we will show that there exists an absolute constant $M>0$ such that for any degree-$d$ polynomial $Q({\bf x})$  of sparsity at least $s$, over an Abelian group $G$ in which all non-zero elements have order greater than $(s+1)!$, that
    $$\Pr_{{\bf x}\in \{0,1\}^n}[Q({\bf x}\ne 0)] \ge 1/2^{d-1}-M^{d^3}/\sqrt{s}.$$
    Note that the above bound is trivial unless $s \ge M^{2d^3}$. 
    
    The proof proceeds by an induction on $d$. 
    \paragraph{The base case.} For $d=1$, we use the known anti-concentration result of Littlewood and Offord~\cite{littlewood-offord}, or rather the subsequent improvement due to Erdős~\cite{erdos}:
    \begin{theorem}[\cite{erdos}, Theorem 2 modified]\label{thm:anticon-deg1}
        There exists a constant $B>0$ such that any degree-$1$ multilinear polynomial $P({\bf x})$ over the reals with at least $r$ many variables with non-zero coefficients, we have
        $$\Pr_{{\bf x}\sim \{0,1\}^n}[P({\bf x})=0] \le \frac{B}{\sqrt{r}}.$$
    \end{theorem}
    However, we cannot apply the above theorem directly  since $G$ need not be the group of real numbers. Nevertheless, using a simple linear algebraic argument, we show in the below claim that the same anticoncentration inequality actually holds as long as the non-zero elements of $G$ have order greater than $s!$. 
    
    \begin{claim}\label{clm:reals-finite}
    Suppose that every degree-$d$ multilinear polynomial $P({\bf x})\in \R[{\bf x}]$ with at least $r$ disjoint non-zero monomials of degree $d$ has at most $c\cdot 2^n$ roots over $\{0,1\}^n$ for some $c=c(d,r)$. Then the following holds for every Abelian group $G$ in which the order of all the non-zero elements is greater than ${rd\choose \le d}!$: Every degree-$d$ polynomial $Q({\bf x})$ over $G$ with at least $r$ disjoint non-zero monomials of degree $d$ has at most $c\cdot 2^n$ roots over $\{0,1\}^n$.  
    \end{claim}

    For a degree-$1$ polynomial with sparsity $s$, we can apply \Cref{clm:reals-finite} and \Cref{thm:anticon-deg1} with $r = s-1$ to get the base case (as long as $M$ is a large enough constant). This finishes the proof of~\Cref{lem:anticonc} for the base case $d=1$, assuming~\Cref{clm:reals-finite}.
    
    \paragraph{The induction step.}
    Suppose $d\ge 2$ and that the claim is true for all degrees until $d-1$. To then prove it for degree $d$, we split the analysis into three cases, where the parameters $s_1$ and $s_2$ are to be fixed later.
    \begin{itemize}
    \item {\bf Case 1:} There exists at least one variable (say $x_1$) that is contained in at least $s_1$ monomials of $Q$. That is, let \begin{equation}\label{eqn:singlevar}
    Q(x_1,x_2,\dots,x_n) = x_1 Q_1(x_2,\dots,x_n) + Q_2(x_2,\dots,x_n),
    \end{equation} where we have $\deg(Q_1)\le d-1$ and $\spars(Q_1) \ge s_1$. We analyze  the probability that $Q({\bf x})$ by first setting  the variables $x_2,x_3,\dots,x_n$ and then setting the value of $x_1$. By induction hypothesis, we have that 
    \begin{equation}\label{eqn:indn}\Pr_{(x_2,\dots,x_n) \sim \{0,1\}^{n-1}} [Q_1(x_2,\dots,x_n) \ne 0] \ge 1/2^{d-2} - M^{(d-1)^3}/\sqrt{s_1}\end{equation}
    Interpreting~\eqref{eqn:singlevar} as a linear polynomial in $x_1$ with coefficients being $Q_1$ and $Q_2$, we have that $\Pr[Q({\bf x}) \ne 0] \ge \Pr[Q_1(x_2,\dots,x_n)\ne 0] \cdot \Pr[Q({\bf x})\ne 0 ~|~Q_1(x_2,\dots,x_n)\ne 0]$. By the DLSZ lemma (\Cref{thm:basic}) for degree-$1$ polynomials, we have that the second factor above $\Pr[Q({\bf x})\ne 0 ~|~Q_1(x_2,\dots,x_n)\ne 0]$ is at least $1/2$. Combined with~\eqref{eqn:indn}, we thus get that $\Pr[Q({\bf x}) \ne 0] \ge 1/2\cdot \paren{1/2^{d-2} - M^{(d-1)^3}/\sqrt{s_1}} \ge 1/2^{d-1} - M^{d^3}/\sqrt{s}$, by taking $s_1 = s/M^{2d^2}$ and $M$ a sufficiently large constant.
    \item {\bf Case 2:} There exist $s_2 = M^{d^2}$ many disjoint monomials of degree $d$ of $Q$ (with non-zero coefficients), again assuming $M$ is sufficiently large. We now apply the following anti-concentration result of Meka, Nguyen and Vu~\cite{MNV} which can be thought of as a generalization of~\cite{littlewood-offord,erdos} to higher degree when there are many disjoint monomials of maximal degree:
    \begin{theorem}[\cite{MNV}, Theorem 1.6 modified]\label{thm:mnv}
        There exists a constant $B>0$ such that for any degree-$d$ multilinear polynomial $P({\bf x})$ over the reals with at least $r$ many disjoint degree-$d$ monomials with non-zero coefficients, we have
        $$\Pr_{{\bf x}\sim \{0,1\}^n}[P({\bf x})=0] \le \frac{Bd^{4/3}\sqrt{\log r}}{r^{1/(4d+1)}}.$$
    \end{theorem}
    Applying the above theorem for $r=s_2=M^{d^2}$ and taking $M$ sufficiently large (as a function of $B$), we see that\begin{align}\label{eqn:antic}\Pr[Q({\bf x})=0] \le \frac{Bd^{4/3}\sqrt{\log s_2}}{s_2^{1/(4d+1)}} \le 1/2 \le 1-1/2^{d-1}.\end{align} However,~\Cref{thm:mnv} as stated only holds over the reals. Regardless, using~\Cref{clm:reals-finite}, we know that the same bound applies if all the non-zero elements of $G$ are of order greater than ${rd\choose \le d}!$.
    Recall that we are already assuming that the non-zero elements of $G$ are of order greater than $(s+1)!$ and that $s\ge M^{2d^3}$ from the hypothesis of~\Cref{lem:anticonc}. Hence, all that remains to be checked is that ${rd\choose \le d}! \le (s+1)!$ for the above choice of $r=s_2=M^{d^2}$. Indeed this inequality holds for sufficiently large $M$. 
    This finishes the proof of~\Cref{lem:anticonc} in Case 2. 
    \item {\bf Case 3:} Suppose neither Case 1 nor Case 2 occur. Using a greedy algorithm that repeatedly picks disjoint monomials of degree $d$ for as long as possible, we can find a vertex cover of size at most $s_2 d$ among the non-zero monomials of $Q$ of degree $d$.  That is, there exists at most $s_2 d$ many variables (call them ${\bf y}$) such that any non-zero monomial of $Q$ of degree $d$ contains at least one of these variables. 
    
    We will analyze the probability that $Q({\bf x})$ is zero by setting the ${\bf y}$ variables arbitrarily. Let $Q'$ denote the polynomial in the variables ${\bf x} \setminus {\bf y}$ after an arbitrary assignment to those of ${\bf y}$. Note that $\deg(Q') \leq d-1$, since we set at least one variable in each monomial of degree $d$. It suffices to show that $Q'$ is non-zero to get that $\Pr_{\bf z}[Q'({\bf z})\ne 0] \ge 1/2^{d-1}$ by a direct use of the DLSZ lemma. Since $Q$ has at least $s$ monomials and each variable is contained in at most $s_1$ monomials, the total number of monomials containing at least one variable from $y$ is at most $|{\bf y}|\cdot s_1 \le s_1 s_2 d$. Hence, even upon setting the variables in ${\bf y}$, at least $s-s_1s_2d$ monomials of $Q$ remain unaffected. However, the monomials that do get affected can cancel out these monomials. Nevertheless, there would be at least $s-2s_1s_2d=s-2M^{d^2}ds/M^{2d^2} = s(1-2d/M^{d^2}) > 0$ non-zero monomials in $Q'$. Thus $\Pr_{\bf x}[Q({\bf x})\ne 0] \ge 1/2^{d-1}$.
    \end{itemize}
This concludes the proof of~\Cref{lem:anticonc} assuming~\Cref{clm:reals-finite}, which we prove below.

\begin{proof}[Proof of~\Cref{clm:reals-finite}]
    We first reduce the number of variables of $Q$ from $n=|{\bf x}|$ to $rd$ by setting the variables that are {\em not} part of the $r$ disjoint disjoint degree-$d$ monomials of $Q$ arbitrarily. It then suffices to show that this restricted polynomial, which we shall denote by $Q'({\bf y})$, has at most $c\cdot 2^{rd}$ roots over $\{0,1\}^{rd}$. Towards a contradiction, let $S \subseteq \{0,1\}^{rd}$ be a subset of size greater than $c\cdot 2^{rd}$ such that $Q'$ evaluates to $0$ on all the points in $S$. Consider the following $\{0,1\}$-valued matrix $M$ of dimensions $\ell\times |S|$ where $\ell={rd \choose \le d}$: the $(i,j)$-th entry of $M$ denotes the evaluation of the $i$-th monomial at the $j$-th point in $S$. Let $a_1,a_2,\dots,a_r \in [\ell]$ denote the rows corresponding to the disjoint monomials of $Q'$ and for each $i\in [\ell]$, let $e_i \in \{0,1\}^{\ell}$ denote the vector that takes value 1 at the $i$-th index and 0 everywhere else. 
    
    We claim that there exists at least one index $k\in [r]$ such that $e_{a_k}$ lies in the column span of $M$, \emph{when $M$ is treated as a matrix over $\mathbb{R}$}. We prove this by contradiction. That is, suppose that $e_{a_k} \notin V$ for all $k\in [r]$, where $V \subseteq \R^\ell$ denotes the column space of $M$. Note that given any vector ${\bf u}\in \R^\ell$, we can uniquely express it as ${\bf u}= {\bf u}^\parallel + {\bf u}^\perp$ such that ${\bf u}^\parallel \in V$ and ${\bf u}^\perp \perp V$ (i.e., $\innerprod{{\bf u}^\perp,{\bf v}}=0$ for all ${\bf v}\in V$ . Denoting the orthogonal subspace of $V$ by $V^\perp$, this is equivalent to ${\bf u}^\perp \in V^\perp$). Since $e_{a_k} \notin V$, we have that $e_{a_k}^\perp \ne {\bf 0}$. Hence, $\{{\bf x}\in V^\perp~|~\innerprod{{\bf x},e_{a_k}^\perp} = 0\}$ is a subspace of $V^\perp$ of co-dimension 1. Since a finite union of subspaces of co-dimension 1 is never equal to the ambient vector space over $\R$ (which is $V^\perp$ in our case), we conclude that there exists a vector ${\bf p} \in V^\perp$ such that $\innerprod{{\bf p},{e_{a_k}}^\perp} \ne 0$ for all $k\in [r]$. Since $\innerprod{{\bf p},e_{a_k}^\parallel} = 0$ as the two vectors are in orthogonal subspaces, we get that $\innerprod{{\bf p},e_{a_k}} \ne 0$ for all $k\in [r]$. Moreover, since ${\bf p}$ is orthogonal to the columnspace of $M$, we get that $\innerprod{{\bf p}, M_j}=0$ where $M_j\in \{0,1\}^\ell$ denotes the $j$-th column of $M$ for every $j\in [|S|]$. Let $P({\bf y})\in \R[{\bf y}]$ denote the polynomial with coefficients represented by ${\bf p}$. Then, the above inner product relations imply that $P$ is a degree-$d$ polynomial with $r$ disjoint non-zero monomials yet it vanishes on $S$. This is a contradiction to the assumption in the claim statement since $|S| > c\cdot 2^{rd}$. Hence, indeed there exists $k\in [r]$ such that $e_{a_k}$ is spanned the columns of $M$. 
    
    Suppose 
    \begin{align} \label{eqn:span}e_{a_k} = \sum_{p=1}^t \alpha_p M_{j_p},\end{align} where ${\boldsymbol{\alpha}}=(\alpha_1,\alpha_2,\dots,\alpha_t) \in \R^t$ and $M_{j_1},M_{j_2},\dots,M_{j_t}$ are linearly independent columns of $M$ for some $t\le \ell$. Let us denote the submatrix indexed by the columns $j_1,j_2,\dots,j_t$ by $M' \in \R^{\ell \times t}$. Since $M'$ is full-column-rank, let $i_1,i_2,\dots,i_t \in [\ell]$ be the indices of the rows of $M'$ that are linearly independent -- we shall denote the corresponding submatrix of $M'$ by $M'' \in \R^{t\times t}$. Let ${\bf v}'=e_{a_k}$ and ${\bf v''}\in \{0,1\}^t$ be the restriction of ${\bf v}'$ to the indices $i_1,i_2,\dots,i_t$. From~\eqref{eqn:span}, we have $M' {\boldsymbol{\alpha}} = {\bf v'}$, and hence $M'' {\boldsymbol{\alpha}} = {\bf v''}$. Since $M''$ is invertible, by Cramer's rule, we have that $\alpha_p = \frac{\det(M''_p)}{\det(M'')}$ for all $p\in [t]$, where $\det(\cdot)$ denotes the determinant and $M''_p$ is the matrix $M''$ with its $p$-th column replaced by the vector ${\bf v}''$. By multiplying $\det(M'')$ on both sides of~\eqref{eqn:span}, we see that there exist integers $\beta_0 \in [t!]$ and $\beta_1,\dots,\beta_t \in \{-t!,\dots,t!\}$ such that $\beta_0 e_{a_k} = \sum_{p=1}^t \beta_p M_{j_p}$, where we are using the fact that the determinant of any $\{0,1\}^{t\times t}$ matrix lies in $\{-t!,\dots,t!\}$. 
    
    We can use this to argue about the polynomial $Q'$ defined above over the group $G$. Assume that the coefficient of the $i$th monomial in $Q'$ is $Q'_i\in G$. Recall that $Q'$ evaluates to $0$ on all points in $S$ and that $Q'$ has disjoint monomials corresponding to rows $a_1,\ldots, a_r\in [\ell]$, implying that $Q'_{a_j}\neq 0$ for each $j\in [r].$ The previous paragraph implies that $\beta_0 Q'_{a_k} = \sum_{p=1}^t \beta_p Q'({\bf x}_p)$, where $Q'_{a_k} \in G$ denotes the coefficient of the corresponding monomial in $Q'$ and ${\bf x}_p$ is the $p$-th point in $S$ (following the same indexing as the columns of $M$). Since $Q'$ vanishes on $S$, we have $\beta_0 Q'_{a_k}=0$. This is  a contradiction since $\beta_0 \le t! \le  {rd\choose \le d}!$ and we assumed that all non-zero elements of $G$ have order greater than ${rd\choose \le d}!$.
\end{proof}

\subsection{Combinatorial bound for a product of $p$-groups ($p<p_0$)}\label{subsec:small-groups}

In this subsection, we prove a combinatorial bound on the list size in case of Abelian groups $G$ that are products of $p$-groups for small $p$, i.e. \Cref{thm:smallchar}.  We say that $G$ is a \emph{$(< p_0)$-group} if every prime factor of $|G|$ is smaller than $p_0$ (or equivalently, that $G$ is a finite product of $p$-groups where $p < p_0$). 

The proof of the combinatorial bound in this case is quite different from results that prove similar statements when the domain is $\mathbb{Z}_p^n$ for constant $p$~\cite{gkz-list-decoding, BhowmickL}. In particular, we depart from the analytic ideas of~\cite{BhowmickL} and use combinatorial and algebraic techniques based on monomial orderings and the well-known technique of `fingerprinting' (see, e.g.~\cite{GH}). 

\paragraph{Monomial ordering and leading monomials.} We fix variables $x_1,\ldots,x_n$ and define an ordering among (not necessarily multilinear) monomials in these variables as follows. Given monomials $m_1, m_2$, we say that $m_1 \preceq m_2$ if $\deg(m_1) < \deg(m_2)$ or $\deg(m_1) = \deg(m_2)$ and for the least $i$ such that the two monomials differ in the exponent of $x_i$, we have a lower power of $x_i$ dividing $m_1$ than $m_2.$ This is also called the \emph{graded lexicographic order}.\footnote{Any graded monomial order in the sense of \cite[Section 2, Chapter 2]{CLO} will do just as well.} Some examples are as follows.
\begin{align*}
    x_{2} x_{3}^{2} \preceq x_{1}^{4}, \quad x_{1} x_{2} x_{3}^{2} \preceq x_{1} x_{2}^{2} x_{3}
\end{align*}
Now, given a polynomial $P\in \mathcal{P}_d(\{0,1\}^n,G)$ for a group $G$, we define its \emph{leading monomial} $\mathrm{LM}(P)$ to the largest monomial (w.r.t. $\preceq$) with non-zero coefficient in $P.$ We identify $\mathrm{LM}(P)$ with the set $S\subseteq [n]$ of size at most $d$ indexing the variables that appear in it.

We will use monomial orderings in the setting when $G = \mathbb{Z}_p$ to show that many distinct polynomials cannot agree with the same function at too many points. This uses crucially the following tail bound, which we view as independently interesting. It is proved using the fairly standard technique of using `footprints' \cite{GH} and an idea for proving tail bounds by Panconesi and Srinivasan~\cite{PS-chernoff}. As far as we know, such a bound for low-degree polynomials has not been observed before.

\begin{thmbox}
\begin{lemma}[Tail bound for degree-$d$ polynomials]
   \label{lem:tailbound}
    Fix any field $\F$ and any integer $d\geq 1$. Let $S_1,\ldots, S_t\in \mathcal{P}_d(\{0,1\}^n, \F)$ be polynomials such that $\mathrm{LM}(S_i)\cap \mathrm{LM}(S_j) = \emptyset$ for every distinct $i,j\in [t].$ Then for every $\eta > 0$, we have
    \begin{equation}
       \label{eq:tailbound}
        \Pr_{{\bf a}}{\left[\left|\{i\in [t]\ |\ S_i({\bf a}) = 0\}\right| \geq \left(1 - \frac{1}{2^d} + \eta\right)\cdot t\right]} \leq \exp(-\Omega(\eta^2 t)).
    \end{equation}
\end{lemma}
\end{thmbox}

The lemma is proved in \Cref{sec:tailbound} below.

We now outline the proof of the main theorem (\Cref{thm:smallchar}) of this section, which involves several steps.

\begin{enumerate}
    \item The first, and more involved, step is to prove the bound in the case that $G = \mathbb{Z}_p$ where $p < p_0$ is prime. The proof in this case splits into three smaller steps. 
    \begin{enumerate}
        \item \textbf{Pigeonhole argument}: Fix an $f:\{0,1\}^n \rightarrow \mathbb{Z}_p.$ Given a list of $L$ polynomials $P_1,\ldots, P_L\in \mathcal{P}_d(\{0,1\}^n,\mathbb{Z}_p)$ that are $(\frac{1}{2^d}-\varepsilon)$-close to $f$, we show how to obtain a sub-list of size $\ell = \Omega(\log_p L )$ polynomials $Q_1,\ldots, Q_\ell$ that moreover satisfy the property that their leading monomials are \emph{distinct}.

        \item \textbf{Sunflower lemma}: We then apply the Sunflower lemma (\Cref{lem:sunflower}) to the monomials $\mathrm{LM}(Q_1),\ldots, \mathrm{LM}(Q_\ell)$ to find a subset of $t = \Omega_d(\ell^{1/d})$ polynomials from $\set{Q_{1},\ldots,Q_{\ell}}$ whose leading monomials form a \emph{sunflower}, i.e. they are pairwise disjoint except for their common intersection.\footnote{The recently improved sunflower lemma due to Alweiss, Lovett, Wu, and Zhang \cite{ALWZ} unfortunately does not lead to significantly improved parameters here, and so we stick to the classical version.}

        \item \textbf{Using the tail bound}: We can now apply the tail bound (\Cref{lem:tailbound}) stated above and a simple combinatorial argument to show that $t\leq \poly(1/\varepsilon).$ Overall, this leads to a bound of $L \leq \exp\left(\bigO_d(1/\varepsilon)^{\bigO(d)}\right),$ concluding the proof in the prime case.
    \end{enumerate}

    \item \textbf{Modifying \cite{DinurGKS-ECCC}}: The second step is to `lift' the list bound $L$ that holds for $\mathbb{Z}_p$ ($p < p_0$) to a list bound that holds over all the finite Abelian $(< p_0)$-groups $G$. In the linear ($d=1$) case handled in~\cite{ABPSS24-ECCC}, this was done using the work of~\cite{DinurGKS-ECCC}, which gives a combinatorial characterization for such a lifting. This combinatorial property holds for the space of linear polynomials, implying (using~\cite{DinurGKS-ECCC}) that a list bound of $L$ in the prime case `lifts' to a list bound of $\poly(L)$ for all $(< p_0)$-groups $G$.
    
    Unfortunately, for $d > 1$, the characterization of~\cite{DinurGKS-ECCC} is not applicable as stated.\snote{Say why?} However, we show that a simpler proof allows us to recover a weaker bound of $L^{\bigO(\log(1/\varepsilon))}$. Using the bound for the prime case, we again get a bound of $\exp\left(\bigO_d(1/\varepsilon)^{\bigO(d)}\right)$ on the list size for any Abelian $(< p_0)$-group $G$.
\end{enumerate}

In the rest of the section, we show how to carry out the above strategy. We start with the proof of \Cref{lem:tailbound} in \Cref{sec:tailbound}, followed by the prime case in \Cref{sec:combboundsmallprime} and the case of Abelian $(< p_0)$-groups in \Cref{sec:combboundsmallgeneral}.

\subsubsection{Proof of the tail bound}\label{sec:tailbound}

To derive the tail bounds (\Cref{lem:tailbound}) we will need the following theorem of Panconesi and Srinivasan \cite{PS-chernoff}, which is an extension of Chernoff-Hoeffding bound.

\begin{theorem}\label{thm:PS-chernoff}
    \cite[Theorem 3.4]{PS-chernoff} Let $Z_1,\ldots, Z_t$ be Boolean random variables such that for some $\alpha \in [0,1]$ and for every subset $I\subseteq [t]$, we have
    $\Pr_{}{\left[\wedge_{i\in I} Z_i = 1\right]} \leq \alpha^{|I|}.$
    Then, we have the following tail bound.
    \[
    \Pr_{}{\left[\sum_{i=1}^t Z_i \geq (\alpha + \eta)\cdot t\right]} \leq \exp(-\Omega(\eta^2 t)).
    \]
\end{theorem}
\begin{proof}[Proof of \Cref{lem:tailbound}]
We start by defining the Boolean random variables $Z_{1},\ldots,Z_{t}$ as follows: For $i \in [t]$, $Z_{i} = 1$ exactly when $S_{i}$ is equal to $0$, i.e. $Z_{i}(\mathbf{a}) = 1$ iff $S_{i}(\mathbf{a}) = 0$. To use \Cref{thm:PS-chernoff}, we need to show the following:\newline
For each subset $I \subseteq [t]$,
\begin{align*}
    \Pr[\wedge_{i \in I} Z_{i} = 1] \, \leq \, \alpha^{|I|},
\end{align*}
where $\alpha := \left(1-\frac{1}{2^d}\right)$. 


For simplicity in notation, assume that $I = \{1,\ldots, r\}$ and define the $\mathrm{Zero}_I\subseteq \{0,1\}^n$ as the set of common zeroes of $S_1,\ldots, S_r$ in $\Boo^{n}$, i.e.
\begin{align*}
    \mathrm{Zero}_{I} \, := \, \setcond{\mathbf{a} \in \Boo^{n}}{S_{i}(\mathbf{a}) = 0, \text{ for all } \, i \in [t]}
\end{align*}
We want to show that $|\mathrm{Zero}_I| \leq \alpha^r \cdot 2^{n}.$

To do this, we use the standard `footprint bound' (see e.g. \cite{GH} and \cite[\S 5.3]{CLO}) which can be seen as a version of the linear algebra method in combinatorics (see, e.g. \cite{Babai-Frankl}). More precisely, we consider the vector space of functions $g: \mathrm{Zero}_I\rightarrow \F$. We denote this vector space by $\mathcal{F}_I$ and we will show that 
\begin{equation}
    \label{eq:zeroI}
    \dim(\mathcal{F}_I) \leq \cdot \alpha^{r} \cdot 2^{n}.
\end{equation}
Note that the set of indicator functions for each point in $\mathrm{Zero}_{I}$ is a basis for $\mathcal{F}_{I}$ and thus $\dim(\mathcal{F}_{I})$ is equal to $|\mathrm{Zero}_{I}|$. Combining it with  \Cref{eq:zeroI}, we will get that $|\mathrm{Zero}_{I}| \leq 2^{n} \cdot \alpha^{r}$ and this will finish the proof.

We bound $\dim(\mathcal{F}_I)$ by constructing a spanning set for $\mathcal{F}_I$ of size at most $\alpha^{r} \cdot 2^{n}$. The set $\mathrm{Zero}_{I}$ is a subset of the product set $\Boo^{n}$. We define another set of polynomials and show that their common set of zeroes in $\mathbb{Z}_{p}^{n}$ is exactly equal to $\mathrm{Zero}_{I}$, gaining the advantage of working over a field and use linear algebra. Define the set $I'$ of polynomials as follows:
\begin{align*}
    I' \, := \, \set{x_1^2 - x_1,\ldots, x_n^2 - x_n, S_1,\ldots, S_r}
\end{align*}
It is easy to see that the common set of zeroes of polynomials in $I'$ over $\mathbb{Z}_{p}^{n}$ is exactly $\mathrm{Zero}_{I}$ (the polynomial constraint $x_{i}^{2} - x_{i}$ forces the $i^{th}$ coordinate to be in $\Boo$).

Now, given any function $g:\mathrm{Zero}_I\rightarrow \F$, we can express $g$ as a polynomial (e.g. via \Cref{thm:basic}). Using a standard division algorithm for multivariate polynomials (see \cite[\S 2.3, Theorem 3]{CLO}), we can write 
\begin{equation}
\label{eq:division}
g = \sum_{i=1}^n a_i\cdot (x_i^2 - x_i) + \sum_{j=1}^r b_j S_j + R
\end{equation}
where $a_1,\ldots, a_n, b_1,\ldots, b_r\in \mathbb{Z}_p[x_1\ldots, x_n]$ are some polynomials and $R\in \mathbb{Z}_p[x_1\dots, x_n]$ is such that no monomial with non-zero coefficient in $R$ is divisible by any of the leading monomials of the polynomials in the set $I'$. In other words, $R$ is a linear combination of \textit{multilinear} monomials that are not divisible by $\mathrm{LM}(S_1),\ldots, \mathrm{LM}(S_r).$\footnote{The last couple of paragraphs is essentially a summary of part of~\cite[\S 5.3, Proposition 4]{CLO} stated here for completeness.}

By \Cref{eq:division}, the polynomials $R$ and $g$ represent the same function over $\mathrm{Zero}_I.$ It follows that the set of multilinear monomials that are not divisible by $\mathrm{LM}(S_1),\ldots, \mathrm{LM}(S_r)$ span the space $\mathcal{F}_I.$ In particular, the dimension of the vector space $\mathcal{F}_I$ is upper bounded by the number of such multilinear monomials. 

Fnially we argue that the number of multilinear monomials not divisible by $\mathrm{LM}(S_{1}), \ldots, \mathrm{LM}(S_{r})$ is at most $\alpha^{r} \cdot 2^{n}$. Picking a uniformly random multilinear monomial corresponds to choosing a uniformly random subset $J \subseteq [n].$ The chance that the random multilinear monomial is not divisible by $\mathrm{LM}(S_j)$ ($j\in [r]$) is equal to the chance that the set $J$ does not contain the set corresponding to $\mathrm{LM}(S_{j})$ which is at most $\alpha = 1-\frac{1}{2^d}$ since $|\mathrm{LM}(S_j)|\leq d$. Since the leading monomials of $S_1,\ldots, S_r$ are pairwise disjoint, these events are mutually independent, leading to the conclusion that a uniformly random multilinear monomial is not divisible by $\mathrm{LM}(S_1),\ldots, \mathrm{LM}(S_r)$ is at most $\alpha^r$. This is equivalent to the desired claim.

This proves \Cref{eq:zeroI}, which concludes the proof of \Cref{lem:tailbound} as described above.
\end{proof}

We will use a corollary of this tail bound, which we now prove.

\begin{corollary}
    \label{cor:tailbound}
    Fix any field $\F$ and any integer $d\geq 1$. Let $Q_1,\ldots, Q_t\in \mathcal{P}_d(\{0,1\}^n, \F)$ be polynomials such that $\mathrm{LM}(Q_i)\cap \mathrm{LM}(Q_j) = \emptyset$ for every distinct $i,j\in [t].$ Then for every $\eta > 0$, we have
    \[
        \Pr_{{\bf a}}{\left[\exists c\in \F\ s.t.\ \left|\{i\in [t]\ |\ Q_i({\bf a}) = c\} \right|\geq \left(1 - \frac{1}{2^d} + \eta\right)\cdot t\right]} \leq t\exp(-\Omega(\eta^2\cdot t)).
    \]
\end{corollary}

\begin{proof}
    Without loss of generality, assume that $\mathrm{LM}(Q_1) \precneq \mathrm{LM}(Q_2) \cdots \precneq \mathrm{LM}(Q_t).$ 
    
    Assume that ${\bf a}\in \{0,1\}^n$ is chosen uniformly at random as in the statement of the corollary. For $j< t,$ let $\mathcal{E}_j$ denote the event that there is a $c\in \F$ such that for at least $\left(1 - \frac{1}{2^d} + \eta\right)\cdot t$ many $i\in [t]$, we have $Q_i({\bf a}) = c$ and furthermore that $j$ is the smallest index such that  $Q_j({\bf a}) = c$.

    The event whose probability we are trying to bound is contained in 
    $
    \mathcal{E}_1\cup \mathcal{E}_2\cdots \cup \mathcal{E}_r
    $
    where $r = t/2^d.$ 

    Fix any $j\leq r$. The probability of $\mathcal{E}_j$ is upper bounded by the probability that at least a $\left(1 - \frac{1}{2^d} + \eta\right)$ fraction of the polynomials
    \[
    Q_{j+1}-Q_j, Q_{j+2} - Q_j, \ldots, Q_t - Q_j
    \]
    all simultaneously vanish at the point ${\bf a}.$ Note that these polynomials have leading monomials $\mathrm{LM}(Q_{j+1}),\ldots, \mathrm{LM}(Q_t)$ respectively, which are pairwise disjoint. Hence the tail bound \Cref{lem:tailbound} is applicable and we can bound the probability of $\mathcal{E}_j$ by $\exp(-\Omega(\eta^2 \cdot (t-j))) = \exp(-\Omega(\eta^2\cdot t)).$

    The corollary follows by a union bound over the probability of $\mathcal{E}_1,\ldots, \mathcal{E}_r.$
\end{proof}

\subsubsection{Combinatorial bound for the prime case}\label{sec:combboundsmallprime}

In this subsection, we prove a combinatorial bound on the list size for degree $d$ polynomials when the coefficients are from $\mathbb{Z_{p}}$ for a prime $p$. The list size is a constant dependent on $d, p$, and $\varepsilon$ (independent of $n$).

\begin{lemma}
    \label{lem:combboundconstprime}
    Fix any $d\geq 1$ and any prime $p.$ Then for every function $f: \Boo^{n} \to \mathbb{Z}_p$ we have, $|\mathsf{List}_{\varepsilon}^{f}| \leq \exp(\bigO_{d,p}(1/\varepsilon)^{\bigO(d)})$.
\end{lemma}

To prove this lemma, we follow the outline from the beginning of the section. For the rest of this section, let $f$ be an arbitrary function as in the statement of the lemma and let $\mathsf{List}_{\varepsilon}^{f} = \{P_1,\ldots, P_L\}$ where $L \in [p^\ell, p^{\ell+1})$ for an integer $\ell.$ Note that $\ell = \Omega(\log_p L).$

\paragraph{\underline{Pigeonhole argument}.} We start with a pigeonhole argument that allows us to find a sub-list of polynomials from $\mathsf{List}^{f}_{\varepsilon}$ that have distinct leading monomials. More formally we prove the following.

\begin{claim}
    \label{clm:pigeonhole}
    For each non-negative integer $i\leq \ell$, there is a function $f_i:\{0,1\}^n\rightarrow \mathbb{Z}_p$ and a set of polynomials $\mathcal{Q}_i = \{Q_1,\ldots, Q_i\} \cup \mathcal{Q}'_i$ such that
    \begin{enumerate}
        \item $\mathcal{Q}_i \subseteq \mathsf{List}_{\varepsilon}^{f_i}$,
        \item $\mathrm{LM}(Q_1)\succneq\cdots \succneq \mathrm{LM}(Q_i) \succneq \mathrm{LM}(Q)$ for each $Q\in \mathcal{Q}'_i,$ and
        \item $|\mathcal{Q}_i'|\geq p^{\ell-i}.$
    \end{enumerate}
\end{claim}

\begin{proof}
    The proof is via induction on $i.$ For the base case $i=0$, we can simply take $f_0 = f$ and $\mathcal{Q}_0 = \mathsf{List}_{\varepsilon}^{f}.$

    Assuming the statement for $i < \ell,$ we prove it for $i+1$ as follows. We define the `plurality polynomial' $\mathrm{Pl'}_i$ as follows: for every multilinear monomial $m$, the coefficient of $m$ in $\mathrm{Pl'}_i$ is the plurality of the coefficient of $m$ among all the polynomials in $\mathcal{Q}_i'$, where we break ties arbitrarily.

    Define the function $f_{i+1} := f_i - \mathrm{Pl'}_i$ and define $\mathcal{Q}_i'' := \{Q- \mathrm{Pl'}_i\ |\  Q\in \mathcal{Q}_i'\}.$ Using $f_{i+1}$ and $\mathcal{Q}_{i}''$, we define $\mathcal{Q}_{i+1}$ next:
    \begin{align*}
        \mathcal{Q}_{i+1} := \{Q_1 - \mathrm{Pl'}_i,\ldots, Q_i - \mathrm{Pl'}_i\}\cup \{Q_{i+1}\} \cup \mathcal{Q}_{i+1}',
    \end{align*}
    where 
    \begin{itemize}
        \item $Q_{i+1}$ is chosen to be any polynomial in $\mathcal{Q}_i''$ whose leading monomial $m$ is as large as possible among the leading monomials of polynomials in $\mathcal{Q}_i'',$
        \item $\mathcal{Q}_{i+1}'$ is the set of polynomials in $\mathcal{Q}_i''$ that have a leading monomial \emph{strictly smaller} than $m$, the leading monomial of $Q_{i+1}$.
    \end{itemize}

    We now show that $f_{i+1}$ and $\mathcal{Q}_{i+1}$ satisfy the required properties. 

    \begin{enumerate}
        \item It is clear that $\mathcal{Q}_{i+1}\subseteq \mathsf{List}_{\varepsilon}^{f_i}$ because each polynomial $\tilde{Q}\in \mathcal{Q}_{i+1}$ can be written as $Q-\mathrm{Pl'}_i$ for some $Q\in \mathcal{Q}_i.$ Hence $\delta(f_{i+1}, \tilde{Q}) = \delta(f_i, Q) \leq \frac{1}{2^d}-\varepsilon$.

        \item We have defined $\mathcal{Q}_{i+1} = \{Q_1 - \mathrm{Pl'}_i,\ldots, Q_i - \mathrm{Pl'}_i, Q_{i+1}\}\cup \mathcal{Q}_{i+1}'$. If $m_j$ denotes the leading monomial of $Q_j$ for $j\leq i,$ we observe that the coefficient of $m_j$ in $\mathrm{Pl}'_i$ is $0$ since the plurality used in defining $\mathrm{Pl'}_i$ is only taken over the polynomials in $\mathcal{Q}_i'$ all of which have a leading monomial smaller than $m_j$ by the induction hypothesis. In particular, the leading monomial of $Q_j - \mathrm{Pl'}_i$ is also $m_j.$ We thus have
        \[
        \mathrm{LM}(Q_1 - \mathrm{Pl'}_i)\succneq\cdots \succneq \mathrm{LM}(Q_i-\mathrm{Pl'}_i) \succneq \mathrm{LM}(Q_{i+1})
        \]
        where for the last inequality we  again used the inductive hypothesis. Finally, given any $Q\in \mathcal{Q}_{i+1}'$, we have $\mathrm{LM}(Q) \precneq m = \mathrm{LM}(Q_{i+1})$ by definition of $\mathcal{Q}_{i+1}'.$

        \item Finally, we note that for any monomial $m$, the plurality of the coefficients of $m$ among the polynomials in $\mathcal{Q}_i''$ is zero, since we defined $\mathcal{Q}_i''$ by subtracting from each $Q\in \mathcal{Q}_i'$ the `plurality polynomial' $\mathrm{Pl'}_i.$ In particular, the number of polynomials in $\mathcal{Q}_{i}''$ that have a leading monomial strictly smaller than $m$ (or equivalently, the number of polynomials in $\mathcal{Q}_i''$ where the coefficient of $m$ is $0$) is at least $|\mathcal{Q}_i'|/p \geq p^{\ell-i-1}.$       
    \end{enumerate}

    This concludes the induction and hence the proof of the claim.
\end{proof}

Applying \Cref{clm:pigeonhole} with $i = \ell$, we see that there is a function $f_{\ell}$ such that $|\mathsf{List}_{\varepsilon}^{f_{\ell}}|\geq \ell$ and $\mathsf{List}_{\varepsilon}^{f_{\ell}}$ contains polynomials $Q_1, \ldots, Q_\ell$ with distinct leading monomials. The rest of the argument will bound $\ell.$

\paragraph{\underline{Sunflower lemma}.} We now come to the second step of the argument, which is an application of the Sunflower lemma \Cref{lem:sunflower} to the set of leading monomials of $Q_1,\ldots, Q_\ell$ (seen as subsets of the universe $[n]$ of size at most $d$). By \Cref{lem:sunflower}, there is a sub-collection of $t = \Omega(\ell^{1/d}/d)$ many polynomials such that their leading monomials form a sunflower. Without loss of generality, we assume that these polynomials are $Q_1,\ldots, Q_t.$ Similarly, we assume that $\mathrm{LM}(Q_1)\cap \cdots \cap \mathrm{LM}(Q_t) = \{j_1,\ldots, j_k\}$ for some $k\in \{0,\ldots, d-1\}.$ Let $J$ denote this core of the sunflower.

For each $i\in [t]$, we can express the polynomial $Q_{i}$ as a polynomial in the variables in core $J$, i.e.
\[
Q_i(x_1,\ldots, x_n) = \sum_{A\subseteq J}Q_{i,A}(x_j: j\in [n]\setminus J)\cdot  x^A
\]
where $x^A$ denotes the product of the variables in $A.$ Note that $\deg(Q_{i,J}) \leq d-k$ for each $i\in [t]$. We also note that $\mathrm{LM}(Q_{i,J}) = \mathrm{LM}(Q_{i})\setminus J$ by the definition of our monomial order and hence the leading monomials of $\mathrm{LM}(Q_{1,J}),\ldots, \mathrm{LM}(Q_{t,J})$ are \emph{pairwise disjoint.}

\paragraph{\underline{Using the tail bound}.} We are now ready to conclude the bound on $\ell$ (and hence on $L = |\mathsf{List}^f_\varepsilon|$) using \Cref{cor:tailbound}. More precisely we prove the following claim.

\begin{claim}
    \label{clm:combboundfinaladhoc}
    Assume $i\in [t]$ and ${\bf a}\in \{0,1\}^n$ are chosen uniformly at random from their respective domains. Then
    \[
    \Pr_{i,{\bf a}}{\left[f_\ell({\bf a}) \neq Q_i({\bf a})\right]}\; \geq \; \frac{1}{2^d}-\frac{\varepsilon}{2} - t\cdot \exp(-\Omega(\varepsilon^2 \cdot t)).
    \]
\end{claim}

The statement of the claim immediately implies that $t = \bigO(\log(1/\varepsilon)/\varepsilon^2)$ since we have $\Pr_{{\bf a}}{\left[ f_\ell({\bf a}) \neq Q_i({\bf a})\right]} = \delta(f_\ell, Q_i) \leq \frac{1}{2^d} - \varepsilon.$ It therefore suffices to prove the claim, which we do now.

\begin{proof}[Proof of \Cref{clm:combboundfinaladhoc}]
    We sample ${\bf a}$ in two steps: we first sample its projection to the variables indexed by $[n]\setminus J$, which we denote by ${\bf a'}$, followed by its projection to the variables indexed by $J$, which we denote by ${\bf a''}.$

    Given a fixing ${\bf a'}$ of the variables indexed by $[n]\setminus J,$ we denote by $f_{\ell, {\bf a'}}$ and $Q_{i,{\bf a'}}$ the corresponding restrictions of $f_\ell$ and $Q_i$ ($i\in [t]$) respectively. Note that each of these is a function of the $k$ variables indexed by $J$ and can hence be expressed uniquely as a multilinear polynomial in these variables (\Cref{thm:basic}). Moreover, the coefficient of the monomial $x^J$ in $Q_{i,{\bf a'}}$ is exactly $Q_{i,J}({\bf a'}).$
    
    We denote by $\mathcal{B} = \mathcal{B}({\bf a'})$ the event that 
    \[
    \exists c\in \mathbb{Z}_p \text{ such that } \left|\{i\in [t]\ |\ Q_{i,J}({\bf a'}) = c\}\right| \geq \left(1 - \frac{1}{2^{d-k}} + \frac{\varepsilon}{2}\right)\cdot t.
    \]
    Since the polynomials $Q_{1,J},\ldots, Q_{t,J}$ have degree at most $d-k$ each and their leading monomials are pairwise disjoint, we can apply \Cref{cor:tailbound} to bound the probability of $\mathcal{B}$ by $t\cdot \exp(-\Omega(\varepsilon^2 t)).$

    Fix any ${\bf a'}$ such that $\mathcal{B}$ does not occur. In this case, the \emph{multiset} of polynomials $\{Q_{i,{\bf a'}}\ |\ i\in [t]\}$ has the property that no polynomial appears more than $(1-\frac{1}{2^{d-k}} + \frac{\varepsilon}{2})\cdot t$ times in it. In particular, when we choose $i\in [t]$ uniformly at random, we see that the probability that $Q_{i,{\bf a'}} \neq f_{\ell, {\bf a'}}$ is at least $\frac{1}{2^{d-k}}-\frac{\varepsilon}{2}.$ Since any pair of distinct functions on $k$ variables disagree at at least a single point, we have
    \[
    \Pr_{i, {\bf a''}}{[f_{\ell, {\bf a'}}({\bf a''})\neq Q_{i,{\bf a'}}({\bf a''})]} \geq \frac{1}{2^k}\cdot \left(\frac{1}{2^{d-k}} - \frac{\varepsilon}{2}\right) \geq \frac{1}{2^d}-\frac{\varepsilon}{2}
    \]
    
    Overall, we have 
    \begin{gather*}
    \Pr_{i,{\bf a}}{\left[f_\ell({\bf a})\neq Q_i({\bf a})\right]} =
    \Pr_{{\bf a'},i, {\bf a''}}{[f_{\ell, {\bf a'}}({\bf a''})\neq Q_{i,{\bf a'}}({\bf a''})]}\\
     \geq (1-\Pr_{{\bf a'}}{[\mathcal{B}]})\cdot \left(\frac{1}{2^d}-\frac{\varepsilon}{2}\right)
    \geq \frac{1}{2^d}-\frac{\varepsilon}{2} - \Pr_{}{[\mathcal{B}]}\\
    \geq \frac{1}{2^d}-\frac{\varepsilon}{2} - t\cdot \exp(-\Omega(\varepsilon^2 \cdot t))
    \end{gather*}
    proving the claim.
\end{proof}

As noted above, \Cref{clm:combboundfinaladhoc} implies that $t = \bigO_{d}(\log(1/\varepsilon)/\varepsilon^2),$ implying that $\ell = \bigO_d(t)^d = \bigO_d(1/\varepsilon)^d.$ Since $\ell = \Omega(\log_p L),$ we get $L = \exp(\bigO_{d,p}(1/\varepsilon)^d)$, proving \Cref{lem:combboundconstprime}.

\subsubsection{Combinatorial bound for the general case}\label{sec:combboundsmallgeneral}
In this subsection, we prove \Cref{thm:smallchar} which we recall below.

\combsmallchar*

To prove \Cref{thm:smallchar} for any $(< p_0)$-groups $G$, we follow the argument from~\cite{ABPSS24-ECCC}. The main bottleneck (as mentioned also above) is that the work of~\cite{DinurGKS-ECCC}, which is useful in `lifting' the combinatorial bound from the prime case to the case of general groups for degree-$1$ polynomials, no longer seems to be applicable here. So the main innovation is to give a proof of such a lifting strategy that also works for higher-degree polynomials (albeit with worse parameters). 

More precisely we prove the following lemma, that in conjunction with \Cref{lem:combboundconstprime} immediately implies \Cref{thm:smallchar}.

\begin{lemma}
    \label{lem:dgksvariant}
    Fix any $\varepsilon > 0$ and any positive integer $d$. Let $\mathcal{S}$ be a set of primes such that for each $p\in \mathcal{S}$, we know that the space $\mathcal{P}_d(\{0,1\}^n,\mathbb{Z}_p)$ is $(\frac{1}{2^d}-\varepsilon, L)$-list decodable for some positive integer $L$. Then, for every finite Abelian group $G$ that is a product of $p$-groups for $p\in \mathcal{S},$ it holds that $\mathcal{P}_d(\{0,1\}^n,G)$ is $(\frac{1}{2^d}-\varepsilon, L')$-list decodable, where $L' = L^{\bigO_d(\log 1/\varepsilon)}.$
\end{lemma}

The proof of the lemma is inspired by the proof of the analogous combinatorial bound in~\cite{DinurGKS-ECCC}, but needs less structure on the codewords in the underlying code.

\begin{proof}
    Fix any function $f:\{0,1\}^n\rightarrow G$. We want to bound the size of $\mathsf{List}^f_\varepsilon.$

    We start by defining a sequence of quotient groups of $G$ as follows. We start with $G_0 = G$. Having defined $G_i$, we define $G_{i+1}$ by fixing any element $h_i\in G_i$ of prime order $p_i\in \mathcal{S}$ (such an element always exists in an Abelian group by a trivial argument, but one can also use Cauchy's theorem (see e.g. \cite{Conrad-CauchyThm}) which applies to all groups) and defining $G_{i+1} = G_i/H_i$ where $H_i$ is the group generated by $h_i$. We stop when the group $G_i$ is the trivial group containing just the identity element.

    Assume that the sequence of groups thus constructed is $G_0=G,\ldots, G_h = \{0\}.$ 
    We define a sequence of functions $f_i:\{0,1\}^n\rightarrow G_i$ ($i\in \{0,\ldots, h\}$) inductively by defining $f_0 = f$ and $f_{i+1}$ by
    \[
    f_{i+1}({\bf a}) = f_i({\bf a}) \pmod{G_i}.
    \]
    for each ${\bf a}\in \{0,1\}^n.$ We also define $\mathsf{List}^{f_i}_\varepsilon$ to be a subset of $\mathcal{P}_d(\{0,1\}^n, G_i)$ in the natural way.

    Finally, we define a tree $T$ of height $h$ as follows. For each $i\in \{0,\ldots, h\},$ we add one vertex at level $i$ in $T$ for each $P_i\in \mathsf{List}^{f_i}_\varepsilon$; note that, in particular, there is exactly one vertex at level $h$. Further, the children of the polynomial $P_{i+1}$ at level $i+1$ in $T$ are (the vertices corresponding to) those polynomials $P_{i}$ at level $i$ such that 
    \[
    P_{i}({\bf x}) = P_{i+1}({\bf x}) \pmod{H_i}.
    \]
    It will be useful to note that by \Cref{thm:basic}, the above equality holds both in terms of the evaluations of the two polynomials and also in terms of their coefficients.

    It is easy to see that each $P_{i}$ at level $i$ is a child of a unique $P_{i+1}$ at level $i+1$, namely the polynomial $P_{i+1}$ defined by the equality above. We thus have indeed defined a tree $T$. The size of $\mathsf{List}^f_\varepsilon$ is upper bounded by the number of leaves (or paths) of $T$, which we bound using the following claim.

    \paragraph{Notation.} Recall that each vertex $v$ of the tree is associated to a polynomial $P$ over group $G_i$ where $i$ denotes the level of $v$. We define by $\mathrm{agr}(v)$ the fraction of points of agreement between $P$ and $f_i$. Further, let $\rho(v) = \mathrm{agr}(v) - \left(1-\frac{1}{2^d}\right).$

    \begin{claim}
        \label{clm:treeprops}
        The tree $T$ defined above has the following properties.
        \begin{enumerate}
            \item Each non-leaf vertex in $T$ has at most $L$ children. 
            \item If $v$ is a child of $u$, then $\rho(u) \geq \rho(v).$ 
            \item If $u$ has two distinct children $v$ and $w$, then $\rho(u) \geq \rho(v) + \rho(w).$
        \end{enumerate}
    \end{claim}

    Using the above claim, we can finish the proof of the lemma as follows. Let $\ell = \lceil \log L\rceil$. We show that for any node $u$ with children $v_1,\ldots, v_t$, we have
    \begin{equation}
        \label{eq:DGKS}
        \rho(u)^\ell \geq \rho(v_1)^\ell + \cdots + \rho(v_t)^\ell.
    \end{equation}
    Assuming this, we get by induction from the root $r$ of $T$ that 
    \[
    \rho(r)^\ell \geq \sum_{\text{$v$ a leaf}} \rho(v)^\ell \geq \varepsilon^\ell \cdot \text{(\# of leaves of $T$)} \geq \varepsilon^\ell \cdot |\mathsf{List}^f_\varepsilon|
    \]
    where for the second inequality, we used the fact that $\rho(v) \geq \varepsilon$ for all nodes $v$ in the tree (by definition of $\rho$ and the properties of the polynomials associated to each $v$). Since $\rho(r) \leq 1,$ we immediately get $|\mathsf{List}^f_\varepsilon| \leq (1/\varepsilon)^\ell = L^{\bigO(\log 1/\varepsilon)}$ as desired.
    
    It remains to prove \Cref{eq:DGKS} and \Cref{clm:treeprops}, which we do in that order. To see \Cref{eq:DGKS}, we assume without loss of generality that $t\geq 2$ (otherwise the inequality is trivial) and that $\rho(v_1) \geq \rho(v_2) \geq \cdots \geq \rho(v_t).$ Then we have
    \begin{align*}
        \rho(v_1)^\ell + \cdots  + \rho(v_t)^\ell \leq (\rho(v_1) + \rho(v_2))^\ell \leq \rho(u)^\ell
    \end{align*}
    where the second inequality is a consequence of Item 3 of \Cref{clm:treeprops} and the first inequality follows by examining the expansion of $(\rho(v_1) + \rho(v_2))^\ell$ which contains one term that is $\rho(v_1)^\ell$ and $2^{\ell}-1 \geq L-1\geq t-1$ (\Cref{clm:treeprops} Item 1) terms that are at least $\rho(v_2)^\ell.$

    \begin{proof}[Proof of \Cref{clm:treeprops}]
        Assume $u$ is a vertex at some level $i+1$ in the tree. Then $u$ corresponds to a polynomial $P$ taking values in the group $G_{i+1}.$ Assume that $u$ has children $v_1,\ldots, v_t$. Each child $v_j$ is associated to a polynomial $Q_j$ over $G_{i}.$ 

        Let $A(u)$ denote the set of points where $P$ agrees with $f_{i+1}$ and similarly, let $A(v_j)$ denote the set of points where $Q_j$ agrees with $f_{i}.$ By the definition of the tree $T$, we have $A(v_j)$ is contained in $A(u)$ for each $j$ and this implies Item 2 of the claim.

        To prove the other two items, we need some notation. Let $h_{i}\in G_{i}$ be the group element used to define $H_{i}$ and hence $G_{i+1}$ above. Fix any system of coset representatives $c_1,\ldots, c_M$ of $H_{i}$ in $G_{i}.$ Given any $g\in G_i,$ we can write $g$ uniquely as
        \begin{equation}
        \label{eq:cosetrepn}
        g = \hat{g} + g'\cdot h_i
        \end{equation}
        where $\hat{g}\in \{c_1,\ldots, c_M\}$ and $g'\in \{0,\ldots, p_i-1\}$ (since $h_i$ has order $p_i$). 

        Hence, for each $j\in [t],$ we can write 
        \[
        Q_j({\bf x}) = \sum_{|I|\leq d} \alpha_I x^I = \underbrace{\sum_{|I|\leq d} \hat{\alpha}_I x^I}_{\hat{Q}_j({\bf x})} + \underbrace{\left(\sum_{|I|\leq d} \alpha'_I x^I\right)}_{{Q}_j'({\bf x})}\cdot h_i.
        \]
        Since $h_i$ has order $p_i$, we can think of $Q_j'$ as taking values in $\mathbb{Z}_{p_i}.$ 

        We note for any two $j,k\in [t],$ the polynomials $\hat{Q}_j$ and $\hat{Q}_k$ are in fact exactly the same. This is because 
        \[
         P({\bf x}) = Q_j({\bf x}) = Q_k({\bf x}) \pmod{H_i}
        \]
        by definition of the tree $T$, implying that the coefficients of $Q_j$ and $Q_k$ are all the same modulo $H_i$. By the uniqueness of the decomposition in \Cref{eq:cosetrepn}, the polynomials $\hat{Q}_j$ and $\hat{Q}_k$ must therefore be the same. We denote this polynomial $\hat{Q}$. Note that this also implies that $Q_j'$ and $Q_k'$ are distinct for any distinct $j,k\in [t]$, since otherwise $Q_j = Q_k$.

        Let $F_i({\bf x}) = f_i({\bf x}) - \hat{Q}({\bf x}).$ In a similar way to what we did above, we can write for each ${\bf x}\in \{0,1\}^n,$
        \[
        F_i({\bf x}) = \hat{F}_i({\bf x}) + F_i'({\bf x})\cdot h_i
        \]
        where again we think of $F_i'$ as taking values in $\mathbb{Z}_{p_i}.$

        Fix any ${\bf x}\in \{0,1\}^n.$ We have 
        \begin{align*}
            {\bf x}\in A(v_j) \Longleftrightarrow f_i({\bf x}) = Q_j({\bf x})
            \Longleftrightarrow 
            F_i({\bf x}) = Q_j'({\bf x})\cdot h_i
            \Longrightarrow
            F_i'({\bf x}) = Q_j'({\bf x})
        \end{align*}
        where the last equality holds as elements of $\mathbb{Z}_{p_i}.$ 
        In particular, this implies that since $Q_j\in \mathsf{List}^{f_i}_\varepsilon$, we must also have $Q_j'\in \mathsf{List}^{F_i'}_\varepsilon$. Using the hypothesized bound of $L$ from the prime case (statement of \Cref{lem:dgksvariant}), we get $t\leq L$ implying the first item of the claim.

        Finally, we note using the same reasoning that if ${\bf x}\in A(v_j)\cap A(v_k)$ for distinct $j$ and $k,$ we must have $Q_j'({\bf x}) = Q_k'({\bf x})$. Since distinct degree-$d$ polynomials can agree on at most a $(1-\frac{1}{2^d})$ fraction of inputs (\Cref{thm:basic}), we see that $|A(v_j)\cap A(v_k)|\leq 2^n\cdot \alpha$ where $\alpha$ denotes $(1-\frac{1}{2^d}).$ We have
        \begin{align*}
        (\rho(u) + \alpha)\cdot 2^n = |A(u)| &\geq |A(v_j) \cup A(v_k)| = |A(v_j)| + |A(v_k)|  - |A(v_j)\cap A(v_k)|\\ 
        &\geq (\rho(v_j)+\alpha)\cdot 2^n + (\rho(v_k) + \alpha)\cdot 2^n - \alpha\cdot 2^n = (\rho(v_j)+\rho(v_k) + \alpha)\cdot 2^n
        \end{align*}
        where we used the fact that $A(u)$ contains $A(v_j)$ and $A(v_k)$ (argued above) for the first inequality. This proves Item 3 and concludes the proof.
    \end{proof}

\end{proof}

\section{Local list correction}
\label{sec:locallistcorrection}

In this section, we design a local list corrector (see \Cref{defn:local-list-algo}) for the class $\mathcal{P}_{d}$ and prove \Cref{thm:listdecoding}. 

Let $G$ be an Abelian group and $f: \Boo^{n} \to G$ be any function with oracle access to it. Let $\mathsf{List}_{\varepsilon}^{f}$ denote the set of degree $d$ polynomials that are $(1/2^{d} - \varepsilon)$-close to $f$, and let $L(\varepsilon) = |\mathsf{List}_{\varepsilon}^{f}|$. Recall from \Cref{thm:comblistdegd} that $L(\varepsilon) =\exp(\bigO_d(1/\varepsilon)^{\bigO(d)})$. 

Our local list corrector has two phases, inspired by previous work~\cite{GoldreichL, STV-list-decoding, ABPSS24-ECCC}.
\begin{itemize}
    \item First we construct $L' = \widetilde{\bigO}(L(\varepsilon/2))$ algorithms with oracle access such that with high probability, for each polynomial in the list $\mathsf{List}^{f}_{\varepsilon}$, there exists an algorithm that is $\delta$-close to the polynomial where $\delta < 1/(10 \cdot 2^{d+1})$, i.e. $\delta$ is in the unique decoding regime.
    \item Secondly we apply the (unique) local corrector for the class $\mathcal{P}_{d}$ from \Cref{thm:uniquedegd} on each of the algorithms from the first phase.
\end{itemize}

We state the main theorem of this section below that describes the first phase of our local list corrector. Recall that $\mathcal{A}^{f}$ denotes that the algorithm $\mathcal{A}$ has oracle access to the function $f$.\\

\begin{thmbox}
\begin{theorem}[Approximate oracles]\label{thm:approx-oracles-list-decoding}
Fix $n \in \mathbb{N}$, $\varepsilon > 0$. Let $f: \Boo^{n} \to G$ be any function. There exists a randomized algorithm $\mathcal{A}_{1}^{f}$ that makes at most $\bigO_{\varepsilon}(1)$ oracle queries and outputs deterministic algorithms $\Psi_{1},\ldots, \Psi_{L'}$ satisfying the following property: with probability at least $3/4$, for every polynomial $P \in \mathsf{List}_{\varepsilon}^{f}$, there exists a $j \in [L']$ such that $\delta(\Psi_{j}, P) < 1/(10 \cdot 2^{d+1})$, and moreover, for every $\mathbf{x} \in \Boo^{n}$, $\Psi_{j}$ computes $P(\mathbf{x})$ by making at most $\bigO_{\varepsilon}(1)$ oracle queries to $f$. Here $L' = \bigO(L(\varepsilon/2)\log L(\varepsilon/2)) = \bigO_\varepsilon(1)$.
\end{theorem}
\end{thmbox}

\noindent
The algorithm $\mathcal{A}_1$ is in fact nearly identical to a similar algorithm from \cite{ABPSS24-ECCC} for the case $d=1$. The main novelty in this paper is in extending the analysis of the algorithm to the case of larger degrees. In particular, we prove the following technical lemmas, which we believe are independently interesting.
\begin{itemize}
    \item \textbf{A sampling lemma for Hamming slices using subcubes}: Let $k$ be an even integer and $\mathsf{C}$ be a subcube of $\{0,1\}^{2k}$ be obtained by partitioning the $2k$ co-ordinates into $k$ pairs uniformly at random and identifying the co-ordinates in each pair (equivalently, we choose a random hash function $h:[2k]\rightarrow [k]$ that is $2$-to-$1$ and consider the subcube $C_{0^{2k},h}$ as defined in \Cref{sec:prelims}). We show that this process has good sampling properties in the sense that the density of a set $S\subseteq \{0,1\}^{2k}$ of vectors of Hamming weight $k$ is roughly preserved in the subcube $\mathsf{C}$. More formally, we show
    \[
    \Pr_{\mathsf{C}}{\left[\left|\frac{|S \cap \mathsf{C}|}{\binom{k}{k/2}} - \frac{|S|}{\binom{2k}{k}}\right|\geq \frac{1}{k^{\Omega(1)}} \right]} \leq \frac{1}{k^{\Omega(1)}}.
    \]
    See \Cref{subsubsec:sampling} for a formal statement and the proof.
    \item \textbf{DLSZ lemma for low-degree polynomials over Hamming slices}: We also give a simple proof of the fact that if a polynomial $P\in \mathcal{P}_d(\{0,1\}^{2k},G)$ does not vanish on the set of points of Hamming weight $k$, then it does not vanish at an $\Omega_d(1)$ fraction of it. While there have been many works (see \cite{MuraliSrinivasan-SymmetricChain, Filmus2014AnOB, Filmus-Ihringer19, Filmus-JuntaSlices}) addressing the properties of polynomial functions over slices in recent years (especially over the reals), we do not know if this fact has appeared in the literature before in this generality.\snote{Literature survey to check if this is true? And has this appeared in some form before over the reals? Ask Yuval Filmus?} See \Cref{sec:DLSZ-slice} for a formal statement and the proof.
\end{itemize}

Putting the above two lemmas together, one can easily derive the following consequence. Fix a degree-$d$ polynomial $R\in \mathcal{P}_d(\{0,1\}^{2k}, G)$ that does not vanish on the Hamming slice of weight $k.$ Then, the probability that it vanishes on a random subcube $\mathsf{C}$ chosen as above is very small, taking a sufficiently large $k$. This corollary will play a crucial role in the proof of \Cref{thm:approx-oracles-list-decoding}. The analogous statement for degree $1$ was proved in \cite{ABPSS24-ECCC}. However, the proof there is based on a strategy that utilizes an understanding of the structure of the polynomial $R$ in a way that seems difficult to implement for higher degrees.

We start with the proofs of the above lemmas in \Cref{sec:locallistcorrectionnew} and then give the proofs of \Cref{thm:approx-oracles-list-decoding} and \Cref{thm:listdecoding} in subsequent sections (which closely follow the analogous proofs in \cite{ABPSS24-ECCC} modulo the above facts).

\subsection{The main technical lemmas}
\label{sec:locallistcorrectionnew}

We prove now the main new technical lemmas that we use to prove \Cref{thm:approx-oracles-list-decoding}. Throughout this section, we use $\{0,1\}^n_{m}$ to denote strings in $\{0,1\}^n$ of Hamming weight exactly $m$.

\subsubsection{Sampling Hamming slices via subcubes}\label{subsubsec:sampling}

Throughout this section, we fix an even positive integer $k$ and consider a random $k$-dimensional subcube $\mathsf{C}\subseteq \{0,1\}^{2k}$ obtained as follows. We partition the $2k$ coordinates into $k$ pairs uniformly at random and identify the coordinates in each pair. Equivalently, we choose a uniformly random map $h:[2k]\rightarrow [k]$ that is $2$-to-$1$, and set $\mathsf{C} = C_{0^{2k}, h}$ (where the latter subcube is as defined in \Cref{defn:random-embedding}).

\begin{thmbox}
\begin{lemma}[Sampling lemma for Hamming slices]
    \label{lem:sampling}
    There is an absolute constant $\eta > 0$ such that for every set $S\subseteq \{0,1\}^{2k}_k$, we have
    \[
    \Pr_{\mathsf{C}}{\left[\left|\frac{|S|}{\binom{2k}{k}} - \frac{|S\cap \mathsf{C}|}{\binom{k}{k/2}}\right| \geq \frac{1}{k^\eta}\right]} \leq \bigO\left(\frac{1}{k^\eta}\right).
    \]
\end{lemma}
\end{thmbox}

\begin{remark}
    \label{rem:samplingtightness}
    The above lemma is seen to be tight up to the value of the constant $\eta$. Consider $S\subseteq \{0,1\}^{2k}_k$ containing exactly those points that differ in the first two co-ordinates. It is easily seen that $|S| = \Omega\left(\binom{2k}{k}\right)$ and further $S\cap \mathsf{C} = \emptyset$ whenever the random partitioning defining $\mathsf{C}$ pairs the first two co-ordinates with each other. The latter event occurs with probability $\Omega(1/k).$
\end{remark}

The proof is via a standard second moment bound, the main step of which is a bound on the spectral gap of a suitable graph. We start with some notations and then state this bound. 

An undirected graph $G$ is an \emph{$(N,D,\lambda)$-expander} if it is a $D$-regular graph on $N$ vertices and the magnitude of the second-largest eigenvalue of its adjacency matrix (in absolute value) is at most $\lambda\cdot D$. We refer to the monograph of Hoory, Linial and Wigderson~\cite{HLW} for a more thorough treatment of expander graphs. We will use primarily the following lemma, due to Alon and Chung~\cite{Alon-Chung-EML}.

\begin{lemma}[Expander Mixing Lemma]
\label{lem:EML}
    Let $G=(V,E)$ be an $(N,D,\lambda)$-expander and let $S,T\subseteq V$ be any sets of density $\sigma$ and $\tau$ respectively. Then, for $u$ a uniformly random vertex of $V$ and $v$ a uniformly random neighbour of $u$, we have
    \[
    \left|\Pr_{u,v}\left[u\in S\wedge v\in T\right] - \sigma\cdot \tau\right| \leq \lambda.
    \]
\end{lemma}

We apply the above lemma to a combinatorially defined graph called the \emph{Johnson graph}. We define the undirected graph $J(2k,k,d)$ with vertex set $V = \{0,1\}^{2k}_k$, where two points ${\bf a}, {\bf b} \in V$ are connected exactly when their Hamming distance $\Delta({\bf a}, {\bf b}) = 2d$~\footnote{Not to be confused with the degree of polynomials $d$ used in the other sections of the paper.} (note that this number is always even). The main technical lemma is the following.

\begin{thmbox}
\begin{lemma}[Eigenvalues of the Johnson graph]
    \label{lem:Johnson-eigenvalue}
    Let $J(2k,k,d)$ be defined as above and assume that $d$ is such that $|d - k/2|\leq C\sqrt{k\log k}$ for $C > 0$ being an absolute constant. Then $J(2k,k,d)$ is a $(\binom{2k}{k}, \binom{k}{d}^2, \frac{A}{k^\eta})$-expander for some absolute constants $A > 0$ and $\eta \in (0,1)$ depending only on $C.$
\end{lemma}
\end{thmbox}

We start by proving the sampling lemma \Cref{lem:sampling} assuming \Cref{lem:Johnson-eigenvalue}.

\begin{proof}[Proof of \Cref{lem:sampling}]
    Throughout, let $\sigma$ denote the density of $S$ inside $\{0,1\}^{2k}_k.$

    Recall that $\mathsf{C} = C_{0^{2k},h}$ contains a unique point $x({\bf y})\in \{0,1\}^{2k}$ for each point ${\bf y}\in \{0,1\}^k.$ Define an indicator random variable $Z_{{\bf y}}$ that is $1$ exactly when $x({\bf y})\in S.$ We note that for each ${\bf y}\in \{0,1\}^k_{k/2},$ the point $x({\bf y})$ is a uniformly random element of $\{0,1\}^{2k}_k,$ implying that $\mathbb{E}[Z_{\bf y}] = \sigma$ for each ${\bf y}\in \{0,1\}^k_{k/2}.$

    Let $Z := |S\cap \mathsf{C}| = \sum_{{\bf y}\in \{0,1\}^{k}_{k/2}} Z_{\bf y}.$ We thus have $\mathbb{E}[Z] = \binom{k}{k/2}\cdot \sigma.$

    We prove the lemma via a second moment estimate. The variance of $Z$ can be bounded as follows.
    \begin{align*}
    \mathrm{Var}[Z] \, = \, \sum_{\mathbf{u}, \mathbf{v} \in \Boo^{k}_{k/2}} \mathrm{Cov}[Z_{\mathbf{u}}, Z_{\mathbf{v}}]
\end{align*}
We divide the above sum into two parts depending on $\Delta(\mathbf{u}, \mathbf{v})$. Let $I = [k/2 - C\sqrt{k\log k}, k/2 + C\sqrt{k\log k}]$ for a suitably large absolute constant (to be chosen below).
\begin{align*}
    \mathrm{Var}[Z]  &= \,  \sum_{\substack{\mathbf{u}, \mathbf{v} \in \Boo^{k}_{k/2} \\ \Delta(\mathbf{u}, \mathbf{v}) \in I}} \mathrm{Cov}[Z_{\mathbf{u}}, Z_{\mathbf{v}}] \, + \, \sum_{\substack{\mathbf{u}, \mathbf{v} \in \Boo^{k}_{k/2} \\ \Delta(\mathbf{u}, \mathbf{v}) \notin I}} \mathrm{Cov}[Z_{\mathbf{u}}, Z_{\mathbf{v}}] \\
    &\leq \sum_{\substack{\mathbf{u}, \mathbf{v} \in \Boo^{k}_{k/2} \\ \Delta(\mathbf{u}, \mathbf{v}) \in I}} \mathrm{Cov}[Z_{\mathbf{u}}, Z_{\mathbf{v}}] \, + \, \binom{k}{k/2}^2\cdot \frac{1}{k^{\Omega(C)}},
\end{align*}
where the last inequality follows from standard concentration bounds for sampling without replacement \cite{Hoeffding}. Now it remains to upper bound $\mathrm{Cov}[Z_{\mathbf{u}}, Z_{\mathbf{v}}]$ for a pair $(\mathbf{u}, \mathbf{v}) \in \Boo^{k}_{k/2} \times \Boo^{k}_{k/2}$ such that $\Delta(\mathbf{u}, \mathbf{v}) \in I$. Fix such a pair $(\mathbf{u}, \mathbf{v})$ and note that $\Delta(\mathbf{u}, \mathbf{v}) = 2d$ for some integer $d$ such that $|d-(k/4)|\leq (C/2)\cdot\sqrt{k\log k}.$ Then we have,
\begin{align*}
    \mathrm{Cov}[Z_{\mathbf{u}}, Z_{\mathbf{v}}] \, = \, \mathbb{E}[Z_{\mathbf{u}}  Z_{\mathbf{v}}] - \mathbb{E}[Z_{\mathbf{u}}] \mathbb{E}[Z_{\mathbf{v}}] \, = \, \Pr[x(\mathbf{u}) \in S \wedge x(\mathbf{v}) \in S] - \sigma^2.
\end{align*}
Moreover, the distribution of the random point $x(\mathbf{v})$ given $x(\mathbf{u)}$ can be checked to be the uniform distribution over the neighbours of $x(\mathbf{u})$ in the graph $J(2k,k,2d).$\footnote{More precisely, $x(\mathbf{u})$ and $x(\mathbf{v})$ differ exactly at the co-ordinates given by $h^{-1}(D)$ where $D\subseteq [k]$ is the set of $2d$ co-ordinates where $\mathbf{u}$ and $\mathbf{v}$ differ. Given $x(\mathbf{u})$, this is a uniformly random set $D'\subseteq [2k]$ of size $4d$ subject to the constraint that symmetric difference of $D'$ and the support of $x(\mathbf{u})$ has size exactly $4d$.} We can thus apply \Cref{lem:Johnson-eigenvalue} to see that 
\[
\mathrm{Cov}[Z_{\mathbf{u}}, Z_{\mathbf{v}}] \leq \frac{A}{k^{\eta}}
\]
for some absolute constants $A> 0$ and $\eta \in (0,1)$ depending on $C.$ Continuing the variance computation above, we get
\[
\mathrm{Var}[Z] \; \leq \;  \frac{\binom{k}{k/2}^{2} \cdot A }{k^\eta} \, + \, \frac{\binom{k}{k/2}^2}{k^{\Omega(C)}} \; \leq \; \frac{\bigO(1)}{k^\eta}\cdot \binom{k}{k/2}^2
\]
for a large enough choice of the constant $C$ (and the corresponding $\eta$). Now, by Chebyshev's inequality, we have
\[
\Pr_{\mathsf{C}}\left[\left|Z - \sigma\cdot \binom{k}{k/2}\right|\geq \frac{1}{k^{\eta/4}}\cdot \binom{k}{k/2} \right] \; \leq \; \frac{\mathrm{Var}[Z]}{\frac{1}{k^{\eta/2}}\cdot \binom{k}{k/2}^2} \; \leq \; \frac{\bigO(1)}{k^{\eta/2}}
\]
which implies the statement of \Cref{lem:sampling}.
\end{proof}

It remains to prove \Cref{lem:Johnson-eigenvalue}, which we do below.

\begin{proof}[Proof of \Cref{lem:Johnson-eigenvalue}]
    We use the known exact expressions for the eigenvalues of the Johnson graphs~\cite{Delsarte}. In particular, following \cite[Corollary 2]{Karloff-GW-analysis}, we know that the eigenvalues of the adjacency matrix of $J(2k,k,d)$ are $\beta_{0},\ldots,\beta_{k}$, where for every $0 \leq s \leq k$,\footnote{The bound in~\cite{Karloff-GW-analysis} looks slightly different than what is stated here, since the parameter $b$ in that statement is the size of the intersection of the supports of the two points, which implies that $b = k-d$.}
    \begin{equation}
        \label{eq:johnson-evals}
        \beta_{s} = \sum_{r = 0}^{s} (-1)^{s-r} \binom{s}{r} \binom{k-r}{k-d-r} \binom{k-s+r}{d-s+r}
    \end{equation}
    Furthermore, for $0 \leq s \leq k$, the eigenvalue $\beta_{s}$ has multiplicity $\binom{2k}{s} - \binom{2k}{s-1}$.

    Clearly, $\beta_0 = \binom{k}{d}^2$, which is equal to the degree of the graph. It suffices to show that the magnitude of the remaining eigenvalues are all at most $\frac{\bigO(1)}{k^\eta}\cdot \beta_0,$ where $\eta$ is as in the statement of the \Cref{lem:Johnson-eigenvalue}. By assuming that the constant factor in the $\bigO(1)$ term is large enough, we may assume that $k$ is greater than a large enough absolute constant.

    We split the analysis of the remaining eigenvalues into two regimes, the first being when $1 \leq s \leq C_{1}$ for a suitably large constant $C_1$ depending on $C$ (chosen below), and the second when $s > C_1.$

    \paragraph{Case 1: $1 \leq s \leq C_1$.} In this case, we analyze $\beta_s$ using \Cref{eq:johnson-evals}. We need the following simple fact about binomial coefficients, which is easily verified from the standard definition using factorials.

    \begin{fact}
        \label{fac:binomials}
        Fix any non-negative integers $r\leq \ell\leq k$. Then
        \[
        \binom{k}{\ell}\cdot\left(\frac{\ell-r}{k-r}\right)^{r}  \leq \binom{k-r}{\ell-r} \leq \binom{k}{\ell}\cdot \left(\frac{\ell}{k}\right)^r.
        \]
    \end{fact}

    In particular, we note that for $\ell \in \{d,k-d\}$ and $r\leq s \leq C_1$, we have 
    \[
     \frac{1}{2}\left(1-\frac{1}{k^{1/4}}\right) \leq \frac{\ell-r}{k-r}\leq \frac{\ell}{k} \leq \frac{1}{2}\cdot \left(1+\frac{1}{k^{1/4}}\right),
    \]
    and thus for large enough $k$, we have by \Cref{fac:binomials}
    \[
    \binom{k-r}{k-d-r}= \binom{k}{k-d}\cdot \frac{1}{2^r}\cdot (1\pm \gamma) \ \ \text{and} \ \ \binom{k-s+r}{d-s+r} = \binom{k}{d}\cdot \frac{1}{2^{s-r}}\cdot (1\pm \gamma)
    \]
    for $\gamma= \bigO(s\cdot k^{-1/4})$. Here the notation $a = b\cdot (1\pm \gamma)$ denotes that $a\in [b\cdot (1-\gamma), b\cdot (1+\gamma)].$ Plugging this into \Cref{eq:johnson-evals} above, we see that
    \begin{align*}
        |\beta_s| &= \left|\sum_{r=0}^s (-1)^{s-r}\binom{s}{r} \cdot \frac{1}{2^r}\cdot \binom{k}{k-d}\cdot (1\pm \gamma)\cdot \frac{1}{2^{s-r}}\cdot\binom{k}{d}\cdot (1\pm \gamma)\right|\\
        &\leq \binom{k}{d}^2 \cdot\frac{1}{2^s}\cdot \left|\sum_{r=0}^s (-1)^{s-r}\binom{s}{r}\cdot (1\pm \bigO(\gamma))\right| \\
        &\leq \binom{k}{d}^2 \cdot\frac{1}{2^s}\cdot \left(\left|\sum_{r=0}^s (-1)^{s-r}\binom{s}{r}\right| + \bigO(\gamma)\cdot \sum_{r=0}^s \binom{s}{r}
        \right) \leq \binom{k}{d}^2 \cdot\frac{1}{2^s}\cdot \bigO(\gamma)\cdot 2^s =  \binom{k}{d}^2\cdot \bigO(\gamma).
    \end{align*}
    We have thus shown that $|\beta_s|\leq |\beta_0|\cdot \frac{\bigO(1)}{k^{1/4}}$, proving the required bound in this case.

    \paragraph{Case 2: $s > C_1$.} In this case, we use an argument from \cite{BCIM}, which utilizes just the multiplicity $m_s = \binom{2k}{s} - \binom{2k}{s-1}$ of the eigenvalue $\beta_s$. We note that $m_s \geq \frac{1}{2}\cdot \binom{2k}{C_1}$ as long as $k$ is large enough.

    Let $M$ denote the adjacency matrix of the Johnson graph $J(2k,k,d).$ We know that the squared Frobenius norm of $M$ (i.e. the sum of the squares of the entries of $M$) is the sum of the squares of its eigenvalues and hence at least $\beta_s^2 m_s.$ On the other hand, since $M$ is an adjacency matrix, this quantity is simply the number of edges in the graph $J(2k,k,d)$. In particular, we have
    \[
    \beta_s^2 m_s \leq \binom{2k}{k}\cdot \binom{k}{d}^2,
    \]
    yielding the bound
    \[
    |\beta_s| \leq \frac{\bigO(1)}{\binom{k}{C_1}}\sqrt{\binom{2k}{k}}\cdot \binom{k}{d}\leq \frac{\bigO(2^k)}{\binom{k}{C_1}}\cdot \binom{k}{d}.
    \]
    Recall that $|d-k/2|\leq C\sqrt{k\log k}$, implying that $\binom{k}{d} \geq 2^k\cdot k^{-\bigO(C^2)}$ by standard binomial estimates.\footnote{for example, when $d \leq k/2$, we can use \cite[Lemma 4]{KY} with $p=1/2$ to lower bound $\frac{1}{2^k}\sum_{i\leq d}\binom{k}{i}$ which itself is at most $\frac{\bigO(k)}{2^k}\cdot \binom{k}{d}$} Thus, for $C_1$ large enough in comparison with $C$ and for $k$ large enough, we see that $|\beta_s|\leq \frac{1}{k}\cdot \binom{k}{d}^2  = \frac{\beta_0}{k}$. This finishes the analysis of Case 2 and hence finishes the proof of the lemma.
\end{proof}
 
\subsubsection{DLSZ lemma over Hamming slices}
\label{sec:DLSZ-slice}

\begin{thmbox}
\begin{lemma}[A DLSZ lemma over Hamming slices]
    \label{lem:DLSZ-slice}
    The following holds for any non-negative integers integers $n, k, d$ where $k\leq n$ and $d\leq \min\{k,n-k\}.$ Let $R\in \mathcal{P}_d(\{0,1\}^n, G)$ be a polynomial such that $R$ does not vanish at some point in $\{0,1\}^n_k$. Then 
    \[
    |\{{\bf a}\in \{0,1\}^n_k\ |\ R({\bf a}) \neq 0\}| \geq \binom{n-2d}{k-d}.
    \]
\end{lemma}
\end{thmbox}

\begin{proof}
    As is standard, we proceed by induction on $d.$ The base case $d = 0$ is trivial.

    Fix $d \geq 1.$ Let $R(x_1,\ldots, x_n)$ be a degree $d$ polynomial that does not vanish on all of $\{0,1\}^n_k.$ 
    
    We can assume that $d$ is \emph{strictly smaller} than $\min\{k,n-k\}$ by the following reasoning. If $d=\min\{k,n-k\},$ the claim reduces to showing that $R$ is non-zero at least one point of $\{0,1\}^n_k$, which trivially follows from the hypothesis. We can also assume that $R$ does not evaluate to the same non-zero value on all of $\{0,1\}^n_k$ since otherwise, the lemma is trivially true.
    
    Without loss of generality, we assume that $R({\bf p}) \neq R({\bf q})$ where ${\bf p}\in \{0,1\}^n_k$ is the point where the first $k$ co-ordinates are $1$ and ${\bf q}$ is obtained from ${\bf p}$ by flipping the first and last co-ordinates to $0$ and $1$ respectively (note that this makes sense as $1\leq d \leq \min\{k,n-k\}$).
    
    By replacing $x_n$ by the linear polynomial $k - \sum_{i < n} x_i$, we get a degree-$d$ polynomial $R'(x_1,\ldots, x_{n-1})$ not involving the variable $x_n$ but that nevertheless evaluates to the same value as $R$ at every point in $\{0,1\}^n_k.$ Moreover, we can assume that $R'$ is multilinear since the inputs are Boolean.

    We write 
    \[
    R' = x_1 P(x_2,\ldots, x_{n-1}) + Q(x_2,\ldots, x_{n-1}).
    \]
    where $P$ has degree at most $d-1.$ Note that $P({\bf p}')\neq 0$ where ${\bf p}'\in \{0,1\}^{n-2}_{k-1}$ is the point obtained by restricting ${\bf p}$ to its middle $n-2$ co-ordinates, since otherwise $R'({\bf p}) = R'({\bf q})$ in contradiction to our assumptions above. 
    
    Thus, we see that $P$ does not vanish on all of $\{0,1\}^{n-2}_{k-1}$. Applying the inductive hypothesis, we see that for
    \[
    S' = \{{\bf a}' \in \{0,1\}^{n-2}_{k-1}\ |\ P({\bf a}')\neq 0\},
    \]
    we have $|S'|\geq \binom{n-2-2(d-1)}{k-1-(d-1)} = \binom{n-2d}{k-d}.$

    For each ${\bf a}'\in S'$, we note that we get at least one distinct point ${\bf a}\in \{0,1\}^n_k$ such that $R'({\bf a})$ is non-zero. This is because setting the variables $x_2,\ldots, x_{n-1}$ according to ${\bf a}'$ restricts the polynomial $R'$ to a non-zero linear polynomial, which must be non-zero at least one of the extensions of ${\bf a}'$ to a point ${\bf a}\in \{0,1\}^n_k.$

    This shows that $R'$ (and hence $R$) is non-zero at at least $\binom{n-2d}{k-d}$ many points in $\{0,1\}^n_k,$ proving the inductive claim.
\end{proof}

\subsubsection{A useful corollary}
\label{sec:usefulcorollary}

We will use \Cref{lem:sampling} and \Cref{lem:DLSZ-slice} in the form of the following corollary.

\begin{corollary}
    \label{cor:samplingDLSZ}
    Fix a degree parameter $d.$ Let $k$ be an even positive integer such that $k\geq d,$ and $R\in \mathcal{P}_d(\{0,1\}^{2k},G)$ be such that $R$ computes a non-zero function on the slice $\{0,1\}^{2k}_k$. Let $h:[2k]\rightarrow [k]$ be a uniformly random $2$-to-$1$ function and define $\mathsf{C} = C_{0^{2k},h}$ (the notation is from \Cref{defn:random-embedding}). Then if $R|_{\mathsf{C}}$ denotes the natural restriction of $R$ to the subcube $\mathsf{C}$ as a polynomial in $k$ variables, we have
    \[
    \Pr_{\mathsf{C}}[{\text{$R|_{\mathsf{C}}$ vanishes on all of $\{0,1\}^k_{k/2}$}}] \leq \bigO_d\left(\frac{1}{k^{\eta}}\right)
    \]
    for some absolute constant $\eta > 0.$
\end{corollary}

\begin{proof}
    Let $S \subseteq \{0,1\}^{2k}_k$ denote the set of points of $\{0,1\}^{2k}_k$ where $R$  does not vanish. By \Cref{lem:DLSZ-slice}, we know that $|S|\geq \binom{2k-2d}{k-d} = \frac{1}{2^{\bigO(d)}}\cdot \binom{2k}{k}.$ Note that $R|_{\mathsf{C}}$ vanishes on $\{0,1\}^{k}_{k/2}$ exactly when $S\cap \mathsf{C} = \emptyset.$ By \Cref{lem:sampling}, the probability of the latter event is at most $\bigO_d\left(\frac{1}{k^{\eta}}\right).$ 
\end{proof}

\begin{remark}
    \label{rem:slice-cor}
    We remark that the above corollary is false under the weaker assumption that $R$ is just a non-zero polynomial. A simple example is given by the linear polynomial $\sum_{i=1}^n x_i$ in the setting of $G = \F_2$.\footnote{In the linear case,~\cite{ABPSS24-ECCC} showed that this is essentially the only `bad' example. However, for degrees $2$ and higher, there are many more such examples. This is what makes \Cref{cor:samplingDLSZ} more challenging to prove in our setting.} However, the assumption that $R$ is non-zero on $\{0,1\}^{2k}_k$ eliminates this example. Interestingly, this is exactly the condition we need in the analysis of the error-reduction algorithm (\Cref{thm:approx-oracles-list-decoding}) below (the reason is essentially in~\Cref{obs:Cb-is-random} below).
\end{remark}

\subsection{The error-reduction algorithm (\Cref{thm:approx-oracles-list-decoding})}

As mentioned above, the proof now closely follows the proof of the analogous theorem in~\cite{ABPSS24-ECCC}, which in turn was based on ideas from \cite{STV-list-decoding}.
 
\subsubsection{Preliminaries and useful observations}
\label{sec:prelimslistcorrection}
In this subsection, we give an overview of the algorithms $\mathcal{A}_{1}^{f}$ and $\Psi_{1},\ldots, \Psi_{L'}$ as mentioned in \Cref{thm:approx-oracles-list-decoding}. 

\paragraph{}We first describe a combinatorial construction from \cite{ABPSS24-ECCC} that will be useful in our local list correctors. Given an embedding of a subcube $\mathsf{C}$ and a point $\mathbf{b}$, we would like to find a subcube $\mathsf{C}'$ such that $\mathsf{C}$ and $\mathbf{b}$ are contained in $\mathsf{C}'$. For completeness, we state the definition and observations from \cite{ABPSS24-ECCC} here.

\begin{definition}[Subcube spanned by $\mathsf{C}$ and $\mathbf{b}$]\label{def:bigger-subcube}
\cite[Definition 9]{ABPSS24-ECCC}. Let $\mathsf{C} = C_{\mathbf{a}, h}$ be an embedding of a subcube of dimension $k$ (see \Cref{defn:random-embedding}). For any point $\mathbf{b} \in \Boo^{n}$, let $\mathbf{v} := \mathbf{a} \oplus \mathbf{b}$. Pick a uniformly random permutation $\sigma:[2k]\rightarrow [2k]$. Define a hash function $h' : [n] \to [2k]$ as follows: For all $i \in [n]$,
\begin{align*}
    h'(i) = \begin{cases}
        \sigma(j), & \text{if } \, h(i) = j \, \text{ and } \, v_{i} = 0 \\
        \sigma(j+k), & \text{if } \, h(i) = j \, \text{ and } \, v_{i} = 1.
    \end{cases}
\end{align*}
For every $\mathbf{z} \in \Boo^{2k}$, $x(\mathbf{z})$ is defined as follows:
\begin{align*}
    x(\mathbf{z})_{i} = z_{h'(i)} \oplus a_{i}.
\end{align*}
$\mathsf{C}^{\mathbf{b}}$ is the set of points $x(\mathbf{z})$ for all $\mathbf{z} \in \Boo^{2k}$, i.e. $\mathsf{C}^{\mathbf{b}} := \setcond{x(\mathbf{z})}{\mathbf{z} \in \Boo^{2k}}$.\\
\end{definition}

\noindent
Since $h'$ refines the partition induced by $h$, $\mathsf{C} \subset \mathsf{C}^{\mathbf{b}}$. Also define $\mathbf{w} \in \Boo^{2k}$ as follows: for $j \in [k]$, $w_{\sigma(j)} = 0$ and $w_{\sigma(j+k)} = 1$. Then $x(\mathbf{w}) = \mathbf{b}$, meaning $\mathbf{b} \in \mathsf{C}^{\mathbf{b}}$. Next, we make a couple of observations that will be useful while analyzing the correctness probability of our local list correctors.

\begin{observation}\label{obs:Cb-is-random}
\cite[Observation 5.2]{ABPSS24-ECCC}. Let $\mathbf{a}$ be sampled uniformly from $\Boo^{n}$ and $h: [n] \to [k]$ be a uniformly random function. For a uniformly random $\mathbf{b} \sim \Boo^{n}$, $\mathsf{C}^{\mathbf{b}}$ as defined above is a random embedding of a subcube of dimension $2k$ (as defined in \Cref{sec:prelims}). Note that $\mathbf{b} = x(\mathbf{w})$ for some $\mathbf{w}$ of Hamming weight exactly $k.$
\end{observation}

\begin{observation}\label{obs:C-is-random}
Assume that $\mathbf{a}, \mathbf{b}\in \{0,1\}^n$ and $h:[n]\rightarrow [k]$ are chosen independently and uniformly at random and define $\mathsf{C}$ and $\mathsf{C}^{\mathbf{b}}$ as above. Conditioned on the choice of the $2k$-dimensional cube $\mathsf{C}^{\mathbf{b}}$ (which we identify with $\{0,1\}^{2k}$), we may define the distribution of $\mathsf{C}$ (which is a subcube of dimension $k$ in $\mathsf{C}^{\mathbf{b}}$) as follows. We sample a random map $\rho: [2k] \to [k]$ that is $2$-to-$1$ (i.e. for each $j \in [k]$, $|\rho^{-1}(\set{j})|$ is of size exactly $2$) and identify the variables in each pair. More formally, we set $\mathsf{C} = C_{0^{2k},\rho}$ following the notation in \Cref{defn:random-embedding}. 
\end{observation}

Finally, we will use the following \emph{non-local} list decoding algorithm from \cite{ABPSS24-ECCC}.

\begin{theorem}[\cite{ABPSS24-ECCC}, Theorem A.2]
\label{thm:non-local-list}
Fix any Abelian group $G$ and degree parameter $d.$ There is a $\poly(2^{n^{d+1}})$-time algorithm that, given oracle access to a function $f:\{0,1\}^n\rightarrow G$ produces a list of all polynomials $P\in \mathcal{P}_d$ such that $\delta(f,P) < 1/2^{d}$.
\end{theorem}

\subsubsection{Overview of the algorithms}
In this subsection, we give an overview of our local list correction algorithm for $\mathcal{P}_{d}(\Boo^{n}, G)$. As mentioned before, our algorithm is mostly similar to the local list correction algorithm for \cite[Section 5.2]{ABPSS24-ECCC}, except for some changes in the parameters to handle degree $d$ polynomials instead of degree $1$ polynomials. Nevertheless, we present the overview and the algorithm for the sake of completeness.

Similar to \cite{STV-list-decoding}, the local list correction algorithm has two key steps. The first step is to produce oracles that approximate the polynomials in the list and the second step is to apply the unique local corrector from \Cref{sec:deg-1-decoding} on each of the approximating oracles. In this subsection, we describe the first step i.e. to produce oracles such that each polynomial in the list is approximated sufficiently well by an oracle.

For the overview below, let $f(x_{1},\ldots,x_{n})$ be the input function and $\mathsf{List}_{\varepsilon}^{f}$ denotes the set as described earlier.
\begin{enumerate}
    \item \textbf{First step (advice):} We construct a randomized algorithm $\mathcal{A}_{1}^{f}$ that makes oracle queries to the input function $f$ and produces a list of oracles such that with high probability (over the randomness of $\mathcal{A}_{1}^{f}$), for every polynomial $P \in \mathsf{List}_{\varepsilon}^{f}$, there exists an oracle that is $1/(10 \cdot 2^{d+1})$-close\footnote{This step is essentially an error reduction because we start with a function $f$ that is $(1/2^{d} - \varepsilon)$-close to $P$ and $\mathcal{A}_{1}^{f}$ produces a list of oracles such that one of them is $1/(10 \cdot 2{d})$-close to $P$.} to $P$ (this is within the unique decoding radius of $\mathcal{P}_{d}(\Boo^{n}, G)$). We give an overview of the algorithm $\mathcal{A}_{1}^{f}$ below:
    \begin{itemize}
        \item Choose a random subcube $\mathsf{C}$ of dimension $k$ (which is a polynomial in the list size and from the previous section we know that the list size is a constant dependent on $\varepsilon$). Then find all the degree $d$ polynomials that are $(1/2^{d} - \varepsilon/2)$-close to $f$ on the subcube $\mathsf{C}$. Repeat this step a few times. Let's call this set of polynomials (union over all the repetitions) as $T$. Since $\mathsf{C}$ is a random subcube, with high probability, restriction of every polynomial in $\mathsf{List}_{\varepsilon}^{f}$ will be in the set $T$. The restriction of each $P \in \mathsf{List}_{\varepsilon}^{f}$ to $\mathsf{C}$ (which is in $T$) will be \emph{advice} for $P$ in the second step. 

        \item The algorithm $\mathcal{A}_{1}^{f}$ also samples a random permutation on $\sigma$ on $2k$ variables. The oracles in the next step will use this permutation while locally list correcting (it is the same permutation for every oracle). For each polynomial $Q$ in $T$, the algorithm $\mathcal{A}_{1}$ produces an oracle $\Psi[\mathsf{C}, \sigma, Q]$ with the polynomial's evaluation on the cube $\mathsf{C}$ as advice.
    \end{itemize}

    \item \textbf{Second step (approximation):} Suppose we want to locally list correct $f$ at a point $\mathbf{b} \in \Boo^{n}$. Each algorithm $\Psi[\mathsf{C}, \sigma, Q]$ creates a subcube $\mathsf{C}'$ of dimension $2k$ that is spanned by $\mathsf{C}$ and $\mathbf{b}$. Then the oracle computes all degree $d$ polynomials that have distance at most $(1/2^{d} - \varepsilon/2)$ from $f$ on the subcube $\mathsf{C}'$ and uses the advice on the subcube $\mathsf{C}$ to filter out $P(\mathbf{b})$, where $P|_{\mathsf{C}} = Q$.
\end{enumerate}

\subsubsection{Formal description of the algorithms}
In this subsection, we describe the algorithms that will prove \Cref{thm:approx-oracles-list-decoding}.
The algorithms in this subsection are the same as in \cite[Section 5.2]{ABPSS24-ECCC} with minor modifications in the parameters to make it work for the class $\mathcal{P}_{d}$. We state the algorithms here for the sake of completeness.

\paragraph{Notation:} Let $\mathsf{C}$ be a $k$-dimensional subcube of $\{0,1\}^n$ as defined in \Cref{defn:random-embedding}. Let $Q: \{0,1\}^k \to G$  a polynomial in $\mathcal{P}_{d}(\Boo^{k}, \, G)$ where the polynomial $Q$ is a function on the subcube $\mathsf{C}$. We will use $\Psi[\mathsf{C},\sigma, Q]$ to denote a \emph{deterministic} algorithm that has the description of the subcube $\mathsf{C}$, a permutation $\sigma:[2k]\rightarrow [2k],$ and evaluation of $Q$ on $\mathsf{C}$ hardwired inside it\footnote{In the final algorithm, $\mathsf{C}$ will be a random subcube of dimension $\bigO_\varepsilon(1)$ and $Q$ with high probability be equal to $P|_{\mathsf{C}}$, for some $P \in \mathsf{List}_{\varepsilon}^{f}$}.

\Cref{algo:list-decoding} is a randomized algorithm that outputs the descriptions of the deterministic oracles and \Cref{algo:high-agreement} describes the oracles themselves. 

\begin{algobox}
\begin{algorithm}[H]
\caption{Approximating Algorithm $\Psi[\mathsf{C},\sigma, Q]$}
\label{algo:high-agreement}
\DontPrintSemicolon

\KwIn{Oracle access to the function $f$, a point $\mathbf{b} \in \Boo^{n}$}
\vspace{3mm}

Let $\mathsf{C}'$ be a subcube spanned by $\mathsf{C}$ and $\mathbf{b}$ using $\sigma$  \tcp*{see \Cref{def:bigger-subcube}}
Let $\mathbf{w} \in \Boo^{2k}$ such that $x(\mathbf{w}) \in \mathsf{C}'$ and $x(\mathbf{w}) = \mathbf{b}$ \tcp*{$|\mathbf{w}| = k$}
\vspace{2mm}
Query $f$ on the subcube $\mathsf{C}'$ \tcp*{Number of queries is $2^{2k}$}
\vspace{2mm}
Find all polynomials $R_{1},\ldots, R_{L''} \in \mathcal{P}_{d}(\Boo^{2k}, \, G)$ that are $\paren{\frac{1}{2^{d}} - \frac{\varepsilon}{2}}$-close to $f|_{\mathsf{C'}}$ \tcp*{using \Cref{thm:non-local-list}} \tcp*{$L'' \leq L(\varepsilon/2)$} \label{line:brute-force}
\vspace{2mm}
\If{there exists an $i \in [L'']$ such that $R_{i}|_{\mathsf{C}} = Q$}{
pick any such $i$ and \Return{$R_{i}(\mathbf{w})$}
}
\Else{
\Return{$0$}\tcp*{An arbitrary value}
}

\end{algorithm}
\end{algobox}

Now we describe the randomized \Cref{algo:list-decoding} that returns the descriptions of the deterministic oracles.

\begin{algobox}
\begin{algorithm}[H]
\caption{Algorithm $\mathcal{A}_{1}$}
\label{algo:list-decoding}

\DontPrintSemicolon

\KwIn{Oracle access to the function $f$}
\vspace{3mm}

Choose $k \leftarrow B_d\cdot (L(\varepsilon/2)/\varepsilon)^c$  \tcp*{$B_d$ and $c$ chosen below}
Set $\ell \leftarrow \log L(\varepsilon)$\;
$T \leftarrow \emptyset$\;
\Repeat{$\ell$ times}{
Sample $\mathbf{a} \sim U_{n}$ and a random hash function $h: [n] \to [k]$ \tcp*{the first source of randomness}
Construct the subcube $\mathsf{C} := C_{\mathbf{a}, h}$ \tcp*{see \Cref{defn:random-embedding}}
\vspace{2mm}
Query $f$ on the subcube $\mathsf{C}$ \tcp*{Number of queries is $2^{k}$}
\vspace{2mm}
Find all polynomials $Q_{1},\ldots,Q_{L'} \in \mathcal{P}_{d}(\Boo^{k}, \, G)$ that are $\paren{\frac{1}{2^{d}} - \frac{\varepsilon}{2}}$-close to $f|_{\mathsf{C}}$ \tcp*{using \Cref{thm:non-local-list}} \label{line:brute-force-2}
Pick a uniformly random permutation $\sigma:[2k]\rightarrow [2k]$ \tcp*{the second source of randomness}
$T \leftarrow T \cup \set{(\mathsf{C},\sigma,Q_{1}),\ldots,(\mathsf{C},\sigma,Q_{L'})}$ \tcp*{$L' \leq L(\varepsilon/2)$}
}
\vspace{2mm}
\Return{$\Psi[\mathsf{C},\sigma, Q]$ for all $(\mathsf{C},\sigma, Q) \in T$} \tcp*{Size of $T$ is $\leq \ell L'$}

\end{algorithm}
\end{algobox}

\subsection{Analysis of the algorithms}
In this subsection, we analyze \Cref{algo:high-agreement} and \Cref{algo:list-decoding}.

\paragraph{Query complexity:}\Cref{algo:list-decoding} makes $2^{k} = \exp(B_d \cdot \poly(L(\varepsilon/2))$ queries to $f$ and returns the description of $\ell L' = L(\varepsilon/2) \log L(\varepsilon)$ oracles. Each oracle (see \Cref{algo:high-agreement}) makes $2^{2k} =\exp(B_d \cdot \poly(L(\varepsilon/2))$ queries to $f$. Hence the total number of queries to $f$ is $\exp(\bigO_d(\poly(L(\varepsilon/2)))$.

\paragraph{Correctness:}We want to show that with probability $\geq 3/4$, for every polynomial $P \in \mathsf{List}_{\varepsilon}^{f}$, there exists an output oracle $\Psi[\mathsf{C},\sigma, Q]$ that is $(1/(10 \cdot 2^{d+1}))$-close to $P$. We prove this in the following steps.
\begin{enumerate}[$\blacktriangleright$]
    \item In a single iteration for \Cref{algo:list-decoding}, the following holds: For every polynomial $P \in \mathsf{List}_{\varepsilon}^{f}$, with probability at least $9/10$, there exists a $1/(10 \cdot 2^{d+1})$-close approximating oracle $\Psi[\mathsf{C}, \sigma, Q_{j}]$.
    We prove this is in \Cref{lemma:list-correction-error}.
    \item As we have $\ell$ independent iterations, the probability that there is no $1/(10 \cdot 2^{d+1})$-close approximating oracle for $P$ is at most $1/10^{\ell}$.  By a union bound for all polynomials $P \in \mathsf{List}_{\varepsilon}^{f}$, we get the desired correctness probability in \Cref{thm:approx-oracles-list-decoding}.
    \item Since each iteration produces a list of size at most $L(\varepsilon/2),$ overall we obtain a list of size $\bigO(L(\varepsilon/2)\cdot \log L(\varepsilon))$ as claimed.
\end{enumerate}
We start by proving \Cref{lemma:list-correction-error}, which is the primary lemma for the correctness of our local list correctors.\\

\noindent
\begin{thmbox}
\begin{lemma}[Correctness of Local List Correction]\label{lemma:list-correction-error}
The following holds as long as the constants $B_d$ (depending on $d$) and $c$ in \Cref{algo:list-decoding} are chosen to be large enough. Fix any polynomial $P \in \mathsf{List}_{\varepsilon}^{f}$.  In each iteration of the loop in \Cref{algo:list-decoding}, the probability (over the randomness of the algorithm) that there does not exist a $1 \leq j \leq L'$ such that $\delta(\Psi[\mathsf{C},\sigma, Q_{j}], \, P) \leq 1/(10 \cdot 2^{d+1})$ is at most $1/10$.
\end{lemma}
\end{thmbox}
\begin{proof}
Fix an iteration of the loop in \Cref{algo:list-decoding}. Let $\mathcal{E}_{P}$ denote the event that there does not exist a $j$ such that $\delta(\Psi[\mathsf{C},\sigma, Q_{j}], \, P) \leq 1/(10 \cdot 2^{d+1})$. We want to upper bound the probability of the ``bad'' event $\mathcal{E}_P$ by $1/10$. Recall that the sources of randomness in \Cref{algo:list-decoding} are point $\mathbf{a}$, hash function $h$, and permutation $\sigma$. We will show that 
\begin{equation}
\label{eq:EP}
    \E_{\mathbf{a},h,\sigma}[\min_{j} \, \delta(\Psi[\mathsf{C},\sigma, Q_{j}], \, P)] \; = \; \E_{\mathbf{a},h,\sigma}[\min_{j} \, \Pr_{\mathbf{b}}[\Psi[\mathsf{C},\sigma, Q_{j}](\mathbf{b}) \, \neq \, P(\mathbf{b})]] \; \leq \; \dfrac{1}{100 \cdot 2^{d+1}},
\end{equation}
from which \Cref{lemma:list-correction-error} follows via an application of Markov's inequality.

Define the following auxiliary events, depending on the choice of $\mathbf{a},h$ and $\sigma$, along with the choice of a random input $\mathbf{b}.$
\begin{enumerate}
    \item Event $\mathcal{E}_{1,P}$ (only depends on $\mathbf{a},h$): In the current iteration of the loop in \Cref{algo:list-decoding}, the algorithm does not find a polynomial $Q_{j}$ such that $Q_{j} = P|_{\mathsf{C}}$.

    \item Event $\mathcal{E}_{2,P}$ (depends on $\mathbf{a}, h, \sigma, \mathbf{b}$): For the triples $(\mathsf{C},\sigma,Q)$ added to $T$ in this iteration of \Cref{algo:list-decoding}, the corresponding oracle $\Psi[\mathsf{C},\sigma,Q]$ is such that when we run this oracle on input $\mathbf{b}$, there does not exist a polynomial $R_{i}$ such that $R_{i} = P|_{\mathsf{C'}}$. Note that this event only depends on the choice of the cube $\mathsf{C}, \sigma$ and $\mathbf{b}$ but not on the specific polynomial $Q$. Hence, the event is exactly the same for each triple $(\mathsf{C},\sigma,Q)$ in $T.$ In particular, we may fix any one such triple.

    \item Event $\mathcal{E}_{3,P}$ (depends on $\mathbf{a}, h, \sigma, \mathbf{b}$): For the triples $(\mathsf{C},\sigma,Q)$ added to $T$ in this iteration of \Cref{algo:list-decoding}, the corresponding oracle $\Psi[\mathsf{C},\sigma,Q]$ is such that when we run this oracle on input $\mathbf{b}$, there exist two polynomials $R_{i_1}$ and $R_{i_2}$ such that  $R_{i_1}(\mathbf{w})\neq R_{i_2}(\mathbf{w})$ but $R_{i_1}|_\mathsf{C} = R_{i_2}|_\mathsf{C}$. Here $\mathbf{w}$ is, as defined in \Cref{algo:high-agreement}, the point in $\{0,1\}^{2k}$ of Hamming weight $k$ such that $x(\mathbf{w}) \in \mathsf{C}'$ and $x(\mathbf{w}) = \mathbf{b}.$ As for the event $\mathcal{E}_{2,P}$, we may fix a triple $(\mathsf{C},\sigma,Q)\in T$ while analyzing this event.
\end{enumerate}



\snote{Commented out one line here. }

We will need the following two claims that show that any of the aforementioned events occur with a small probability.
\begin{claim}\label{claim:l-restriction-far}
 $\Pr_{\mathbf{a},h}[\mathcal{E}_{1,P}], \Pr_{\mathbf{a},h,\sigma,\mathbf{b}}[\mathcal{E}_{2,P}] \leq 1/(10000\cdot 2^{d+1}).$
\end{claim}

\begin{claim}\label{claim:E3P}
$\Pr_{\mathbf{a},h,\sigma, \mathbf{b}}[\mathcal{E}_{3,P}]\leq 1/(10000\cdot 2^{d+1}).$
\end{claim}

Let us proceed with the proof of \Cref{lemma:list-correction-error} assuming the above two claims. We first show that if $\mathcal{E}_{P}$ occurs, then at least one of the auxiliary events occurs. This implies that it is sufficient to upper bound the probability of the auxiliary events occurring and use a union bound.

For $\mathbf{a}, h$ such that the event $\mathcal{E}_{1,P}$ does not occur, we can fix a $j^* \in [L']$ such that $P|_{\mathsf{C}} = Q_{j^*}.$  Thus, we have
\begin{equation}
\label{eq:E1P}
\E_{\mathbf{a},h,\sigma}[\min_{j} \, \Pr_{\mathbf{b}}[\Psi[\mathsf{C},\sigma,Q_j](\mathbf{b}) \, \neq \, P(\mathbf{b})]] \; \leq \; \Pr_{\mathbf{a},h}[\mathcal{E}_{1,P}] \, + \, \E_{\mathbf{a},h,\sigma}[\mathbf{1}_{\neg \mathcal{E}_{1,P}}\cdot \Pr_{\mathbf{b}}[\Psi[\mathsf{C},\sigma,Q_{j^*}](\mathbf{b}) \, \neq \, P(\mathbf{b}) ]]
\end{equation}
Fix any $\mathbf{a}, h$ such that the event $\mathcal{E}_{1,P}$ does not occur. Further, if the event $\mathcal{E}_{2,P}$ does not occur, then there is an $i^*\leq L''$ such that $P|_{\mathsf{C}'} = R_{i^*}.$ In particular, $R_{i^*}|_\mathsf{C} = P|_{\mathsf{C}} = Q_{j^*}.$ 

Finally, if event $\mathcal{E}_{3,P}$ also does not occur, then there is no $i\neq i^*$ such that $R_{i^*}(\mathbf{w}) \neq R_i(\mathbf{w})$ but $R_{i}|_{\mathsf{C}} = R_{i^*}|_\mathsf{C}.$ In particular, the only possible output of the algorithm $\Psi[\mathsf{C},\sigma,Q_{j^*}]$ on input $\mathbf{w}$ is $R_{i^*}(\mathbf{w}) = P(x(\mathbf{w})) = P(\mathbf{b}).$

We have thus shown that 
\[
\E_{\mathbf{a},h,\sigma}[\mathbf{1}_{\neg \mathcal{E}_{1,P}}\cdot \Pr_{\mathbf{b}}[\Psi[\mathsf{C},\sigma,Q_{j^*}](\mathbf{b})\neq P(\mathbf{b}) ]] \leq 
\Pr_{\mathbf{a},h,\sigma,\mathbf{b}}[\mathcal{E}_{2,P}\vee \mathcal{E}_{3,P}] \leq \Pr_{\mathbf{a},h,\sigma,\mathbf{b}}[\mathcal{E}_{2,P}] + \Pr_{\mathbf{a},h,\sigma,\mathbf{b}}[\mathcal{E}_{3,P}].
\]
Plugging the above into \Cref{eq:E1P}, we get
\begin{equation}
    \label{eq:E1P2P3P}
    \E_{\mathbf{a},h,\sigma}[\min_{j}\Pr_{\mathbf{b}}[\Psi[\mathsf{C},\sigma,Q_j](\mathbf{b})\neq P(\mathbf{b})] \; \leq \; \Pr_{\mathbf{a},h}[\mathcal{E}_{1,P}] \, + \, \Pr_{\mathbf{a},h,\sigma,\mathbf{b}}[\mathcal{E}_{2,P}] \, + \, \Pr_{\mathbf{a},h,\sigma,\mathbf{b}}[\mathcal{E}_{3,P}].
\end{equation}
Using \Cref{claim:l-restriction-far} and \Cref{claim:E3P}, we get,
\begin{align*}
    \E_{\mathbf{a},h,\sigma}[\min_{j}\Pr_{\mathbf{b}}[\Psi[\mathsf{C},\sigma,Q_j](\mathbf{b})\neq P(\mathbf{b})] \; \leq \; \dfrac{3}{10000\cdot 2^{d+1}} \, \leq \, \dfrac{1}{100 \cdot 2^{d+1}}
\end{align*}
So now it remains to prove \Cref{claim:l-restriction-far} and \Cref{claim:E3P}. We start with \Cref{claim:l-restriction-far} which essentially follows from \Cref{lemma:sampling-subcube}.
\begin{adjustwidth}{1.5cm}{0cm}
\begin{proof}[Proof of \Cref{claim:l-restriction-far}]

Recall that $\delta(P,f) \leq (1/2^{d} - \varepsilon).$ Equivalently, the set of points $T$ where $f$ and $P$ differ has density at most $(1/2^{d} -\varepsilon)$ in $\{0,1\}^n.$ For a cube $\mathsf{C}$, the non-existence of $Q_{j}$ such that $Q_{j} = P|_{\mathsf{C}}$ is equivalent to $\delta(P|_{\mathsf{C}}, f|_{\mathsf{C}}) > (1/2^{d} - \varepsilon/2)$.

For a random $\mathbf{a}$ and a random $h$, the subcube $\mathsf{C}$ is a random subcube. Using \Cref{lemma:sampling-subcube} for the subset $T$ as mentioned above, we get that for $k\geq 1/\varepsilon^5$,
\begin{align*}
    \Pr_{\mathsf{C}}\brac{\delta(P|_{\mathsf{C}}, f|_{\mathsf{C}}) > \frac{1}{2^{d}}-\frac{\varepsilon}{2}} \leq \dfrac{1}{10000\cdot 2^{d+1}}.
\end{align*}
Here, we are assuming that $B_d$ and $c$ are large enough so that $k$ as defined in \Cref{algo:list-decoding} satisfies the hypothesis of \Cref{lemma:sampling-subcube}. Hence $\Pr[\mathcal{E}_{1,P}] \leq 1/(10000\cdot 2^{d+1})$.

Using \Cref{obs:Cb-is-random}, we know that for a random $\mathbf{a}, h$ and a random permutation $\sigma$, the subcube $\mathsf{C'}$ is a random subcube of dimension $2k$ in $\{0,1\}^n$ as required in \Cref{lemma:sampling-subcube}. Proceeding as above, we get the stated upper bound on $\Pr[\mathcal{E}_{2,P}]$. 
\end{proof}
\end{adjustwidth}

Now it remains to prove \Cref{claim:E3P}, which we prove next.

\begin{adjustwidth}{1.5cm}{0cm}
\begin{proof}[Proof of \Cref{claim:E3P}]
We start by conditioning the choice of the subcube $\mathsf{C}'.$ This fixes the polynomials $R_{1},\ldots,R_{L''}$ obtained in Line 4 of \Cref{algo:high-agreement}. Fix any two distinct polynomials $R_{i_{1}}, R_{i_{2}}: \Boo^{2k} \to G$ in this list that differ at \emph{at least one point} in $\{0,1\}^{2k}$ of Hamming weight $k$ (in particular, this includes pairs of polynomials that differ at the point $\mathbf{w}$ such that $x(\mathbf{w}) = \mathbf{b}$).

We want to upper bound the probability of the event that $R_{i_{1}}|_{\mathsf{C}} = R_{i_{2}}|_{\mathsf{C}}$. In the end, we use a union bound over the number of all possible pairs $(R_{i_{1}}, R_{i_{2}})$.

Define the polynomial $R := R_{i_{1}} - R_{i_{2}}$, where $R: \Boo^{2k} \to G$ is a non-zero degree 
$d$ polynomial, defined on the subcube $\mathsf{C}'$. We want to upper bound the probability that $R|_{\mathsf{C}}$ is identically zero polynomial. Using \Cref{obs:C-is-random} and \Cref{cor:samplingDLSZ}, we see that the probability of this is at most $\bigO_d(1/k^{\eta})$ for some absolute constant $\eta > 0.$ For $k$ as defined in \Cref{algo:list-decoding} where $B_d$ is large enough (depending on $d$) and $c$ is a large enough absolute constant, we get a probability of at most $1/(10000\cdot 2^{d+1}\cdot L(\varepsilon/2)^2).$ 

Since $L''\leq L(\varepsilon/2),$ a union bound over all such pairs $R_{i_1}, R_{i_2}$ yields the bound stated in the claim.
\end{proof}
\end{adjustwidth}

As discussed above, we have proved \Cref{claim:l-restriction-far} and \Cref{claim:E3P} and substituting them in \Cref{eq:E1P2P3P}, we get the desired bound, and this concludes the correctness of the local list correction algorithm.
\end{proof}

\subsection{Local list corrector}

Let us now see how \Cref{thm:approx-oracles-list-decoding} implies \Cref{thm:listdecoding}. Let us first recall \Cref{thm:listdecoding}.

\listdecoding*

\begin{proof}[Proof of \Cref{thm:listdecoding}]
We first run the algorithm given by \Cref{thm:approx-oracles-list-decoding} and it outputs $\psi_{1},\ldots, \psi_{L'}$ for $L'\leq \exp(\bigO_d(1/\varepsilon)^{\bigO(d)})$ (by \Cref{thm:comblistdegd}). Next we run our local correction algorithm for $\mathcal{P}_{d}$ (see \Cref{thm:uniquedegd} and \Cref{sec:deg-1-decoding}) with $\psi_{1}, \ldots, \psi_{L'}$ as oracles, and these algorithms will be $\phi_{1}, \ldots, \phi_{L'}$. This completes the description of the local list correction algorithm $\mathcal{A}^{f}$ for $\mathcal{P}_{d}$, and the bound on correctness probability follows from the correctness probability of \Cref{thm:uniquedegd} and \Cref{thm:approx-oracles-list-decoding}.

The algorithm $\mathcal{A}_{1}$ makes $\bigO_{\varepsilon}(1)$ queries to $f$ as stated in \Cref{thm:approx-oracles-list-decoding}, and then each $\phi_{i}$ makes $\bigO_{\varepsilon}(1) \cdot \Tilde{\bigO}(\log n) = \Tilde{\bigO}_\varepsilon(\log n)$ queries to $f$.
\end{proof}

\bibliographystyle{alpha}
\bibliography{references}




\addtocontents{toc}{\protect\setcounter{tocdepth}{1}}
	
\end{document}